\documentclass[letter,11pt]{article}
\usepackage[margin=1in]{geometry} 
\usepackage{amsfonts,amssymb,amsmath,amsthm,mathtools}
\usepackage{graphicx,float,afterpage} 
\usepackage[dvipsnames]{xcolor}
\usepackage{enumerate,color,caption,boxedminipage,soul}
\usepackage{setspace}
\usepackage{longtable}
\usepackage[normalem]{ulem} 
\usepackage{datetime}
\usepackage{changepage}
\usepackage[round]{natbib}

\usepackage{hyperref}
\hypersetup{colorlinks=true,linkcolor=BrickRed,citecolor=OliveGreen}

\usepackage{chngcntr,apptools,titlesec}
\titleformat{\section}{\normalfont\centering\fontsize{11}{13}\bfseries}{\thesection}{1em}{}
\titleformat{\subsection}{\normalfont\normalsize\bfseries}{\thesubsection}{1em}{}
\titleformat{\subsubsection}{\normalfont\normalsize\bfseries}{\thesubsubsection}{1em}{}
\renewcommand{\thesection}{\arabic{section}}
\renewcommand{\thesubsection}{\arabic{section}.\arabic{subsection}}

\usepackage[ruled,linesnumbered,algoruled,boxed]{algorithm2e}
\usepackage{array,multirow,ctable}

\newtheorem{lemma}{Lemma}
\newtheorem{theorem}{Theorem}
\newtheorem{proposition}{Proposition}

\theoremstyle{definition}
\newtheorem{definition}{Definition}
\theoremstyle{remark}
\newtheorem{remark}{Remark}

\DeclareMathOperator*{\argmin}{arg~min~}

\newcommand{\vertiii}[1]{{\left\vert\kern-0.25ex\left\vert\kern-0.25ex\left\vert #1 
		\right\vert\kern-0.25ex\right\vert\kern-0.25ex\right\vert}}
\newcommand{\vertii}[1]{{\left\vert\kern-0.25ex\left\vert #1 
		\right\vert\kern-0.25ex\right\vert}}
\newcommand{\smalliii}[1]{{\vert\kern-0.25ex\vert\kern-0.25ex\vert #1 \vert\kern-0.25ex\vert\kern-0.25ex\vert}}
\newcommand{\smallii}[1]{{\vert\kern-0.25ex\vert #1\vert\kern-0.25ex\vert}}

\newcommand{\VAR}{\mathrm{VAR}}
\newcommand{\R}{\mathbb{R}}
\newcommand{\I}{\mathrm{I}}

\newcommand{\SEN}{\text{SEN}}
\newcommand{\SPC}{\text{SPC}}
\newcommand{\PC}{\scriptsize{\text{PC}}}

\newcommand{\op}{\text{op}}
\newcommand{\F}{\text{F}}
\newcommand{\PP}{\mathbb{P}}

\newcommand{\Err}{\mathrm{Err}}
\newcommand{\RSC}{\text{RSC}}

\makeatletter
\newsavebox{\@brx}
\newcommand{\llangle}[1][]{\savebox{\@brx}{\(\m@th{#1\langle}\)}%
	\mathopen{\copy\@brx\mkern2mu\kern-0.9\wd\@brx\usebox{\@brx}}}
\newcommand{\rrangle}[1][]{\savebox{\@brx}{\(\m@th{#1\rangle}\)}%
	\mathclose{\copy\@brx\mkern2mu\kern-0.9\wd\@brx\usebox{\@brx}}}
\makeatother

\title{\large\bf Regularized Estimation of High-dimensional Factor-Augmented\\Vector Autoregressive (FAVAR) Models}
\author{\normalsize Jiahe Lin\footnote{Department of Statistics, University of Michigan. \texttt{jiahelin@umich.edu}}\quad and \quad George Michailidis\thanks{\textbf{Corresponding Author}. Department of Statistics and the Informatics Institute, University of Florida. \texttt{gmichail@ufl.edu}}}
\date{\normalsize}

\begin{document}
	\maketitle
	\begin{abstract}
		A factor-augmented vector autoregressive (FAVAR) model is defined by a VAR equation that captures lead-lag correlations amongst a set of observed variables $X$ and latent factors $F$, and a calibration equation that relates another set of observed
		variables $Y$ with $F$ and $X$. The latter equation is used to estimate the factors that are subsequently used in estimating
		the parameters of the VAR system. The FAVAR model has become popular in applied economic research, since it can summarize a large number of variables of interest as a few factors through the calibration equation and subsequently examine their influence on core variables of primary interest through the VAR equation. However, there is increasing need for examining lead-lag relationships between a large number of time series, while incorporating information from another high-dimensional
		set of variables. Hence, in this paper we investigate the FAVAR model under high-dimensional scaling. We introduce an 
		appropriate identification constraint for the model parameters, which when incorporated into the formulated optimization problem yields estimates with good statistical properties. Further, we address a number of technical challenges introduced by the fact that estimates of the VAR system model parameters are based on estimated rather than directly observed quantities. The performance of the proposed estimators is evaluated on synthetic data. Further, the model is applied to commodity prices and reveals interesting and interpretable relationships between the prices and the factors extracted from a set of global macroeconomic indicators.
	\end{abstract}
	{\bf Source code:} \url{https://github.com/jhlinplus/High_dim_FAVAR_estimation}
	
	
	\section{Introduction}\label{sec:intro}
	
	There is a growing need in employing a large set of time series (variables) for modeling social or physical systems. For example, economic policy makers have concluded based on extensive empirical evidence \citep[e.g.][]{sims1980macroeconomics,bernanke2005measuring,banbura2010large} that large scale models of economic indicators
	provide improved forecasts, together with better estimates of how current economic shocks propagate into the future,
	which produces better guidance for policy actions. Another reason for considering large number of time series in social sciences is that key variables implied by theoretical models for policy decisions\footnote{such as the concept of output gap for monetary policy, the latter defined as the difference between the actual output of an economy and its potential output} are not directly observable, but related to a large number of other variables that collectively act as a good proxy of the 
	unobservable key variables. In other domains such as genomics and neuroscience, advent of high throughput technologies have enabled researchers to obtain measurements on hundreds of genes from functional pathways of interest 
	\citep{shojaie2010discovering} or brain regions \citep{seth2015granger}, thus allowing a more comprehensive modeling to gain insights into biological mechanisms of interest. There are two popular modeling paradigms for such large panel of time series, with the first being the Vector Autoregressive (VAR) model \citep{lutkepohl2005new} and the second being the Dynamic Factor Model (DFM) \citep{stock2002forecasting,lutkepohl2014structural}.
	
	The VAR model has been the subject of extensive theoretical and empirical work primarily in econometrics, due to its relevance in macroeconomic and financial modeling. However, the number of model parameters increases quadratically with the number of time series included for each lag period considered, and this feature has limited its applicability since in many applications it is hard to obtain adequate number of time points for accurate estimation. Nevertheless, there is a recent body of technical work that leveraging {\em structured sparsity} and the corresponding regularized estimation framework has established results for consistent estimation of the VAR parameters under high dimensional scaling. \citet{basu2015estimation} examined Lasso penalized Gaussian VAR models and proved consistency results, while at the same time providing technical tools useful for analysis of sparse models involving temporally dependent data. \citet{melnyk2016estimating} extended the results to other regularizers, \citet{lin2017regularized} to the inclusion of exogenous variables (the so-called VAR-X model in the econometrics literature), \citet{hall2016inference} to models for count data and \citet{nicholson2017varx} to the simultaneous estimation of time lags and model parameters. However, a key requirement for the theoretical developments is a spectral radius constraint that ensures the {\em stability} of the underlying VAR process \citep[see][for details]{basu2015estimation, lin2017regularized}. For large VAR models,
	this constraint implies a smaller magnitude on average for all model parameters, which makes their estimation more challenging, unless one compensates with a higher level of sparsity. Nevertheless, very sparse VAR models may not be adequately informative, while their estimation requires larger penalties that in turn induce higher bias due to shrinkage, when the sample size stays fixed. 
	
	The DFM model aims to decompose a large number of time series into a few common latent factors and idiosyncratic components. The premise is that these common factors are the key drivers of the observed data, which themselves can exhibit temporal dynamics. They have been extensively used for forecasting purposes in economics \citep{stock2002forecasting}, while their statistical properties have been studied in depth \citep[see][and references therein]{bai2008large}. Despite their 
	ability to handle very large number of time series, theoretically appealing properties and extensive use in empirical work in economics, DFMs aggregate the underlying time series and hence are not suitable for examining their individual cross-dependencies. Since in many applications researchers are primarily interested
	in understanding the interactions between key variables \citep{sims1980macroeconomics,stock2016dynamic}, while accounting for the influence of many others so as to avoid model misspecification that leads to biased results, DFMs may not be the most appropriate model.
	
	To that end, \citet{bernanke2005measuring} proposed a ``fusion" model, namely the Factor Augmented VAR, that aims to summarize the information contained in a large set of time series
	by a small number of factors and includes those in a standard VAR model. Specifically, let $\{F_t\}\in\mathbb{R}^{p_1}$ be the latent factor and $\{X_t\}\in\mathbb{R}^{p_2}$ the observed sets of variables, they jointly form a VAR system given by
	\begin{equation}\label{md:FAVAR-VAR}
	\begin{bmatrix}
	F_t \\ X_t
	\end{bmatrix} = A^{(1)} \begin{bmatrix}
	F_{t-1} \\ X_{t-1}
	\end{bmatrix} + \cdots + A^{(d)} \begin{bmatrix}
	F_{t-d} \\ X_{t-d}
	\end{bmatrix} + \begin{bmatrix}
	w^F_t \\ w^X_t
	\end{bmatrix}.
	\end{equation}
	In addition, there is a large panel of observed time series $Y_t\in\mathbb{R}^q$, whose current values are influenced by both $X_t$ and $F_t$; i.e., the calibration equation:
	\begin{equation}\label{md:FAVAR-info}
	Y_t = \Lambda F_t + \Gamma X_{t} + e_t.
	\end{equation} 
	The primary variables of interest $X_t$ together with the unobserved factors $F_t$---both are assumed to have small and fixed dimensions---drives the dynamics of the system, and the factors are inferred from~\eqref{md:FAVAR-info}. 
	
	Even in the low-dimensional setting ($p_2$ fixed), there is very limited theoretical work \citep{bai2016estimation} on the FAVAR model and some work on identification restrictions for the model parameters \citep[e.g.][]{bernanke2005measuring}. However, the fixed dimensionality assumption is rather restrictive in many applications; in particular, 
	the model has been extensively used in empirical work in economics and finance \citep[e.g.][]{eickmeier2014understanding,caggiano2014uncertainty}, yet customarily a very small size block $X_t$ is considered.
	For example, in \citet{bernanke2005measuring} that introduces the FAVAR model, $X_t$ comprises of three ``core" economic indicators (industrial production, consumer
	price index and the federal funds rate) and $Y_t$ of 120 other economic indicators. The VAR system is augmented by one factor summarizing the macroeconomic indicators, and the augmented system shows 7-lag time dependence that significantly increases the sample size requirement for estimation purposes.
	In a recent application, \citet{stock2016dynamic} apply the FAVAR model to macroeconomics effects of oil supply shocks;
	the augmented VAR system consists of 8 times series (observed and latent), but due to the limitation in sample size to avoid non-stationarities ($T=120$) the lag of the model is fixed to 1. Hence, as argued in \citet{stock2016dynamic}, there is growing
	need for large scale FAVAR models and this paper aims to examine their estimation and theoretical properties in high-dimensions, leveraging sparsity constraints on key model parameters. 
	
	The key contributions of this paper are twofold: (1) the introduction of an identifiability constraint compatible with the high-dimensional nature of the model, under sparsity assumptions on model parameters $\Gamma$ and $\{A^{(k)}\}$, and (2) the ensuing formulation of the optimization problem that leads to their estimators based on observational data and estimators' high-probability error bounds. At the technical level there are two sets of challenges that are successfully resolved: (i) the calibration equation involves both an observed set of covariates and a set of latent factors, and their interactions require careful handling to enable accurate estimation of the factors that constitute part of the input to the augmented VAR system and are crucial for estimating the transition matrix; and (ii) with the presence of a block of variables in the VAR system that are subject to error due to being estimated rather than directly observed, a number of new technical challenges emerge and they are compounded by the presence of temporal dependence. Note that for ease of presentation, the main technical developments are shown for Gaussian data (all noise processes in~\eqref{md:FAVAR-VAR} and~\eqref{md:FAVAR-info}
	are assumed to be Gaussian), but the key theoretical results are also established for sub-Gaussian and sub-exponential
	error processes; see Appendix C for a result of independent theoretical interest, even for the standard
	sparse VAR model.
	
	\paragraph{Outline of the paper.} The remainder of the paper is organized as follows. In Section~\ref{sec:id-formulation}, the model identifiability constraint is introduced, followed by formulation of the objective function to be optimized that obtains estimates of the model parameters. Theoretical properties of the proposed estimators, specifically, their high probability finite-sample error bounds, are investigated in Section~\ref{sec:theory}. Subsequently in Section~\ref{sec:simulation}, we introduce an empirical implementation procedure for obtaining the estimates and present its performance evaluation based on synthetic data. An application of the model on interlinkages of commodity prices and the influence of world macroeconomic indicators on them is presented in Section~\ref{sec:realdata}, while Section~\ref{sec:discussion} provides some concluding remarks. All proofs and other supplementary materials are deferred to Appendices.
	
	\paragraph{Notations.} Throughout this paper, we use $\smalliii{A}_{\cdot}$ to denote matrix norms for some generic matrix $A\in\R^{m\times n}$. For example, $\smalliii{A}_1$ and $\smalliii{A}_\infty$ respectively denote the matrix induced $1$-norm and infinity norm, $\smalliii{A}_{\op}$ the matrix operator norm and $\smalliii{A}_{\F}$ the Frobenius norm. Moreover, We use $\|A\|_1$ and $\|A\|_\infty$ respectively to denote the element-wise $1$-norm and infinity norm. For two matrices $A$ and $B$ of commensurate dimensions, denote their inner product by $\llangle A, B\rrangle = \text{tr}(A^\top B)$. Finally, we write $A\gtrsim B$ if there exists some absolute constant $c$ that is independent of the model parameters such that $A\geq cB$; and $A\asymp B$ if $A\gtrsim B$ and $B\gtrsim A$ hold simultaneously.

	\section{Model Identification and Problem Formulation}\label{sec:id-formulation}
	
	The FAVAR model proposed in \citet{bernanke2005measuring} has the following two components, as seen in Section
	\ref{sec:intro}: a system given in ~\eqref{md:FAVAR-VAR} that describes the dynamics of the latent block~$F_t\in\R^{p_1}$ and the observed block~$X_t\in\R^{p_2}$ that jointly follow a stationary $\VAR(d)$ model (the ``VAR equation"); and the model in~\eqref{md:FAVAR-info} that characterizes the contemporaneous dependence of the large observed informational series $Y_t\in\R^q$ as a linear function of $X_t$ and $F_t$ (the ``calibration equation"). Further, $w^F_t$, $w^X_t$ and $e_t$ are all noise terms that are independent of the predictors, and we assume they are serially uncorrelated mean-zero Gaussian random vectors: $w^F_t\sim \mathcal{N}(0,\Sigma_w^F)$, $w^X_t\sim \mathcal{N}(0,\Sigma_w^X)$ and $e_t\sim \mathcal{N}(0,\Sigma_e)$. In this study we consider a potentially large VAR system that has many coordinates, hence in contrast to \citet{bernanke2005measuring} and \citet{bai2016estimation} where both $p_1$ and $p_2$ are fixed and small, we allow the size of the observed block, $p_2$, to be large\footnote{We do not impose the restriction that $p_2$ is smaller than the available sample size.} and to grow with the sample size; yet the size of the latent block, $p_1$, can not be too large and is still assumed fixed. Moreover, the size of the informational series, $q$, can also be large and grow with the sample size. Further, we assume that the transition matrices $\{A^{(i)}\}_{i=1}^d$ and the regression coefficient matrix $\Gamma$ are {\em sparse}. 
	Finally, the factor loading matrix $\Lambda$ is assumed to be dense. 
	
	
	\subsection{Model identification considerations}\label{sec:id}
	
	The latent nature of $F_t$ leads to the following observational equivalence across the following two models encoded by $(\Lambda,\Gamma)$ and $(\widetilde{\Lambda},\widetilde{\Gamma})$, respectively: for any invertible matrix $Q_1\in\R^{p_1\times p_1}$ and $Q_2\in\R^{p_1\times p_2}$,
	\begin{equation*}
	Y_t = \Lambda F_t + \Gamma X_t + e_t ~\equiv ~\widetilde{\Lambda}\widetilde{F}_t + \widetilde{\Gamma} X_t + e_t, \qquad (Y_t\in\R^q, F_t\in\R^{p_1}, X_t\in\R^{p_2})
	\end{equation*}
	where 
	\begin{equation}\label{eqn:indeterminacy}
	\widetilde{\Lambda} := \Lambda Q_1, \qquad \widetilde{F}_t := Q_1^{-1}F_t - Q_1^{-1} Q_2 X_t, \qquad \widetilde{\Gamma} :=\Gamma + \Lambda Q_2.
	\end{equation}
	In other words, the key model parameters $(\Lambda,\Gamma)$ and the latent factors $F_t$ are {\em not uniquely} identified, a known problem even in classical factor analysis \citep{anderson1958introduction}. Thus, additional restrictions are required to overcome this indeterminacy, since there is an equivalence class parametrized by $(Q_1,Q_2)$ within which individual models are not mutually distinguishable based on observational data. For the FAVAR model, a total number of $p_1^2+p_1p_2$ restrictions
	are needed for unique identification of $\Lambda$, $\Gamma$ and $F_t$. 
	
	Various schemes have been proposed in the literature to address this issue. Specifically, \citet{bernanke2005measuring} impose the necessary restrictions through the coefficient matrices of the calibration equation, requiring $\Lambda=\left[\begin{smallmatrix}\I_{p_1} \\ * \end{smallmatrix}\right]$ and $\Gamma_{[1:p_1],\cdot}=0$; that is, the upper $p_1\times p_1$ block of $\Lambda$ is set to the identity matrix and the first $p_1$ rows of $\Gamma$ to zero. \citet{bai2016estimation} consider different sets of restrictions that involve combinations of coefficients from the calibration equation and the noise term from the VAR equation. In the low-dimensional setting ($p_2$ fixed), one can proceed to estimate the parameters subject to these restrictions, by adopting either a single-step Bayesian likelihood approach \citep{bernanke2005measuring} or an orthogonal projection-based approach by profiling out $X_t$ \citep{bai2016estimation}. However, neither approach is applicable in high-dimensional settings, due to the growing dimension $p_2$ which would render a projection-based approach infeasible or add to the computational demands of a Bayesian procedure.
	
	\medskip
	To overcome these issues in high-dimensional settings, we introduce an alternative identification scheme ``IR$+$Compactness" that is compatible with the model specification and can also be seamlessly incorporated in the estimation procedure, leveraging sparsity of the regression coefficient $\Gamma$. Specifically, we first impose constraint (IR):
	\begin{description}
		\item[(IR)]$\Lambda=\left[\begin{smallmatrix}\I_{p_1} \\ * \end{smallmatrix}\right]$: the upper $p_1\times p_1$ block of $\Lambda$ is an identity matrix, while the bottom block is left unconstrained.
	\end{description}
	Note that (IR) imposes $p_1^2$ constraints but crucially not on the latent factors, given their subsequent utilization in the VAR system. Further, it yields uniquely identifiable $\Lambda$ and $F_t$, for any given product 
	$\Lambda F_t$, and the indeterminacy incurred by $Q_1\in\R^{p_1\times p_1}$ in~\eqref{eqn:indeterminacy} vanishes. 
	
	However, the issue is not fully resolved, since for any 
	$Q_2\in\mathbb{R}^{p_1\times p_2}$, the following relationship holds:
	\begin{equation*}
	Y_t  = \Lambda F_t  + \Gamma X_t + e_t \equiv \Lambda\check{F}_t + \check{\Gamma}X_t + e_t, 
	\end{equation*}
	where 
	\begin{equation}\label{eqn:indeterminacy2}
	\check{F}_t = F_t -Q_2X_t \qquad \check{\Gamma} := \Gamma + \Lambda Q_2. 
	\end{equation}
	All such models encoded by $(\check{F}_t,\check{\Gamma})$, form an equivalence class parametrized by $Q_2$ that specifies the transformation. We denote this equivalence class by $\mathcal{C}(Q_2)$. If $Q_2=O$, then $\mathcal{C}(Q_2)$ degenerates to a singleton that contains only the true data-generating model, which requires the imposition of $p_1p_2$ restrictions on primary model quantities. One applicable constraint out of theoretical consideration is to impose orthogonality on $X_t$ and $F_t$ --- it yields the necessary $p_1p_2$ restrictions; yet is excessively stringent and limits the appeal of the FAVAR model, while also being challenging to operationalize. Therefore as a good working alternative, we address the identifiability issue through a weaker constraint that effectively limits sufficiently the size of the $\mathcal{C}(Q_2)$. 
	
	To this end, let $\mathbf{X}\in\R^{n\times p_2}$, $\mathbf{Y}\in\R^{n\times q}$ and $\mathbf{F}\in\R^{n\times p_1}$ be centered data matrices whose rows are samples of $X_t$, $Y_t$ and the latent process $F_t$ respectively, and $\check{\mathbf{F}}$ is analogously defined. The characterization of $\mathcal{C}(Q_2)$ is through the sample versions of the underlying processes. Specifically, define the set of {\em factor hyperplanes} induced by $\mathcal{C}(Q_2)$ by
	\begin{equation*}
	\mathcal{S}(\check{\Theta}):= \{ \check{\Theta}:=\check{\mathbf{F}}\Lambda^\top\in\mathbb{R}^{n\times q}~|~\check{\mathbf{F}}\text{ are samples of $\check{F}_t$ defined through~\eqref{eqn:indeterminacy2}}\},
	\end{equation*}
	and we let $\Theta^\star$ denote the factor hyperplane associated with the true data-generating model, to distinguish it from some generic element in $\mathcal{S}(\check{\Theta})$ that is denoted by $\check{\Theta}$. Note that 
	$\Theta^\star\in\mathcal{S}(\check{\Theta})$ and $\check{\Theta}$ coincides with $\Theta^\star$ when $Q_2=0$. Moreover, all elements in $\mathcal{S}(\check{\Theta})$ are at most of rank $p_1$, hence a low-rank component relative to their size $n\times q$. Next, in a similar spirit to \citet{negahban2012restricted}, we define the following constrained set:
	\begin{equation*}
	\mathcal{S}_{\phi}(\check{\Theta}) := \big\{\varphi_{\mathcal{R}}(\check{\Theta}) \leq \phi(n,q)~\big|~\check{\Theta}\in\mathcal{S}(\check{\Theta}) \big\},
	\end{equation*}
	where $\varphi_{\mathcal{R}}(\Theta)$ is defined according to
	\begin{equation*}
	\varphi_{\mathcal{R}}(\Theta) := \kappa(\mathcal{R}^*)\mathcal{R}^*(\Theta)\smalliii{\mathbf{X}/\sqrt{n}}_{\op},
	\end{equation*}
	and  $\kappa(\mathcal{R}^*) :=\sup\nolimits_{\Theta\neq 0}\big( \smalliii{\Theta}_{\F}/\mathcal{R}^*(\Theta)\big)$ with $\mathcal{R}^*$ being the dual norm of some regularizer $\mathcal{R}$. Base on the above definition, $\varphi_{\mathcal{R}}(\Theta)$ captures the interaction between the factor space and the observed $\mathbf{X}$-space; the product $\kappa(\mathcal{R}^*)\mathcal{R}^*(\Theta)$ measures the spikiness of $\Theta$ w.r.t. $\mathcal{R}$, and in the case where $\mathcal{R}$ corresponds to the sparsity-induced $\ell_1$-norm which would be the setup of interest in this paper (see Section~\ref{sec:formulation}), $\mathcal{R}^*(\Theta) = \|\Theta\|_\infty$ and $\kappa(\mathcal{R}^*) = \sqrt{nq}$. With the definition of $\mathcal{S}_{\phi}(\check{\Theta})$, we impose the following compactness constraint on $\check{\Theta}$ to further encourage identifiability:
	\begin{description}
		\item[(Compactness)] $\check{\Theta}\in \mathcal{S}_{\phi}(\check{\Theta})$ for some $\phi(n,q)$ satisfying $\phi(n,q)\geq\phi^\star:=\varphi_{\mathcal{R}}(\Theta^\star)$.
	\end{description}
	(Compactness) effectively limits the spikiness of all possible $\check{\Theta}$'s by imposing a {\em box constraint} through the dual norm corresponding to the sparsity regularizer, and for an arbitrary set of fixed realizations, it restricts the factor hyperplane set induced by $\mathcal{C}(Q_2)$ to its $\phi$-radius subset $\mathcal{S}_{\phi}(\check{\Theta})$. This in turn limits the size of the equivalence class $\mathcal{C}(Q_2)$ under consideration, since there is a one-to-one correspondence at the set level between $\mathcal{C}(Q_2)$ and the factor hyperplane set induced by it. This further implies that although the models encoded by $(F_t,\Gamma)$ and $(\check{F}_t,\check{\Gamma})$ may not be perfectly distinguishable based on observational data, at the population level the discordance between the two models can not be too large. It is worth pointing out that the bound $\phi(n,q)$ is allowed to grow, but at a much slower rate than the size of $\check{\Theta}$; specifically, we require $\phi(n,q)=o(\kappa(\mathcal{R}^*))$. For ease of presentation, we use $\phi$ to denote this bound henceforth and further note that it is in fact a constant in any finite sample setting.
	
	In summary, our proposed identification scheme comprises of two parts: (IR) and (Compactness). The former provides exact identification within the factor hyperplane and narrows the scope of observationally equivalent models to $\mathcal{C}(Q_2)$, while the latter limits its size; and they jointly incur {\em approximate identification} of the true
	data generating model; and thus, for estimation purposes henceforth, it becomes adequate to focus on this restricted equivalence class, rather than its individual elements. The proposed scheme is suitable for the high-dimensional nature of the problem and can easily be incorporated in the formulation of the optimization problem for parameter estimation (see Section~\ref{sec:formulation}), which in turn yields estimates with tight error bounds (see Section~\ref{sec:theory}).

	
	\subsection{Proposed formulation}\label{sec:formulation}
	
	Without loss of generality, we focus on the case where $d=1$ in subsequent technical developments, so that $Z_t:=(F^\top_t,X^\top_t)^\top$ follows a $\VAR(1)$ model $Z_t = AZ_{t-1} + W_t$:
	\begin{equation}\label{md:FAVAR-VAR1}
	\begin{bmatrix}
	F_t \\ X_t
	\end{bmatrix} = \begin{bmatrix}
	A_{11} & A_{12} \\ A_{21} & A_{22}
	\end{bmatrix} \begin{bmatrix}
	F_{t-1} \\ X_{t-1}
	\end{bmatrix} + \begin{bmatrix}
	w^F_t \\ w^X_t
	\end{bmatrix}.
	\end{equation}
	The generalization to the $\VAR(d)$~$(d>1)$ case is straightforward since for any generic $\VAR(d)$ process satisfying $\mathcal{A}_d(L) Z_t = w_t$ where $\mathcal{A}_d(L):=\mathrm{I}-A^{(1)}L - \cdots - A^{(d)}L^d$, it can always be written in the form of a $\VAR(1)$ model for some $dp$-dimensional process $\widetilde{Z}_t$ \citep[see][for details]{lutkepohl2005new}.
	
	Based on the introduced model identification scheme (IR+Compactness), we propose the following procedure to estimate the FAVAR model, whose parameters include a sparse coefficient matrix $\Gamma$, a dense loading matrix $\Lambda$, and a sparse transition matrix $A$. Observed data matrices $\mathbf{X}$ and $\mathbf{Y}$ are identical to what have been previously defined, and to distinguish the responses from their lagged predictors when considering the VAR system, we let $\mathbf{X}_{n-1}:=[x_1,\dots,x_{n-1}]^\top$ denote the predictor matrix and $\mathbf{X}_{n}:=[x_2,\dots,x_n]^\top$ the response one; $\mathbf{F}_n, \mathbf{F}_{n-1}, \mathbf{Z}_n,\mathbf{Z}_{n-1}$ are analogously defined. Based on these notations, the sample versions of the VAR system and the calibration equation in~\eqref{md:FAVAR-VAR1} and~\eqref{md:FAVAR-info} can be written as 
	\begin{equation*}
	\mathbf{Z}_n  = ~\mathbf{Z}_{n-1}A^\top + \mathbf{W}, \quad \text{and} \quad 
	\mathbf{Y}  = ~\mathbf{F}\Lambda^\top + \mathbf{X}\Gamma^\top + \mathbf{E} =: \Theta + \mathbf{X}\Gamma^\top + \mathbf{E}.
	\end{equation*}
	We propose the following estimators obtained from a two-stage procedure for the coefficient matrices $\Lambda$, $\Gamma$ and subsequently the transition matrices~$\{A_{ij}\}_{i,j=1,2}$.
	\begin{itemize}
		\item Stage I: estimation of the calibration equation under (IR+Compactness). We formulate the following {\em constrained optimization} problem using a least squares loss function and  incorporating the sparsity-induced $\ell_1$ regularization of the sparse block $\Gamma$, the rank constraint on the hyperplane~$\Theta$, and (Compactness):
		\begin{equation}\label{opt:solveTheta}
		\begin{split}
		(\widehat{\Theta},\widehat{\Gamma})~~:= &\argmin\limits_{\Theta\in\R^{n\times q},\Gamma\in\R^{q\times p_2}} \big\{\frac{1}{2n}\smalliii{\mathbf{Y}-\Theta-\mathbf{X}\Gamma^\top}_{\F}^2  + \lambda_\Gamma\|\Gamma\|_1\big\}, \\ &\text{subject to}~~~\text{rank}(\Theta)\leq r,~~~ \|\Theta\|_\infty\leq \frac{\phi}{\sqrt{nq}\cdot\smalliii{\mathbf{X}/\sqrt{n}}_\op}.
		\end{split}
		\end{equation}
		Once $\widehat{\Theta}$ is obtained, under (IR), the estimated factors $\widehat{\mathbf{F}}$ and the corresponding loading matrix $\widehat{\Lambda}$ are extracted as follows:
		\begin{equation}\label{eqn:F-Lambda}
		\widehat{\mathbf{F}} = \widehat{\mathbf{F}}^{\PC}(\widehat{\Lambda}^{\PC}_1)^{\top}, \qquad \widehat{\Lambda}=\widehat{\Lambda}^{\PC}(\widehat{\Lambda}^{\PC}_1)^{-1},
		\end{equation}
		where $\widehat{\Lambda}^{\PC}_1$ is the upper $p_1$ sub-block of $\widehat{\Lambda}^{\PC}$, with $\widehat{\mathbf{F}}^{\PC}$ and $\widehat{\Lambda}^{\PC}$ being the PC estimators \citep{stock2002forecasting} given by $\widehat{\mathbf{F}}^{\PC} := \sqrt{n}\widehat{U}$ and $\widehat{\Lambda}^{\PC} :=\widehat{V}\widehat{D}/\sqrt{n}$. The estimates $\widehat{U},\widehat{D}$ and $\widehat{V}$ are obtained from the SVD of $\widehat{\Theta}=\widehat{U}\widehat{\Theta}\widehat{V}^{\top}$. Note that after these algebra, $\widehat{\mathbf{F}}$ corresponds to the first $p_1$ columns of $\widehat{\Theta}$.  
		
		Of note, $\smalliii{\mathbf{X}/\sqrt{n}}_{\op}^2 =\Lambda_{\max}(\mathbf{X}^\top\mathbf{X}/n)$ and it can be shown that for any random realizations $\mathbf{X}$, the latter can be bounded with high probability (see Lemma~\ref{lemma:Smax}).
		\item Stage II: estimation of the VAR equation based on $\mathbf{X}$ and $\widehat{\mathbf{F}}$. With the estimated factor $\widehat{\mathbf{F}}$ as the surrogate for the true latent factor $\mathbf{F}$, the transition matrix $A$ can be estimated by solving
		\begin{equation}\label{opt:solveA}
		\widehat{A}~~:=\argmin\limits_{A\in\R^{(p_1+p_2)\times (p_1+p_2)}} \big\{ \frac{1}{2n} \smalliii{\widehat{\mathbf{Z}}_n - \widehat{\mathbf{Z}}_{n-1}A^\top}_{\F}^2 + \lambda_A \|A\|_1   \big\},
		\end{equation}
		where $\widehat{\mathbf{Z}}_n:=[\widehat{\mathbf{F}}_n,\mathbf{X}_n]$ and $\widehat{\mathbf{Z}}_{n-1}$ is analogously defined. The $\ell_1$-norm penalty induces sparsity on $A$ according to the model assumption.
	\end{itemize}
	In the presence of additional contemporaneous dependence amongst the coordinates for the error processes $w_t$, one may consider a maximum likelihood-based loss function, but the full estimation would require additional structural assumptions of $\Sigma_w$ (or its inverse) given the high dimensionality; we do not further elaborate in this study, since our prime interest is estimating the coefficient/transition matrices of the FAVAR model.
	
	The formulation in~\eqref{opt:solveA} based on the least squares loss function and the surrogate $\widehat{\mathbf{F}}$ is straightforward. However, the formulation for the calibration equation merits additional discussion. First, note that the factor hyperplane $\Theta$ has at most rank $p_1$ and therefore has low rank structure relative to its size $n\times q$.
	We impose a rank constraint in the estimation procedure to enforce such structure. Together with the (IR+Compactness) constraint introduced above, the objective then becomes to estimate accurately the parameters of a model within the equivalence class $\mathcal{C}(Q_2)$, in the sense that the estimate obtained by solving~\eqref{opt:solveTheta} effectively corresponds to recovering an arbitrary $\check{\Theta},\check{\Theta}\in\mathcal{C}(Q_2)$; such an estimate, however, will be close to the true data generating $\Theta^\star$. Once this goal is achieved, this would enable accurate estimation of the transition matrix of the VAR system.
	
	From an optimization perspective, the objective function admits a low-rank-plus-sparse decomposition and compactification is necessary for establishing statistical properties of the global optima in the absence of explicitly specifying the interaction structure between the low rank and the sparse blocks (or the spaces they live in). Note that the form of the compactness constraint is dictated by the statistical problem under consideration. For example, \citet{agarwal2012noisy} 
	study a multivariate regression problem, where the coefficient is decomposed to a sparse and a low rank block. In that setting,
	a compactness constraint is imposed through the entry-wise infinity norm bound of the low rank block. \citet{chandrasekaran2012latent} study a graphical model with latent variables where the conditional concentration matrix is the parameter of interest. The marginal concentration matrix is decomposed to a sparse and a low rank block via the alignment of the Schur complement, and the compactness constraint is imposed on both blocks and manifests through the corresponding regularization terms in the resulting optimization problem. Hence, the compactness constraint takes different forms but ultimately serves the same goal,
	namely, to introduce an upper bound on the magnitude of the low rank--sparse block interaction, with the latter being an important component in analyzing the estimation errors.
	The compacteness constraint adopted for the FAVAR model serves a similar purpose, although the presence of temporal dependence introduces a number of additional technical challenges compared to the two aforementioned settings that consider independent and identically distributed data.
	
	Finally, we remark that the model identification scheme (IR+Compactness) incorporated in the optimization problem as a constraint, enables us to establish high-probability error bounds (relative to the true data generating parameters/factors) for the proposed estimators, as shown next in Section~\ref{sec:theory}. Therefore, although (IR+Compactness) does not encompass the full $p_1^2+p_1p_2$ restrictions, it provides sufficient identifiability for estimation purposes.
	
	\section{Theoretical Properties}\label{sec:theory}
	
	In this section, we investigate the theoretical properties of the estimators proposed in Section~\ref{sec:formulation}. We focus on formulations~\eqref{opt:solveTheta} and~\eqref{opt:solveA}, whose global optima correspond to $(\widehat{\Theta},\widehat{\Gamma})$ and $\widehat{A}$, respectively.
	
	Since~\eqref{opt:solveA} relies not only on prime observable quantities (namely $X_t$), but also on estimated quantities from Stage I (namely $\widehat{\mathbf{F}}$), the analysis requires a careful examination of how the estimation error in the factor propagates to that of $\widehat{A}$. We start by outlining a road map of our proof strategy together with a number
	of regularity conditions needed in subsequent developments. Section~\ref{sec:theory-fix} establishes error bounds 
	for $\widehat{\Gamma}$, $\widehat{\Theta}$~\footnote{Consequently, the error bounds of $\widehat{\mathbf{F}}$ and $\widehat{\Lambda}$ under (IR) are also obtained.}and $\widehat{A}$ under certain regularity conditions and employing suitable choices of the tuning parameters, for {\em deterministic realizations} from the underlying observable processes. Specifically when considering the error bound of $\widehat{A}$, the error of the plug-in estimate $\widehat{\mathbf{F}}$ is assumed non-random and given. Subsequently, Section~\ref{sec:theory-random} examines the probability of the events in which the regularity conditions are satisfied for {\em random realizations}, and further establishes high-probability upper bounds for quantities to which the tuning parameters need to conform. Finally, the high-probability finite sample error bounds for the estimates obtained based on random realizations of the data generating processes readily follow after properly aligning the conditioning arguments, and the results are presented in Section~\ref{sec:theory-high}. All proofs are deferred to Appendices~\ref{sec:proof-thm-prop} and~\ref{appendix:lemma}. 
	
	\paragraph{Additional notations.} Throughout, we use superscript $\star$ to denote the true value of the parameters of interest, and $\Delta$ for errors of the estimators; e.g., $\Delta_{A}=\widehat{A}-A^\star$. For sample quantities (e.g., $\mathbf{X}$ and $\mathbf{F}$) and their corresponding error (e.g., $\Delta_{\mathbf{F}}$), we use subscript $(n-1)$ to denote their first $n-1$ rows. We let $S_{\mathbf{E}}:=\tfrac{1}{n}\mathbf{E}^\top\mathbf{E}$ denote the sample covariance matrix of $\mathbf{E}$ and the sample covariance of other quantities are analogously defined. Additionally, denote the density level of $\Gamma^\star$ by $s_{\Gamma^\star}:=\|\Gamma^\star\|_0$, and that of $A^\star$ by $s_{A^\star}$. 
	
	\paragraph{A road map for establishing consistency results.} As previously mentioned, the key steps are:
	\begin{itemize}
		\item Part 1: analyses based on deterministic realizations using the optimality of the estimators, assuming the parameters of the objective function (e.g., the Hessian and the penalty parameter) satisfy certain regularity conditions;
		\item Part 2: analyses based on random realizations that the probability of the regularity conditions being satisfied,  primarily involving the utilization of concentration inequalities. 
	\end{itemize}
	In Part 1, note that the first-stage estimators obtained from the calibration equation are based on observed data and thus
	the regularity conditions needed are imposed on (functions of) the observed samples. On the other hand, the second-stage estimator relies on the plugged-in first-stage estimates that have bounded errors; therefore, the analysis is carried out in 
	an analogous manner to problems involving error-in-variables. Specifically, the required regularity conditions on quantities appearing in the optimization~\eqref{opt:solveA} involve the error of the first stage estimates, with the latter assumed fixed. In Part 2, the focus shifts to the probability of the regularity conditions being satisfied under random realizations, again starting from the first stage estimates, with the aid of Gaussian concentration inequalities and proper accounting for temporal dependence. Once the required regularity conditions are shown to hold with high probability, combining the results established in Part 1 for deterministic realizations, the high-probability error bounds for $\widehat{\Theta}$ and $\widehat{\Gamma}$ are established. The high-probability error bound of the estimated factors readily follows, which ensures that the variables which Stage II estimates rely upon are sufficiently accurate with high probability. Based on the latter result, the regularity conditions required for the Stage II estimates are then verified to hold with high probability at a certain rate. In the FAVAR
	model, since the estimation of the VAR equation is based on quantities among which one block is subject to error, to obtain an accurate estimate of the transition matrix requires more stringent conditions on population quantities (e.g., extremes of the spectrum), so that the regularity conditions hold with high probability. In essence, the joint process $Z_t$ need to be adequately ``regular" in order to get good estimates of the transition matrix , vis-a-vis the case of the standard VAR model where all variables are directly observed.
	
	\medskip
	Next, we introduce the following key concepts that are widely used in establishing theoretical properties of high-dimensional regularized $M$-estimators \citep[e.g.][]{negahban2012unified,loh2012high}, as well as quantities that are related to processes exhibiting temporal dependence \citep[see also][]{basu2015estimation}.
	
	\begin{definition}[Restricted strong convexity (RSC)] A matrix $\mathbf{X}\in\R^{n\times p}$ satisfies the RSC condition with respect to norm $\Phi$ with curvature $\alpha_{\text{RSC}}>0$ and tolerance $\tau_n\geq 0$, if
		\begin{equation*}
		\frac{1}{2n} \smalliii{\mathbf{X}\Delta}_{\F}^2 \geq \frac{\alpha_{\text{RSC}}}{2}\smalliii{\Delta}_{\F}^2 - \tau_n \Phi^2(\Delta), \qquad \forall~\Delta\in\R^{p\times p}.
		\end{equation*}	
		In our setting, we consider the norm $\Phi(\Delta)=\|\Delta\|_1$. 
	\end{definition}
	
	\begin{definition}[Deviation condition]\label{defn:db} For a regularized $M$-estimator given in the generic form of
		\begin{equation*}
		\widehat{A}:= \min_{A}\big\{ \frac{1}{2n}\smalliii{\mathbf{Y}-\mathbf{X}A^\top}_{\F}^2 + \lambda_A\|A\|_1 \big\},
		\end{equation*}
		with $\mathcal{H}_A:=\tfrac{1}{n}\mathbf{X}^\top\mathbf{X}$ denoting the Hessian and $\mathcal{G}_A:=\tfrac{1}{n}\mathbf{Y}^\top \mathbf{X}$ denoting the gradient, we define the tuning parameter $\lambda_A$ to be selected in accordance with the deviation condition, if 
		\begin{equation*}
		\lambda_A\geq c_0 \|\mathcal{H}_{A} - \mathcal{G}_{A}(A^\star)^\top\|_\infty, \quad \text{for some }c_0.
		\end{equation*} 
	\end{definition}
	Under the current model setup, however, the exact form of the deviation bound becomes more involved and requires proper modifications to incorporate quantities associated with the factor hyperplane, as seen in Proposition~\ref{prop:Theta-Gamma-fixed}.
	
	\begin{definition}[Spectrum and its extremes] For a $p$-dimensional stationary process $X_t$, its spectral density $f_X(\omega)$ is defined~as
		\begin{equation}
		f_X(\omega) := \tfrac{1}{2\pi} \sum_{h=-\infty}^\infty \Sigma_X(h) e^{i\omega h},
		\end{equation}
		where $\Sigma_X(h):= \mathbb{E}(X_tX_{t+h}^\top)$. Its upper and lower extremes are defined as
		\begin{equation*}
		\mathcal{M}(f_X):=\mathop{\text{ess sup}}\limits_{\omega\in[-\pi,\pi]} \Lambda_{\max}(f_X(\omega)), \quad \text{and}\quad  \mathfrak{m}(f_X):=\mathop{\text{ess inf}}\limits_{\omega\in[-\pi,\pi]} \Lambda_{\min}(f_X(\omega)).
		\end{equation*}
		The cross-spectrum for two generic stationary processes $X_t$ and $Y_t$ is defined as 
		\begin{equation*}
		f_{X,Y}(\omega) := \frac{1}{2\pi} \sum_{h=-\infty}^\infty \Sigma_{X,Y}(h) e^{i\omega h},
		\end{equation*}
		where $\Sigma_{X,Y}(h):= \mathbb{E}(X_tY_{t+h}^\top)$, and its upper extreme is defined as
		\begin{equation*}
		\mathcal{M}(f_{X,Y}):=\mathop{\text{ess sup}}\limits_{\omega\in[-\pi,\pi]} \sqrt{\Lambda_{\max}\big(f^*_{X,Y}(\omega)f_{X,Y}(\omega) \big)}, 
		\end{equation*}
		where $*$ denotes the conjugate transpose.
	\end{definition}
	
	\medbreak
	We start by providing error bounds for $\widehat{\Gamma}$ and  $\widehat{\Theta}$, as well as those of the corresponding 
	$\widehat{\mathbf{F}}$ and $\widehat{\Lambda}$ extracted under (IR). For the optimization problem given in~\eqref{opt:solveTheta}, we assume that $r\geq p_1$ and $\phi$ is always compatible with the true data generating mechanism, so that $\Theta^\star$ is always feasible. To this end, the error bounds of $\widehat{\Theta}$ and $\widehat{\Gamma}$ for deterministic realizations crucially rely on two components: (i) $\mathbf{X}$ satisfying the RSC condition with curvature $\alpha_{\RSC}^{\mathbf{X}}$; and (ii) the tuning parameter $\lambda_\Gamma$ being chosen in accordance with the deviation bound condition that is associated with the interaction between $\mathbf{X}$ and $\mathbf{E}$, the strength of the noise, and the interaction between the space spanned by the factor hyperplane and the observed $\mathbf{X}$. Upon the satisfaction of these conditions, the error bounds of $\widehat{\Theta}$ and $\widehat{\Gamma}$ are given by
	\begin{equation*}
	\smalliii{\Delta_\Gamma}_{\F}^2 + \smalliii{\Delta_{\Theta}/\sqrt{n}}_{\F}^2 \leq C_1\lambda_\Gamma^2\big(p_1+r+4s_{\Gamma^\star}\big)/\min\{\alpha^{\mathbf{X}}_{\RSC},1\}^2,
	\end{equation*} 
	and these conditions hold with high probability for random realizations of $X_t$ and $Y_t$. Since $\widehat{\mathbf{F}}$ is the first $p_1$ columns of $\widehat{\Theta}$, it possesses an error bound of the similar form. 
	
	Next, we briefly sketch the error bounds of $\widehat{A}$. For the optimization in~\eqref{opt:solveA}, for deterministic realizations, the results in \citet{basu2015estimation} can be applied with the corresponding RSC condition and deviation condition imposed on quantities associated with $\widehat{\mathbf{Z}}_n$ and $\widehat{\mathbf{Z}}_{n-1}$, and the error for $\widehat{A}$ is in the form of
	\begin{equation*}
	\smalliii{\Delta_A}_{\F}^2 \leq C_2s_{A^\star}\lambda_A^2/(\alpha_{RSC}^{\widehat{\mathbf{Z}}})^2.
	\end{equation*}
	Then, for random realizations, assuming $\Delta_{\mathbf{F}}$ known and non-random, to satisfy the corresponding regularity conditions, we additionally require that the following functional involving the spectral density of the underlying joint process $Z_t$ exhibits adequate curvature, that is, $\mathfrak{m}(f_Z)/\sqrt{\mathcal{M}(f_Z)} > c_0 h_1(\Delta_{\mathbf{F}_{n-1}})$ for constant $c_0$ and some function $h_1$ of the error $\Delta_{\mathbf{F}_{n-1}}$ that captures its magnitude. Moreover, the deviation bound is of the form $h_2(\Delta_{\mathbf{F}})$, which can be viewed as another function of the error\footnote{note the deviation bound in principle also depends on other population quantities such as $\mathfrak{m}(f_Z)$, $\mathcal{M}(f_Z)$, $\Lambda_{\max}(\Sigma_w)$ etc.}. Further, since $\Delta_{\mathbf{F}}$ is bounded with high probability from the analysis in Stage I, it will be established that $h_1(\Delta_{\mathbf{F}})$ and $h_2(\Delta_{\mathbf{F}})$ are both upper bounded at a certain rate, thus ensuring that the RSC condition and the deviation conditions can both be satisfied unconditionally, by properly choosing the required constants. 
	
	\subsection{Statistical error bounds with deterministic realizations}\label{sec:theory-fix}
	
	Proposition~\ref{prop:Theta-Gamma-fixed} below gives the error bounds for the estimators in~\eqref{opt:solveTheta}, assuming certain regularity conditions hold for deterministic realizations of the processes $X_t$ and $Y_t$, upon suitable choice of the regularization parameters.  
	
	\begin{proposition}[Bound for $\Delta_{\Theta}$ and $\Delta_{\Gamma}$ under fixed realizations]\label{prop:Theta-Gamma-fixed} Suppose the fixed realizations $\mathbf{X}\in\R^{n\times p_2}$ of process $\{X_t\in\R^{p_2}\}$ satisfies the RSC condition with curvature $\alpha_{\text{RSC}}^{\mathbf{X}}>0$ and a tolerance $\tau_{\mathbf{X}}$ for which
		\begin{equation*}
		\tau_{\mathbf{X}}\cdot\big(p_1+r+4s_{\Gamma^\star}\big) < \min\{\alpha^{\mathbf{X}}_{\text{\RSC}},1\}/16. 
		\end{equation*}
		Then, for any matrix pair $(\Theta^\star,\Gamma^\star)$ satisfying the constraint $\varphi_{\mathcal{R}}(\Theta^\star)\leq \phi$ that generates $\mathbf{Y}$, for estimators $(\widehat{\Theta},\widehat{\Gamma})$ obtained by solving~\eqref{opt:solveTheta} with regularization parameters $\lambda_\Gamma$ satisfying
		\begin{equation*}
		\lambda_\Gamma\geq \max\big\{ 2\|\mathbf{X}^\top\mathbf{E}/n\|_\infty,~4\phi/\sqrt{nq},~\Lambda_{\max}^{1/2}(S_{\mathbf{E}})\big\},
		\end{equation*}
		the following bound holds:
		\begin{equation}\label{eqn:boundThetaB}
		\smalliii{\Delta_\Gamma}_{\F}^2 + \smalliii{\Delta_{\Theta}/\sqrt{n}}_{\F}^2 \leq \frac{16\lambda_\Gamma^2 \big( p_1 + r + 4s_{\Gamma^\star} \big)}{\min\{\alpha^{\mathbf{X}}_{\RSC},1\}^2}.
		\end{equation}
	\end{proposition}
	
	Based on Proposition~\ref{prop:Theta-Gamma-fixed}, under fixed realizations of $X_t$ and $Y_t$, the error bounds of $\widehat{\Gamma}$ and $\widehat{\Theta}$ are established. Using these Stage I estimates and the IR condition, estimates of the
	factors and their loadings can be calculated. In particular, since $\Delta_{\mathbf{F}}$ corresponds to the first $p_1$ columns of $\Delta_{\Theta}$, the above bound automatically holds for $\Delta_{\mathbf{F}}$. Further, the following lemma provides the relative error of the estimated $\Lambda$ under (IR) and the condition on $\Lambda^{1/2}_{\max}(S_{\mathbf{F}})$, with the latter
	translating to the requirement that the leading signal of $\mathbf{F}$ overrules the averaged row error of $\Delta_{\Theta}$.
	
	\begin{lemma}[Bound of $\Delta_\Lambda$]\label{lemma:err-Lambda} The following error bound holds for $\widehat{\Lambda}$, provided that $\Lambda_{\max}^{1/2}(S_{\mathbf{F}})>\smalliii{\Delta_{\Theta}/\sqrt{n}}_{\F}$:
		\begin{equation}\label{eqn:bound1}
		\frac{\smalliii{\Delta_\Lambda}_\F}{\smalliii{\Lambda^\star}_\F}\leq \frac{\sqrt{p_1}\cdot\smalliii{\Delta_{\Theta}/\sqrt{n}}_{F}}{\Lambda^{1/2}_{\max}(S_{\mathbf{F}})- \smalliii{\Delta_{\Theta}/\sqrt{n}}_{\F}}\Big(1 + 1/\smalliii{\Lambda^\star}_{\F} \Big).
		\end{equation}
	\end{lemma}
	
	Up to this point, error bounds have been obtained for all the parameters in the calibration equation. The following proposition establishes the error bound for the estimator obtained from solving ~\eqref{opt:solveA}, based on observed $\mathbf{X}$ and estimated $\widehat{\mathbf{F}}$, and assuming $\Delta_{\mathbf{F}}$ is fixed. 
	\begin{proposition}[Bound for $\Delta_A$ under fixed realization and a non-random $\Delta_{\mathbf{F}}$]\label{prop:A-fixed} Consider the estimator $\widehat{A}$ obtained by solving~\eqref{opt:solveA}. Suppose the following conditions hold:
		\begin{itemize}
			\item[A1.] $\widehat{\mathbf{Z}}_{n-1}:=[\widehat{\mathbf{F}}_{n-1},\mathbf{X}_{n-1}]$ satisfies the RSC condition with curvature $\alpha_{\RSC}^{\widehat{\mathbf{Z}}}$ and tolerance $\tau_\mathbf{Z}$ for which $s_{A^\star}\tau_{\mathbf{Z}} < \alpha_{\RSC}^{\widehat{\mathbf{Z}}}/64$;
			\item[A2.] $\|\widehat{\mathbf{Z}}_{n-1}^\top\big( \widehat{\mathbf{Z}}_n - \widehat{\mathbf{Z}}_{n-1}(A^\star)^\top \big)/n\|_\infty \leq C(n,p_1,p_2)$ where $C(n,p_1,p_2)$ is some function that depends on $n,p_1$ and $p_2$. 
		\end{itemize}
		Then, for any $\lambda_A\geq 4C(n,p_1,p_2)$, the following error bound holds for $\widehat{A}$:
		\begin{equation*}
		\smalliii{\Delta_A}_{\F} \leq 16\sqrt{s_{A^\star}}\lambda_A/\alpha_{\RSC}^{\widehat{\mathbf{Z}}}.
		\end{equation*}
	\end{proposition}
	Note that Proposition~\ref{prop:A-fixed} applies the results in \citet[Proposition 4.1]{basu2015estimation} to the setting in this study, where Stage II estimation of the transition matrix is based on $\widehat{\mathbf{Z}}_n$ and $\widehat{\mathbf{Z}}_{n-1}$; consequently, the regularity conditions should be imposed on corresponding quantities associated with $\widehat{\mathbf{Z}}_{n}$ and $\widehat{\mathbf{Z}}_{n-1}$.
	
	Propositions~\ref{prop:Theta-Gamma-fixed} and~\ref{prop:A-fixed} give finite sample error bounds for the estimators of the parameters obtained by solving optimization problems~\eqref{opt:solveTheta} and~\eqref{opt:solveA} based on fixed realizations of the observable processes $X_t$ and $Y_t$, and the regularity conditions outlined. Next, we examine and verify these conditions for random realizations of the processes, to establish high probability error bounds for these estimators.
	
	\subsection{High probability bounds under random realizations} \label{sec:theory-random}
	
	We provide high probability bounds or concentrations for the quantities associated with the required regularity conditions, for random realizations of $X_t$ and $Y_t$. Specifically, we note that when $X_t$ is considered separately from the joint system, it follows a high-dimensional VAR-X model \citep{lin2017regularized}
	\begin{equation*}
	X_t = A_{22}X_{t-1} + A_{21}F_{t-1} + w_t^X, 
	\end{equation*}
	whose spectrum $f_X(\omega)$ satisfies 
	\begin{equation*}
	f_X(\omega) = \big[\mathcal{A}^{-1}_X(e^{-i\omega})\big]\big( A_{21}f_F(\omega)A_{21}^\top + f_{w^X}(\omega) + f_{w^X,F}(\omega)A_{21}^\top + A_{21}f_{F,w^X}(\omega) \big)\big[\mathcal{A}^{-1}_X(e^{-i\omega})\big]^*,
	\end{equation*}
	where $\mathcal{A}_{X}(L):=\mathrm{I}-A_{22}L$. Similar properties hold for $F_t$. Throughout, we assume $\{X_t\},\{F_t\}$ and $\{Y_t\}$ are all mean-zero stable Gaussian processes. 
	
	Lemmas~\ref{lemma:RSCX} to~\ref{lemma:Emax} respectively verify the RSC condition associated with $\mathbf{X}$ and establish the high probability bounds for $\|\mathbf{X}^\top\mathbf{E}/n\|_\infty$, $\Lambda_{\max}(S_{\mathbf{E}})$ and $\Lambda_{\max}(S_{\mathbf{X}})$.
	
	\begin{lemma}[Verification of the RSC condition for $\mathbf{X}$]\label{lemma:RSCX} 
		Consider $\mathbf{X}\in\R^{n \times p_2}$ whose rows correspond to a random realization $\{x_1,\dots,x_{n}\}$ of the stable Gaussian $\{X_t\}$ process, and its dynamics are governed by~\eqref{md:FAVAR-VAR1}. Then, there exist positive constants 
		$c_i>0, i=1,2$, such that with probability at least $1-c_1\exp(-c_2n\min\{\gamma^{-2},1\})$ where $\gamma:=54\mathcal{M}(g_X)/\mathfrak{m}(g_X)$, the RSC condition holds for $\mathbf{X}$ with curvature $\alpha_{\RSC}^{\mathbf{X}}$ and tolerance~$\tau_\mathbf{X}$ satisfying
		\begin{equation*}
		\alpha_{\RSC}^{\mathbf{X}} = \pi\mathfrak{m}(f_X), \qquad \tau_\mathbf{X} = \alpha_{\RSC}\gamma^2\Big( \frac{\log p_2}{n}\Big)/2~,
		\end{equation*}
		provided that $n\gtrsim \log p_2$. 		
	\end{lemma}
	
	\begin{lemma}[High probability bound for $\|\mathbf{X}^\top\mathbf{E}/n\|_\infty$]\label{lemma:boundXE} There exist positive constants $c_i~(i=0,1,2)$ such that for sample size $n\gtrsim \log (p_2q)$, with probability at least $1-c_1\exp(-c_2\log (p_2q))$, the following bound holds:
		\begin{equation}\label{eqn:boundXE}
		\|\mathbf{X}^\top\mathbf{E}/n\|_\infty \leq c_0\Big(2\pi\mathcal{M}(f_X) + \Lambda_{\max}(\Sigma_e) \Big)\sqrt{\frac{\log p_2 + \log q}{n}}.
		\end{equation}
	\end{lemma}
	
	\begin{lemma}[High probability bound for $\Lambda_{\max}(S_{\mathbf{E}})$]\label{lemma:Emax} Consider $\mathbf{E}\in\R^{n\times q}$ whose rows are independent realizations of the mean zero Gaussian random vector $e_t$ with covariance $\Sigma_e$. Then, for sample size $n\gtrsim q$, with probability at least $1-\exp(-n/2)$, the following bound holds:
		\begin{equation*}
		\Lambda_{\max}(S_{\mathbf{E}}) \leq 9\Lambda_{\max}(\Sigma_e). 
		\end{equation*}
	\end{lemma}
	
	\begin{lemma}[High probability bound for $\Lambda_{\max}(S_\mathbf{X})$]\label{lemma:Smax} Consider $\mathbf{X}\in\R^{n\times p_2}$ whose rows correspond to a random realization $\{x_1,\dots,x_{n}\}$ of the stable Gaussian $\{X_t\}$ process, and its dynamics are governed by~\eqref{md:FAVAR-VAR1}. There exist positive constants $c_i>0, i=0,1,2$, such that  for sample size $n\gtrsim p_2$, with probability at least $1-c_1\exp(-c_2n)$, the following bound holds:
		\begin{equation*}
		\Lambda_{\max}(S_{\mathbf{X}})\leq c_0\mathcal{M}(f_X).
		\end{equation*}
	\end{lemma}

	In the next two lemmas, we verify the RSC condition for random realizations of $\widehat{\mathbf{Z}}_{n-1}$ and obtain the high probability bound $C(n,p_1,p_2)$ for $\|\widehat{\mathbf{Z}}_{n-1}^\top\big( \widehat{\mathbf{Z}}_n - \widehat{\mathbf{Z}}_{n-1}(A^\star)^\top \big)/n\|_\infty$, with the underlying truth $\mathbf{F}$ being random but the error $\Delta_{\mathbf{F}}$ non-random. Note that this can be equivalently viewed as a {\em conditional} RSC condition and deviation bound, when conditioning on some fixed $\Delta_{\mathbf{F}}$. 
	\begin{lemma}[Verification of RSC for $\widehat{\mathbf{Z}}_{n-1}$]\label{lemma:RSCZ} Consider $\widehat{\mathbf{Z}}_{n-1}$ given by 
		\begin{equation*}
		\widehat{\mathbf{Z}}_{n-1} = \mathbf{Z}_{n-1} + \Delta_{\mathbf{Z}_{n-1}} = [\mathbf{F}_{n-1},\mathbf{X}_{n-1}] + [\mathbf{\Delta}_{\mathbf{F}_{n-1}},O],
		\end{equation*}
		with rows of $[\mathbf{F}_{n-1},\mathbf{X}_{n-1}]$ being a random realization drawn from process $\{Z_t\}$ whose dynamics are given by~\eqref{md:FAVAR-VAR1}. Suppose the lower and upper extremes of its spectral density $f_Z(\omega)$ satisfy
		\begin{equation*}\label{eqn:fZcondition}
		\mathfrak{m}(f_Z)/\mathcal{M}^{1/2}(f_Z) > c_0\cdot \Lambda^{1/2}_{\max}\big(S_{\Delta_{\mathbf{F}_{n-1}}}\big),~~~\text{where}~~S_{\Delta_{\mathbf{F}_{n-1}}}:=\Delta_{\mathbf{F}_{n-1}}^\top\Delta_{\mathbf{F}_{n-1}}/n,
		\end{equation*}
		for some constant $c_0\geq 6$. Then, with probability at least $1-c_1\exp(-c_2n)$, $\widehat{\mathbf{Z}}_{n-1}$ satisfies the RSC condition with curvature
		\begin{equation}\label{eqn:curvatureZ}
		\alpha_{\text{RSC}}^{\widehat{\mathbf{Z}}}= \pi\mathfrak{m}(f_Z) - 54\Lambda^{1/2}_{\max}\big(S_{\Delta_{\mathbf{F}_{n-1}}}\big) \sqrt{2\pi\mathcal{M}(f_Z) + \pi\mathfrak{m}(f_Z)/27},
		\end{equation}
		and tolerance 
		\begin{equation*}
		\tau_n = \Big(\frac{\pi}{2}\mathfrak{m}(f_Z) + 27\Lambda^{1/2}_{\max}\big(S_{\Delta_{\mathbf{F}_{n-1}}}\big) \sqrt{2\pi\mathcal{M}(f_Z) + \pi\mathfrak{m}(f_Z)/27}\Big)\omega^2\sqrt{\frac{\log (p_1+p_2)}{n}}, 
		\end{equation*}
		where $\omega=54\tfrac{\mathcal{M}(f_Z)}{\mathfrak{m}(f_Z)}$, provided that the sample size $n \gtrsim \log (p_1+p_2)$.
	\end{lemma}
	
	\begin{lemma}[Deviation bound for $\|\widehat{\mathbf{Z}}_{n-1}^\top\big( \widehat{\mathbf{Z}}_n - \widehat{\mathbf{Z}}_{n-1}(A^\star)^\top \big)/n\|_\infty$]\label{lemma:deviation-A} There exist positive constants $c_i~(i=1,2)$ and $C_i~(i=1,2,3)$ such that with probability at least $1-c_1\exp\big(-c_2\log (p_1+p_2)\big)$ we have
		\begin{align}
		C(n,p_1,p_2) \leq  &~C_1\big[\mathcal{M}(f_Z) + \frac{\Lambda_{\max}(\Sigma_w)}{2\pi} + \mathcal{M}(f_{Z,W^+})\big]\sqrt{\frac{\log(p_1+p_2)}{n}} \notag \\
		&+ ~C_2\big[\mathcal{M}^{1/2}(f_Z)\max\limits_{j\in\{1,\dots,p_1\}}\|\Delta_{\mathbf{F}_{n,\cdot j}}/\sqrt{n}\|\big]\sqrt{\frac{\log p_1 + \log(p_1+p_2)}{n}}  \notag \\
		& + ~ C_3\big[ \Lambda^{1/2}_{\max}(\Sigma_w)\max\limits_{j\in\{1,\dots,(p_1+p_2)\}}\|\varepsilon_{n,\cdot j}/\sqrt{n}\| \big]\sqrt{\frac{\log(p_1+p_2)}{n}} \label{eqn:deviationCnp} \\
		&+ ~ \frac{1}{n}\|\Delta_{\mathbf{F}_{n-1}}^\top \Delta_{\mathbf{F}_{n}}\|_\infty + \frac{1}{n}\|\Delta_{\mathbf{F}_{n-1}}^\top \Delta_{\mathbf{F}_{n-1}}(A_{11}^\star)^\top\|_\infty, \notag 
		\end{align}
		where $\varepsilon_n := \Delta_{\mathbf{Z}_{n}} - \Delta_{\mathbf{Z}_{n-1}}(A^\star)^\top = [\Delta_{\mathbf{F}_n}-\Delta_{\mathbf{F}_{n-1}}(A_{11}^\star)^\top, -\Delta_{\mathbf{F}_{n-1}}(A_{21}^\star)^\top ]$, and $\{W^+_t\}:=\{W_{t+1}\}$ is the shifted $W_t$ process.  
	\end{lemma}

	\begin{remark}\label{rmk-1}
		Before moving to the high probability error bounds of the estimates, we discuss the conditions and the various
		quantities appearing in Lemmas~\ref{lemma:RSCZ} and~\ref{lemma:deviation-A} that determine the error bound of the estimated transition matrix and underlie the differences between the original VAR estimation problem based on primal observed quantities (the ``vanilla VAR problem" henceforth), and the present one in which one block of the variables enters the VAR system with errors. Note that the statements in the two lemmas are under the assumption that the error in the $F_t$ block 
		is pre-determined and non-random. 
		
		As previously mentioned, due to the presence of the error of the latent factor block, the corresponding regularity conditions need to be imposed and verified on quantities with the error incorporated, namely, $\widehat{\mathbf{Z}}$, instead of the original true random realizations $\mathbf{Z}$. Lemma~\ref{lemma:RSCZ} shows that with high probability, the random design matrix although exhibits error-in-variables, will still satisfy the RSC condition with some positive curvature as long as the spectrum of the process $Z_t$ has sufficient regularity relative to the magnitude of the error, with the former determined by $\mathfrak{m}(f_X)/\mathcal{M}^{1/2}(f_X)$ and the latter by $\Lambda^{1/2}_{\max}(S_{\Delta_{\mathbf{F}_{n-1}}})$. In particular, the RSC curvature is pushed toward zero compared with that in the vanilla VAR problem, due to the presence of the second term in~\eqref{eqn:curvatureZ} that would be~0 if $\Delta_{\mathbf{F}_{n-1}}=0$, i.e., there were no estimation errors. This curvature affects the constant scalar part of the ultimate high probability error bound obtained for the transition matrix. 
		
		Lemma~\ref{lemma:deviation-A} gives the deviation bound associated with the Hessian and the gradient (both random), which comprises of three components attributed to the random samples observed, the non-random error, and their interactions, respectively. Further, it is the relative order of these components that determines the error rate (as a function of model dimensions and the sample size). In particular, for the vanilla VAR problem, only the first term in~\eqref{eqn:deviationCnp} exists and yields an error rate of $\mathcal{O}(\sqrt{\log(p_1+p_2)/n})$ \citep[see also][]{basu2015estimation}. For the current setting, as it is later shown in Theorem~\ref{thm:info}, since $\smalliii{\Delta_{\mathbf{F}}/\sqrt{n}}_\F\asymp \mathcal{O}(1)$, the dominating term of the three components is the one attributed to the non-random error\footnote{with the implicit assumption that $\log(p_1+p_2)/n = o(1)$ which is satisfied for this study.} and it ultimately determines the error rate of $\widehat{A}$, which will also be $\mathcal{O}(1)$. 
	\end{remark}
	
	\subsection{High probability error bounds for the estimators}\label{sec:theory-high}
	
	Given the results in Sections~\ref{sec:theory-fix} and~\ref{sec:theory-random}, we provide next high probability error bounds for the estimates, obtained by solving the optimization problems in~\eqref{opt:solveTheta} and~\eqref{opt:solveA} based on random snapshots from the underlying processes $X_t$ and $Y_t$.
	
	Theorem~\ref{thm:info} combines the results in Proposition~\ref{prop:Theta-Gamma-fixed} and Lemmas~\ref{lemma:RSCX} to~\ref{lemma:Emax} and provides the high probability error bound of the estimates, when $\widehat{\Theta}$ and $\widehat{\Gamma}$ are estimated based on random realizations from the observable processes $X_t$ and $Y_t$, with the latter driven by both $X_t$ and the latent $F_t$.  
	\begin{theorem}[High probability error bounds for $\widehat{\Theta}$ and $\widehat{\Gamma}$]\label{thm:info} 
		Suppose we are given some randomly observed snapshots $\{x_1,\dots,x_n\}$ and $\{y_1,\dots,y_n\}$ obtained from the stable Gaussian processes $X_t$ and $Y_t$, whose dynamics are described in~\eqref{md:FAVAR-VAR1} and~\eqref{md:FAVAR-info}. Suppose the following conditions hold for some $(C_{X,l},C_{X,u})$ and $(C_{e,l},C_{e,u})$:
		\begin{itemize}
			\item[C1.] $C_{X,l}\leq \mathfrak{m}(f_X)\leq \mathcal{M}(f_X)\leq C_{X,u}$;
			\item[C2.] $C_{e,l}\leq \Lambda_{\min}(\Sigma_e)\leq \Lambda_{\max}(\Sigma_e)\leq C_{e,u}$.
		\end{itemize}
		Then, there exist universal constants $\{C_i\}$ and $\{c_i\}$ such that for sample size $n\gtrsim q$, by solving~\eqref{opt:solveTheta} with regularization parameter 
		\begin{equation}\label{eqn:gamma_n}
		\begin{split}
		\lambda_\Gamma = \max\Big\{ C_1(2\pi\mathcal{M}(f_X)+\Lambda_{\max}(\Sigma_e))\sqrt{\frac{\log(p_2q)}{n}},
		~C_2\phi/\sqrt{nq},~C_3\Lambda^{1/2}_{\max}(\Sigma_e)  \Big\},
		\end{split}
		\end{equation}
		the solution $(\widehat{\Theta},\widehat{\Gamma})$ has the following bound with probability at least $1-c_1\exp(-c_2\log(p_2q))$:
		\begin{equation}\label{eqn:err1}
		\smalliii{\Delta_{\Theta}/\sqrt{n}}_{\F}^2 +\smalliii{\Delta_{\Gamma}}_{\F}^2 \lesssim \frac{\lambda^2_\Gamma}{\mathfrak{m}(f_X)} \psi(s_{\Gamma^\star},p_1,r)=:K_1,
		\end{equation}
		for some function $\psi(\cdot)$ that depends linearly on $s_{\Gamma^\star}, p_1$ and $r$.
	\end{theorem}
	
	Note that the above bound also holds if we replace $\Delta_{\Theta}$ by $\Delta_{\mathbf{F}}$ under (IR). Next, using the results in Proposition~\ref{prop:A-fixed}, Lemmas~\ref{lemma:RSCZ} and~\ref{lemma:deviation-A} and combine the bound in Theorem~\ref{thm:info}, we establish a high probability error bound for the estimated $\widehat{A}$ in Theorem~\ref{thm:A}.
	\begin{theorem}[High probability error bound for $\widehat{A}$]\label{thm:A} Under the settings and with the procedures in Theorem~\ref{thm:info}, we additionally assume the following condition holds for the spectrum of the joint process $Z_t$:
		\begin{itemize}
			\item[C3.] $\mathfrak{m}(f_Z)/\mathcal{M}^{1/2}(f_Z) > C_Z$ for some constant $C_Z$.
		\end{itemize}
		Then there exists universal constants $\{c_i\}$, $\{c'_i\}$ and $\{C_i\}$ such that for sample size $n\gtrsim q$, such that the estimator $\widehat{A}$ obtained by solving for~\eqref{opt:solveA} with $\lambda_A$ satisfying
		\begin{equation*}
		\begin{split}
		\lambda_A = &~C_1\big(\mathcal{M}(f_Z)+\frac{\Sigma_w}{2\pi}+\mathcal{M}(f_{Z,W^+})\big)\sqrt{\frac{\log(p_1+p_2)}{n}} \\
		&+ C_2\mathcal{M}^{1/2}(f_Z)\sqrt{\frac{\log(p_1+p_2)+\log p_1}{n}} 
		~+ C_3\Lambda^{1/2}_{\max}(\Sigma_w)\sqrt{\frac{\log(p_1+p_2)}{n}} + C_4,
		\end{split}
		\end{equation*}
		with probability at least 
		\begin{equation}\label{eqn:prob}
		\Big(1-c_1\exp\{-c_2\log(p_2q)\}\Big)\Big(1-c'_1\exp\{-c'_2\log(p_1+p_2)\}\Big),
		\end{equation}
		the following bound holds for $\Delta_A$:
		\begin{equation*}
		\smalliii{\Delta_A}_{\F}^2 \leq \check{C}(K_1,\mathfrak{m}(f_Z),\mathcal{M}(f_Z))\cdot \check{\psi}(s_{A^\star}),
		\end{equation*}
		for some function $\check{C}(K_1,\mathfrak{m}(f_Z),\mathcal{M}(f_Z))$ that does not depend on $n,p_2,q$ and $\check{\psi}(\cdot)$ that depends linearly on $s_{A^\star}$. Here $K_1$ denotes the upper bound of the first stage error shown in~\eqref{eqn:err1}.
	\end{theorem}

	\begin{remark}[Rate of convergence] It is worth pointing out similarities in the formulation of the calibration equation and a matrix completion problem. 
		Note that the factor hyperplane corresponds to the low-rank component one seeks to recover in the latter problem in a noisy
		setting. Hence, the resulting similarity in the rate obtained in our setting to that established for the matrix completion
		problem \citep{candes2010matrix}, is a consequence of absence of the restricted isometry property (RIP) \citep[see also][]{gunasekar2015unified}. 
	\end{remark}

	\begin{remark}[Sample size requirement] To establish the finite-sample high probability error bound for the estimated transition matrices $\widehat{A}$, the proposed estimation procedure requires the sample size to satisfy $n\gtrsim q$; this condition is more stringent compared to the standard VAR estimation problem under sparsity, given by $n\gtrsim\sqrt{\log(p_1+p_2)}$. However, this is due to the fact that in the FAVAR formulation the $F_t$ block is latent and needs to be estimated from the data and hence comes with ``measurement error". The more  restrictive sample size requirement reflects the latter fact and is embedded in the factor recovery step in the calibration equation -- specifically, the concentration of $\Lambda_{\max}(S_{\mathbf{E}})$ that is necessary for providing adequate control over $\Delta_{\mathbf{F}}$.
	\end{remark}
	
	\begin{remark}[Generalization to $\VAR(d)$] As a straightforward generalization, for a $\VAR(d), d>1$ system $Z_t = (F^\top_t,X^\top_t)^\top$, a similar error bound holds by considering the augmented process $\widetilde{Z}^\top_t:=(Z_t,Z_{t-1},\dots,Z_{t-d+1})$ that satisfies 
		\begin{equation*}
		\widetilde{Z}_t = \widetilde{A}\widetilde{Z}_{t-1} + \widetilde{W}_t, \qquad \text{where}\qquad \widetilde{A}:=\left[\begin{smallmatrix}
		A^{(1)} & A^{(2)} & \cdots & A^{(d)} \\ \I_{p} & O & O & O \\ \vdots & \ddots & \vdots & \vdots \\ O & O & \I_{p} &   O 
		\end{smallmatrix}\right], \quad \widetilde{W}_t= \left[\begin{smallmatrix}
		W_t \\ 0 \\ \vdots \\ 0
		\end{smallmatrix}\right].
		\end{equation*}
		In particular, with probability at least $$\big(1-c_1\exp\{-c_2\log(p_2q)\}\big)\big(1-c'_1\exp\{-c'_2\log(d(p_1+p_2))\}\big),$$
		the following bound holds for the estimate of $\widetilde{A}$:
		\begin{equation*}
		\smallii{\Delta_{\widetilde{A}}}_{\F}^2 \leq \widetilde{C}(K_1,\mathfrak{m}(f_{\widetilde{Z}}),\mathcal{M}(f_{\widetilde{Z}}))\cdot \widetilde{\kappa}(s_{\widetilde{A}^\star}).
		\end{equation*}
		However, note that although the error bound is still of the same form, the stronger temporal dependence yields a larger $\widetilde{C}(K_1,\mathfrak{m}(f_{\widetilde{Z}}),\mathcal{M}(f_{\widetilde{Z}})$ through the RSC curvature parameter; specifically, a smaller value of $\mathfrak{m}(f_{\widetilde{Z}})$. Its impact on the deviation bound will not manifest itself in terms of the order of the error, since it only affects the constants in front of lower order terms in the expression of choosing $\lambda_A$. 
	\end{remark}
	
	\section{Implementation and Performance Evaluation} \label{sec:simulation}
	
	We first discuss implementation issues of the proposed problem formulation for the high-dimensional FAVAR model.
	Specifically, the formulation requires imposing the compactness constraint for identifiability purposes and for
	obtaining the necessary statistical guarantees for the estimates of the model parameters. However, the value $\phi$
	in the compactness constraint is hard to calibrate in any real data set. Hence, in the implementation we relax this
	constraint and assess the performance of the algorithm. Due to its importance in constraining the size of the
	equivalence class $\mathcal{C}(Q_2)$, we examine in Appendix~\ref{appendix:fail} certain relatively extreme settings where the proposed relaxation fails to provide accurate estimates of the model parameters.
	
	\paragraph{Implementation.} The following relaxation of~\eqref{opt:solveTheta} is used in practice:
	\begin{equation}\label{opt:solveTheta-empirical}
	\begin{split}
	\min\limits_{\Theta,\Gamma} f(\Theta,\Gamma)&:=\big\{\frac{1}{2n}\smalliii{\mathbf{Y}-\Theta-\mathbf{X}\Gamma^\top}_{\F}^2  + \lambda_\Gamma\|\Gamma\|_1\big\} \\
	& \text{subject to}~~~\text{rank}(\Theta)\leq r,
	\end{split}
	\end{equation}
	which leads to Algorithm~\ref{algo}. 
	\begin{algorithm}[ht]
		\small
		\setstretch{0.01}
		\caption{Computational procedure for estimating $A$, $\Gamma$ and $\Lambda$.}\label{algo}
		\KwIn{Time series data $\{x_i\}_{i=1}^n$ and $\{y_i\}_{i=1}^n$, $(\lambda_\Gamma,r)$, and $\lambda_A$.}
		\BlankLine
		\textbf{Stage I:} recover the latent factors by solving~\eqref{opt:solveTheta-empirical}, through iterating between (1.1) and  (1.2) until $|f(\Theta^{(m)},\Gamma^{(m)}) - f(\Theta^{(m-1)},\Gamma^{(m-1)})|< \text{tolerance}$:
		\BlankLine
		-- (1.1) Update $\widehat{\Theta}^{(m)}$ by singular value thresholding (SVT): do SVD on the residual hyperplane, i.e., $\mathbf{Y} - \mathbf{X}(\widehat{\Gamma}^{(m-1)})^\top = UDV$, where $D := \text{diag}\big(d_1,\dots,d_{\min(n,q)}\big)$, and construct $\widehat{\Theta}^{(m)}$~by 
		$$
		\widehat{\Theta}^{(m)} = UD_rV, \qquad \text{where}~~~D_r:=\text{diag}\big(d_1,\dots,d_r,0,\dots,0\big).
		$$
		\BlankLine
		-- (1.2) Update $\widehat{\Gamma}^{(m)}$ with the plug-in $\widehat{\Theta}^{(m)}$ so that each row $j$ is obtained with Lasso regression (in parallel) and solves
		\begin{equation*}
		\min\nolimits_\beta \big\{ \frac{1}{2n}\smallii{(\mathbf{Y} - \widehat{\Theta}^{(m)})_{\cdot j} - \mathbf{X}\beta}^2  + \lambda_A \|\beta_1\|_1  \big\}
		\end{equation*}
		\BlankLine
		\underline{Stage I output:} $\widehat{\Theta}$ and $\widehat{\Gamma}$; the estimated factor $\widehat{\mathbf{F}}$ and  $\widehat{\Lambda}$ via~\eqref{eqn:F-Lambda} under (IR)\;
		\BlankLine
		\textbf{Stage II:} estimate the transition matrix by solving~\eqref{opt:solveA}: update each row of $A$ (in parallel) by solving the Lasso problem:
		\begin{equation*}
		\min\nolimits_{\beta} \big\{\frac{1}{2n} \smallii{(\widehat{\mathbf{Z}}_n)_{\cdot j} -\widehat{\mathbf{Z}}_{n-1}\beta }^2 + \lambda_A\|\beta\|_1\big\}.
		\end{equation*}
		\BlankLine
		\underline{Stage II output:} $\widehat{A}$.
		\BlankLine
		\KwOut{Estimates $\widehat{\Gamma}$, $\widehat{\Lambda}$, $\widehat{A}$ and the latent factor $\widehat{\mathbf{F}}$.}
	\end{algorithm}
	The implementation of Stage I requires the pair of tuning parameters $(\lambda_\Gamma,r)$ as input, and the choice of $r$ is particularly critical since it determines the effective size of the latent block. In our implementation, we select the optimal pair based on the Panel Information Criterion (PIC) proposed in \citet{ando2015selecting}, which searches for
	$(\lambda_\Gamma,r)$ over a lattice that minimizes
	\begin{equation*}
	\text{PIC}(\lambda_\Gamma,r) := \frac{1}{nq}\vertiii{\mathbf{Y}-\widehat{\Theta} - \mathbf{X}\widehat{\Gamma}^\top}_{\F}^2 + \widehat{\sigma}^2 \Big[ \frac{\log n}{n}\|\widehat{\Gamma}\|_0 + r(\frac{n+q}{nq})\log(nq)\Big],
	\end{equation*}
	where $\widehat{\sigma}^2 = \tfrac{1}{nq}\smalliii{\mathbf{Y}-\widehat{\Theta} - \mathbf{X}\widehat{\Gamma}^\top}_{\F}^2$. Analogously, the implementation of Stage II requires $\lambda_A$ as input, and we select $\lambda_A$ over a grid of values that minimizes the Bayesian Information Criterion (BIC):
	\begin{equation*}
	\text{BIC}(\lambda_A) = \sum_{i=1}^q \log \text{RSS}_i + \frac{\log n}{n}\|\widehat{A}\|_0, 
	\end{equation*}
	where $\text{RSS}_i:= \|(\mathbf{X}_{n})_{\cdot i}-\mathbf{X}_{n-1}\widehat{A}^\top_{i\cdot}\|^2$ is the residual sum of square of the $i^{\text{th}}$ regression. Extensive numerical work shows that these two criteria select very satisfactory values for the tuning parameters, which in turn yield highly accurate estimates of the model parameters.
	
	\paragraph{Simulation setup.} Throughout, we assume $\Sigma_w^X$, $\Sigma_X^F$ and $\Sigma_e$ are all diagonal matrices, and the sample size is fixed at 200, unless otherwise specified. We first generate samples of $F_t\in\R^{p_1}$ and $X_t\in\R^{p_2}$ recursively according to the $\VAR(d)$ model in~\eqref{md:FAVAR-VAR}, and then the samples of $Y_t\in\R^q$ are generated according to the linear model given in~\eqref{md:FAVAR-info}. Specifically, (IR) is imposed on the true value of the parameter, hence $\Lambda^\star$ that is used for generating $Y_t$ always satisfies the restriction $\Lambda=\left[\begin{smallmatrix}\I_{p_1} \\ * \end{smallmatrix}\right]$. Unless otherwise specified, all error terms are generated according to some mean-zero Gaussian distribution.
	
	For the calibration equation, the density level of the sparse coefficient matrix $\Gamma\in\R^{q\times p_2}$ is fixed at $5/p_2$ for each regression;
	thus, each $Y_t$ coordinate is affected by 5 series (coordinates) from the $X_t$ block on average. The bottom $(q-p_1)\times p_1$ block of the loading matrix $\Lambda\in\R^{q\times p_1}$ is dense. The magnitude of nonzero entries of $\Gamma$ and that of entries of $\Lambda$ may vary to capture different levels of signal contributions to $Y_t$, and we adjust the standard deviation of $e_t$ to maintain the desired level of the signal-to-noise ratio for $Y_t$ (averaged across all coordinates). 
	
	For the transition matrix $A$ of the VAR equation, the density for each of its component block $\{A_{ij}\}_{i,j=1,2}$ varies
	across settings, so as to capture different levels of the influence from the lagged value of the latent block $F_{t}$ on the observed $X_t$. Note that to ensure stability of the VAR system, the spectral radius of $A$, $\varrho(A)$, needs to be smaller than 1. In particular, when a $\VAR(d)~(d>1)$ system is considered, we need to ensure that the spectral radius of $\widetilde{A}$ is smaller than 1\footnote{In practice, this can be achieved by first generating $A^{(1)},\dots,A^{(d)}$, align them in $\widetilde{A}_{\text{initial}}$ and obtain the scale factor $\zeta:=\varrho_{\text{target}}/\varrho(\widetilde{A}_{\text{initial}})$, then scale $A^{(i)}$ by $\zeta^i$. The validity of this procedure follows from simple algebraic manipulations.}, where we let $p=p_1+p_2$ and 
	\begin{equation*}
	\widetilde{A}:=\left[\begin{matrix}
	A^{(1)} & A^{(2)} & \cdots & A^{(d)} \\ \I_{p} & O & O & O \\ \vdots & \ddots & \vdots & \vdots \\ O & O & \I_{p} &   O 
	\end{matrix}\right]. 
	\end{equation*}
	Table~\ref{table:sim-setup} lists the simulation settings and their parameter setup.
	\begin{table}[!ht]
		\captionsetup{font=small}
		\fontsize{8}{9.5}\selectfont
		\centering
		\caption{Parameter setup for different simulation settings for the VAR equation.}\label{table:sim-setup} 
		\begin{tabular}{l|lll|lllll|l}
			\specialrule{.1em}{.05em}{.05em} 
			& $q$ & $p_1$ & $p_2$ &  &$s_{A_{11}}$ & $s_{A_{12}}$ & $s_{A_{21}}$ & $s_{A_{22}}$ & $\text{SNR}(Y_t)$ \\  \hline
			A1 & 100 & 5 & 50 & & \multicolumn{4}{c|}{$3/(p_1+p_2)$}  & 1.5\\  \hline
			A2 & 200 & 10 & 100 &  &\multicolumn{4}{c|}{$3/(p_1+p_2)$}  & 1.5\\  \hline
			A3 & 200 & 5 & 100 & & $3/p_1$ & $2/p_2$ & $2/p_1$ & $2/p_2$ & 1.5\\  \hline
			A4 & 300 & 5 & 500 & &$3/p_1$ & $2/p_2$ & $0.8$ & $2/p_2$ & 1.5 \\ \hline
			B1 & \multirow{2}{0.8cm}{$200$} & \multirow{2}{0.8cm}{$5$} & \multirow{2}{0.8cm}{$100$} & $A^{(1)}:$ & \multicolumn{4}{c|}{$3/(p_1+p_2)$} &  \multirow{2}{1cm}{2}\\ 
			$(d=2)$	&   & & & $A^{(2)}:$  &\multicolumn{4}{c|}{$2/(p_1+p_2)$} & \\ \hline
			& \multirow{4}{0.8cm}{$200$} & \multirow{4}{0.8cm}{$5$} & \multirow{4}{0.8cm}{$100$} & $A^{(1)}:$ & $0.5$ & $3/p_2$ & $0.5$ & $3/p_2$ & \multirow{4}{1cm}{2}\\
			B2	&   & & &  	$A^{(2)}:$&		  $0.2$ & $2/p_2$ & $0.25$ & $2/p_2$ \\
			$(d=4)$	&  & & & $A^{(3)}:$ &		 \multicolumn{4}{c|}{$2/(p_1+p_2)$} \\
			& & & & $A^{(4)}:$ & \multicolumn{4}{c|}{$2/(p_1+p_2)$}\\ \hline
			& \multirow{4}{0.8cm}{$100$} & \multirow{4}{0.8cm}{$5$} & \multirow{4}{0.8cm}{$25$} & $A^{(1)}:$ &$0.5$ & $2/p_2$ & $0.5$ & $2/p_2$ & \multirow{4}{1cm}{2}\\
			B3	&   & & &  $A^{(2)}:$ &	$0.2$ & $1.5/p_2$ & $0.1$ & $1.5/p_2$ \\
			$(d=4)$	&  & & & 	$A^{(3)}:$ & \multicolumn{4}{c|}{$1/(p_1+p_2)$} \\
			& & & & $A^{(4)}:$ &  \multicolumn{4}{c|}{$0.8/(p_1+p_2)$}\\ \hline
			C1 & \multicolumn{9}{l}{same as setting A1 with $t_4$ noise for the VAR system} \\ 
			C2 & \multicolumn{9}{l}{same as setting B1 with $t_8$ noise for the VAR system} \\ 
			C3 & \multicolumn{9}{l}{same as setting B2 with sub-exponential noise for the VAR system} \\ 
			C4 & \multicolumn{9}{l}{same as setting B2 with sub-exponential noise for the VAR system and 500 observations} \\ 
			\specialrule{.1em}{.05em}{.05em} 
		\end{tabular}
	\end{table} 
	
	Specifically, in settings A1\,--\,A4, $(F^\top_t,X^\top_t)^\top$ jointly follows a $\VAR(1)$ model. The (average) signal-to-noise ratio for each regression of $Y_t$ is 1.5. For settings A1 and A2, the transition matrix $A$ is uniformly sparse, with A2 corresponding to a larger system; for settings A3 and A4, we increase the density level (the proportion of nonzero entries) for the transition matrices that govern the effect of $F_{t-1}$ on $F_{t}$ and $X_{t}$. In particular, for setting A4, we consider a large system with 500 coordinates in $X_t$, and the factor effect is almost pervasive on these coordinates (through the lags), as the density level of $A_{21}$ is set at 0.8. Settings B1, B2 and B3 consider settings with more lags ($d=2$ and $d=4$, respectively), and to compensate for the higher level of correlation between $F_t$ and $X_t$, we elevate the signal-to-noise for each regression of $Y_t$ to 2. For B1, the transition matrices for both lags ($A^{(1)}$ and $A^{(2)}$) have uniform sparsity patterns, with $A^{(2)}$ being slightly more sparse compared to $A^{(1)}$; for B2, the transition matrices for the first two lags have higher density in the component that governs the $F_{t-i}\rightarrow X_t$ cross effect, and those for the last two lags have uniform sparsity. B3 has approximately the same scale as observed in real data, and due to a small $p_2$, the system exhibits a higher sparsity level in general. In settings C1\,--\,C4, the error terms of the VAR system are generated from distributions with tails heavier than a Gaussian (e.g. $t$-distributions, squares of Gaussian which
	have sub-exponential tails), and the joint process $(F_t',X_t)'$ will be heavy-tailed as a result of the recursive data generating mechanism.
	
	\paragraph{Performance evaluation.} We consider both the estimation and the forecasting performance of the proposed estimation procedure. The performance metrics used for estimation are sensitivity (SEN), specificity (SPC) and the relative error in Frobenius norm (Err) for the sparse components (transition matrices $A$ and the coefficient matrix $\Gamma$), defined as
	\begin{equation*}
	\textrm{SEN} = \frac{\textrm{TP}}{\textrm{TP}+\textrm{FN}},\quad \textrm{SPC} = \frac{\textrm{TN}}{\textrm{FP}+\textrm{TN}}, \quad \textrm{Err} = \smalliii{\Delta_M}_\F/\smalliii{M^\star}_{\F}.
	\end{equation*}
	We also track the estimated size of the latent component (i.e., the rank constraint in~\eqref{opt:solveTheta}, jointly with $\lambda_{\Gamma}$ is selected by PIC), as well as the relative errors of $\widehat{\Theta}$, $\widehat{\mathbf{F}}$ and $\widehat{\Lambda}$. For forecasting, we focus on evaluating the $h$-step-ahead predictions for the $X_t$ block. Specifically, for settings where the VAR system is 1-lag dependent (A1--A4, C1), we consider $h=1$; for settings where the VAR system has more lag dependencies (B1--B3, C2--C4), we consider $h=1,2$. We use the same benchmark model as in 
	\citet{banbura2010large} which is based on a special case of the Minnesota prior distribution 
	\citep{litterman1986forecasting}, so that the for any generic time series $X_t\in\R^p$, each of its coordinates $j=1,\dots,p$ follows a centered random walk:
	\begin{equation}\label{md:benchmark}
	X_{t,j} = X_{t-1,j} + u_{t,j}, \qquad u_{t,j}\sim \mathcal{N}(0,\sigma^2_u).
	\end{equation}
	For each forecast $\widehat{x}_{T+h}$, its performance is evaluated based on the following two measures: 
	\begin{equation*}
	\text{rel-err} = \|\widehat{x}_{T+h}-x_{T+h}\|_2^2/\|x_{T+h}\|_2^2, \qquad  \text{rel-err-ratio} = \dfrac{ \tfrac{1}{p_2}\sum_{j=1}^{p_2} \big| \frac{\widehat{x}_{T+h,j}-x_{T+h,j}}{x_{T+h,j}}  \big|}{\tfrac{1}{p_2}\sum_{j=1}^{p_2} \big| \frac{\widetilde{x}_{T+h,j}-x_{T+h,j}}{x_{T+h,j}}  \big|},
	\end{equation*}
	where $\text{rel-err}$ measures the $\ell_2$ norm of the relative error of the forecast to the true value;
	whereas for $\text{rel-err-ratio}$, it measures the ratio between the relative error of the forecast and the above described benchmark. In particular, its numerator and denominator respectively capture the averaged relative error of all coordinates of the forecast $\widehat{x}_{T+h}$ and that of the benchmark $\widetilde{x}_{T+h}$ that evolves according to~\eqref{md:benchmark}, while the ratio measures how much the forecast based on the proposed FAVAR model outperforms $(<1)$ or under-performs $(>1)$ compared to the benchmark. 
	
	All tabulated results are based on the average of 50 replications. Tables~\ref{table:sim-info-result}, \ref{table:sim-VAR-result} and~\ref{table:sim-forecast}, respectively, depict the performance of the estimates of the parameters in the calibration
	and the VAR equations, as well as the forecasting performance under the settings considered. 
	\begin{table}[!pht]
		\captionsetup{font=small}
		\fontsize{8}{9.5}\selectfont
		\centering
		\caption{Performance evaluation of estimated parameters in the calibration equation.}\vspace*{-2mm}\label{table:sim-info-result}
		\begin{tabular}{c|c|ccc|ccc}
			\specialrule{.1em}{.05em}{.05em} 
			& PIC-selected $r$ & $\Err(\widehat{\Theta})$ &  $\Err(\widehat{\mathbf{F}})$  &  $\Err(\widehat{\Lambda})$ & $\SEN(\widehat{\Gamma})$ & $\SPC(\widehat{\Gamma})$ & $\Err(\widehat{\Gamma})$  \\ \hline 
			A1 &  4.80(.40) &  0.32(.010) & 0.56(.074) & 0.67(.345) & 0.99(.007) & 0.98(.003) & 0.45(.013) \\
			A2 &  9.96(.19) & 0.32(.008) & 0.90(.065) & 2.54(1.30) & 0.99(.005) & 0.98(.001) & 0.52(.010) \\
			A3 & 4.78(.54) & 0.33(.048) & 0.73(.103) & 2.59(1.59) & 0.99(.003) & 0.99(.001) & 0.57(.009)\\
			A4 & 4.42(.49) & 0.38(.040) & 0.84(.100) & 2.66(2.14) & 0.97(.009) & 0.99(.001) & 0.59(.015) \\ \hline
			B1 & 5(0) & 0.23(.004) & 0.41(.043) & 0.54(.020) & 1.00(.000) & 0.97(.011) & 0.27(.014) \\  \hline
			B2 & 5(0) & 0.26(.007) & 0.38(.047) & 0.42(.087) & 1.00(.000) & 0.99(.002) & 0.37(.007) \\ 
			B3 &  5(0) & 0.25(.007) & 0.34(.031) & 0.34(.080) & 1.00(.000) & 0.99(.001) & 0.32(.012) \\ \hline
			C1 & 4.96(.20) & 0.32(.019) & 0.58(.075) & 0.86(.564) & 0.99 (.001) & 0.96(.009) & 0.47(.017) \\ \hline 
			C2 & 5(0) & 0.23(.005) & 0.43(.042) & 0.54(.155) & 1.00 (.000) & 0.96(.008) & 0.27(.010) \\  \hline
			C3 & 5(0) & 0.21(.006) & 0.39(.040) & 0.41(.123) & 1.00 (.000) & 0.97(.003) & 0.27(.052)  \\
			C4 & 5(0) & 0.20(.007) & 0.27(.028) & 0.25(.041) & 1.00 (.000) & 0.97(.011) & 0.18(.012)  \\
			\specialrule{.1em}{.05em}{.05em} 
		\end{tabular}
		
		\bigskip
		\caption{Performance evaluation of the estimated transition matrices in the VAR equation.}\vspace*{-2mm}\label{table:sim-VAR-result}
		\begin{tabular}{c|c|ccc|ccc}
			\specialrule{.1em}{.05em}{.05em} 
			& coef & $\SEN(\widehat{A})$ & $\SPC(\widehat{A})$ & $\Err(\widehat{A})$ & $\SEN(\widehat{A}_{22})$ & $\SPC(\widehat{A}_{22})$ & $\Err(\widehat{A}_{22})$\\ \hline 
			A1 & $A$ & 0.99(.003) & 0.95(.012) & 0.35(.019) & 0.99(.001) & 0.96(.013) & 0.31(.022) \\
			A2 & $A$ & 0.98(.008) & 0.97(.004) & 0.46(.018) & 0.99(.001) & 0.98(.003) & 0.39(.017)  \\
			A3 & $A$ & 0.86(.050) & 0.98(.006) & 0.73(.029) & 0.93(.032) & 0.98(.005) & 0.65(.034) \\
			A4 & $A$ & 0.75(.046) & 0.92(.002) & 0.71(0.024) & 0.99(.001) & 0.92(.002) & 0.60(.018)  \\   \hline
			B1 & $A^{(1)}$ & 0.99(.003) & 0.98(.002) & 0.47(.017) & 0.99(.002) & 0.98(.002) & 0.46(.017) \\ 
			& $A^{(2)}$ & 0.97(.010) & 0.98(.002) & 0.55(.017) & 0.98(.011) & 0.98(.003) & 0.55(.018)  \\ \hline 
			B2 & $A^{(1)}$ & 0.89(.017) & 0.88(.003) & 0.71(.014) & 0.90(.017) & 0.99(.003) & 0.70(.014)  \\
			& $A^{(2)}$ & 0.75(.028) & 0.88(.003) & 0.89(.020) & 0.77(0.032) & 0.88(.003) & 0.90(.021)  \\ 
			&$A^{(3)}$ & 0.84(.025) & 0.88(.003) & 0.85(.015) & 0.85(.027) & 0.88(.004) & 0.84(.018) \\
			& $A^{(4)}$ & 0.72(.022)&  0.88(.003) & 0.99(.017) & 0.73(.025) & 0.88(.003) &  0.98(.017)  \\ \hline
			B3 & $A^{(1)}$ & 0.93(.034) & 0.96(.010) & 0.61(.043) & 0.94(.035) & 0.97(.009) & 0.60(.045) \\
			& $A^{(2)}$ & 0.77(.078) & 0.96(.010) & 0.74(.044) & 0.78(.084) & 0.97(.010) & 0.74(.046) \\
			& $A^{(3)}$ & 0.80(.098) & 0.96(.012) & 0.75(.052) & 0.81(.102) & 0.97(.010) & 0.74(.056)  \\
			& $A^{(4)}$ & 0.74(.122) & 0.97(.011) & 0.78(.059) &  0.72(.134) & 0.97(.009) & 0.79(.065)   \\ \hline
			C1 & $A$	& 0.99(.007) & 0.95(.012) & 0.42(.024) & 0.99(.002) & 0.96(.011) & 0.38(.024) \\ 
			C2 & $A^{(1)}$ & 0.99(.004) & 0.98(.002) & 0.46(.013) & 0.99(.003) & 0.98(.002) & 0.45(.015) \\
			& $A^{(2)}$ & 0.98(.008) & 0.97(.003) & 0.54(.018) & 0.98(.009) & 0.98(.003) & 0.54(.019) \\    \hline
			C3 & $A^{(1)}$ & 0.93(.013) & 0.42(.005) & 1.54(.024) & 0.93(.013) & 0.42(.006) & 1.61(.027)  \\
			& $A^{(2)}$ & 0.86(.019) & 0.44(.006) & 2.11(.029) & 0.86(.023) & 0.44(.006) & 2.30(.032)  \\ 
			&$A^{(3)}$ & 0.88(.023) & 0.44(.006) & 2.06(.028) & 0.89(.023) & 0.44(.005) & 2.07(.028) \\
			& $A^{(4)}$ & 0.82(.023)&  0.44(.006) & 2.51(.043) & 0.83(.025) & 0.44(.006) &  2.51(.041)  \\  
			C4 & $A^{(1)}$ & 0.89(.016) & 0.96(.002) & 0.67(.013) & 0.89(.016) & 0.96(.002) & 0.65(.014)  \\
			& $A^{(2)}$ & 0.73(.025) & 0.96(.006) & 0.78(.029) & 0.74(.026) & 0.96(.002) & 0.79(.011)  \\ 
			&$A^{(3)}$ & 0.82(.027) & 0.96(.002) & 0.74(.015) & 0.82(.028) & 0.96(.002) & 0.74(.017) \\
			& $A^{(4)}$ & 0.60(.031)&  0.96(.002) & 0.87(.014) & 0.60(.033) & 0.96(.002) &  0.87(.015)  \\  
			\specialrule{.1em}{.05em}{.05em}    
		\end{tabular}
		
		\bigskip
		\caption{Performance evaluation of forecasting.}\vspace*{-2mm}\label{table:sim-forecast}
		\begin{tabular}{c|cccc}
			\specialrule{.1em}{.05em}{.05em} 
			& \multicolumn{2}{c}{$h=1$} & \multicolumn{2}{c}{$h=2$}  \\
			\cmidrule(r){2-3} \cmidrule(l){4-5}
			& rel-err & rel-err-ratio & rel-err & rel-err-ratio \\ \hline
			A1 & 0.53(.117) &  0.38(.065) & - & - \\
			A2 & 0.60(.075) & 0.38(.046) & - & -\\
			A3 & 0.80(.075) &  0.45(.064) &  - & -\\
			A4 &  0.56(.109) & 0.40(.055) & - & -\\ \hline
			B1 &  0.62(.060) & 0.35(.171) & 0.66(.127) & 0.24(.071) \\
			B2 & 0.89(.091) & 0.42(.217) & 0.94(.173) & 0.29(.118) \\
			B3 & 0.81(.094) & 0.32(.129) & 0.90(.402) & 0.26(.174) \\ \hline
			C1 & 0.59(.176) & 0.50(.118)& - & - \\
			C2 & 0.59(.121) & 0.41(.350) & 0.61(.270) & 0.26(.089) \\ \hline
			C3 & 1.25(.305) & 0.19(.081) & 1.26(.396) & 0.15(.059) \\
			C4 & 0.52(.073) & 0.12(.071) & 0.51(.168) & 0.07(.034) \\
			\specialrule{.1em}{.05em}{.05em} 
		\end{tabular}	
	\end{table}
	Based on the results listed in Tables~\ref{table:sim-info-result} and~\ref{table:sim-VAR-result}, we notice that in all settings, the parameters in the calibration equation $\widehat{\Theta}$ and $\widehat{\Gamma}$ are well estimated, while
	the rank slightly underestimated. Further, the SEN and SPC measures of $\widehat{\Gamma}$ show excellent performance regarding
	support recovery. It is worth pointing out that the estimation accuracy of the parameters in the calibration equation 
	strongly depends on the signal-to-noise ratio of $Y_t$. In particular, if the signal-to-noise ratio in A1-A4 is increased 
	to 1.8, the rank is always correctly selected by PIC, and the estimation relative error of $\widehat{\Theta}$ further decreases(results omitted for space considerations)\footnote{This also comes up when comparing the relative error of $\widehat{\Theta}$ in the A1-A4 settings to that in the B1-B2 ones, where the latter two have a higher SNR.}. Under the given IR, we decompose the estimated factor hyperplane into the factor block and its loadings. The results show that both quantities exhibit a higher relative error compared to that of the factor hyperplane. Of note, the loadings estimates exhibit a lot of variability as indicated by the high standard deviation in the Table.  
	
	Regarding the estimates in the VAR equation, for settings A1, A2 and B1 that are characterized by an adequate degree of sparsity, the recovery of the skeleton of the transition matrices is very good. However, performance deteriorates if the latent factor becomes ``more pervasive" (settings A3 and A4), which translates to the $A_{21}$ block having lower sparsity. On the other hand,
	this does not have much impact on the recovery of the $A_{22}$ sub-block, as for these two settings, SEN and SPC of $A_{22}$ still remain at a high level. For 
	settings with more lags, performance deteriorates (as expected) although SEN and SPC remain fairly satisfactory. On the other hand, the relative error of the transition matrices increases markedly. Nevertheless, the estimates of the first lag transition
	matrix is better than the remaining ones. Further, the results indicate that smaller size VAR systems (B3) exhibit better
	performance than larger ones. Finally, in terms of forecasting (results depicted in Table~\ref{table:sim-forecast}), the one-step-ahead forecasting value yields approximately 50\% to 90\% rel-err (compared to the truth), depending on the specific setting and the actual SNR, while it outperforms the forecast of the benchmark by around 40\% (based on the $\text{rel-err-ratio}$ measure). Of note, the $2$-step-ahead forecasting value for settings with more lags outperforms the benchmark by an even wider margin with the rel-err-ratio decreasing to less than 0.3. 
	
	Finally, the proposed methodology is robust in the presence of heavier than Gaussian tails in the VAR processes.
	Further, note that in setting C3 wherein the temporal dependence is strong and the error terms are generated according to a sub-exponential distribution, the performance of the estimated transition matrices deteriorates significantly, as expected from the theoretical results outlined in Appendix C. Nevertheless, with proper compensation in terms of sample size (setting C4), the performance improves markedly.
	
	\section{Application to Commodity Price Interlinkages}\label{sec:realdata}
	
	Interlinkages between commodity prices represent an active research area in economics and have been a source of concern
	for policymakers. Commodity prices, unlike stocks and bonds, are determined more strongly by global demand and supply considerations. Nevertheless, other factors are also at play as outlined next. The key ones are: (i) the state of the global macro-economy and the state of the business cycle that manifest themselves as direct demand for commodities; (ii) monetary policy, specifically, interest rates that impact the opportunity cost for holding inventories, as well as having an impact on investment and hence production capacity that subsequently contribute to changes in supply and demand in the market; and (iii) the relative performance of other asset classes through portfolio allocation \citep[see][and references therein]{frankel2008effect,frankel2014effects}. 
	We employ the FAVAR model and the proposed estimation method to investigate interlinkages amongst major commodity prices. The 
	$X_t$ block corresponds to the set of commodity prices of interest, while the $Y_t$ block contains representative indicators for the global economic environment.  We extract the factors $F_t$ based on the calibration equation and then consider the augmented VAR system of $(F_t,X_t)$, so that the estimated interlinkages amongst commodity prices are based on a larger
	information set that takes into account broader economic activities.
	
	\paragraph{Data.} The commodity price data ($X_t$) are retrieved from the International Monetary Fund, comprising of 16 commodity prices in the following categories: Metal, Energy (oil) and Agricultural. The set of economic indicators ($Y_t$) 
	contain core macroeconomic variables and stock market composite indices from major economic entities including China, EU, Japan, UK and US, with a total number of 54 indicators. Specifically, the macroeconomic variables primarily account for: Output \& Income (e.g. industrial production index), Labor Market (unemployment), Money \& Credit (e.g. M2), Interest \& Exchange Rate (e.g. Fed Funds Rate and the effective exchange rate), and Price Index (e.g. CPI). For variables that
	reflect interest rates, we use both the short-term interest rate such as 6-month LIBOR, and the 10-year T-bond yields from the secondary market. Further, to ensure stationarity of the time series, we take the difference of the logarithm for $X_t$; for $Y_t$, we apply the same transformation as proposed in \citet{stock2002forecasting}. A complete list of the commodity prices and economic indicators used in this study is provided in Appendix~\ref{appendix:list}. For all time series considered, we use monthly data spanning the January 2001 to December 2016 period. Further, based on previous empirical findings in the literature
	related to the global financial crisis of 2008 \citep{stock2017twenty}, we break the analysis into the following three sub-periods \citep{stock2017twenty}: pre-crisis (2001--2006), crisis (2007--2010) and post-crisis (2011--2016), each having sample size (available time points) 72, 48, and 72, respectively\footnote{For each individual time series, we test for its normality using data spanning the pre-crisis, crisis, and post-crisis periods, respectively. Based on the Shapiro-Wilk test, the null hypothesis of normality is not rejected for selected time series (e.g., ALUMINUM) and rejected for others (e.g., OIL). However, when testing for multivariate normality of the joint distribution of all time series resp. across the three periods, we fail to reject the null hypothesis. The latter result may be due to inadequate power of the test given the relatively small sample size.}. 
	
	We apply the same estimation procedure for each of the above three sub-periods. Starting with the calibration equation, we estimate the factor hyperplane $\Theta$ and the sparse regression coefficient matrix $\Gamma$, then extract the factors based on the estimated factor hyperplane under the (IR) condition. For each of the three sub-periods, 4, 3, and 3 factors are respectively identified based on the PIC criterion, with the key variable loadings (collapsed into categories) on each extracted factor listed in Table~\ref{table:factor_composition}, after adjusting for $\Gamma X_t$. 
	\begin{table}[!ht]
		\centering\scriptsize
		\caption{Composition of the factors identified for three sub-periods. $+$, $-$ and $*$ respectively stand for positive (all economic indicators in that category have a positive sign in $\Lambda$), negative and mixed (sign) contribution.}\label{table:factor_composition}
		\begin{tabular}{l|cccc|ccc|ccc}
			\specialrule{.1em}{.05em}{.05em} 
			& \multicolumn{4}{c|}{pre-crisis} & \multicolumn{3}{c|}{crisis} & \multicolumn{3}{c}{post-crisis} \\
			&F1	&F2	&F3	&F4	&F1	&F2	&F3	&F1	&F2	&F3	\\ \hline
			bond return 			&$-$&	&$+$&$+$&$-$&$+$&	&	&	&$-$ \\
			economic output 		&$+$&	&	&	&	&	&$+$&	&$+$&	\\
			equity return 			&$+$&	&	&	&$-$&$-$&	&$-$&&$+$\\
			interest/exchange rate	&	&	&$*$&	&	&	&	&$*$&	&\\
			labor 					&	&$+$&	&	&$-$&	&$-$&	&	&	\\
			money \& credit 		&	&	&$+$&	&$+$&	&	&	&$+$&\\
			price index 			&	&$+$&	&	&	&	&$+$&	&	&$-$ \\
			trade 					&	&$-$&	&	&	&$*$&	&$*$&	&\\
			\specialrule{.1em}{.05em}{.05em} 
		\end{tabular}
	\end{table}
	Based on the composition of the factors, we note that the factors summarize both the macroeconomic environment and also capture information from the secondary market (bond \& equity return), as suggested by economic analysis of potential
	contributors to commodity price movements \citep{frankel2008effect,frankel2014effects}. Hence, the obtained factors
	summarize the necessary information to include in the VAR system that examines commodity price interlinkages over time.
	Further, across all three periods considered, Economic Output and Money \& Credit indicators contribute positively to the factor composition. In particular, the positive contribution from the M2 measure of money supply for the US during the crisis period and that from the Fed Funds Rate post crisis are pronounced; hence, the estimated factors strongly reflect the effect of the Quantitative Easing policy adopted by the US central bank. The contribution of the other categories are mixed, with that from bond returns being noteworthy due to their role as a proxy for long-term interest rates, which impact both the cost of investment in increasing production capacity and on holding inventories, as well as on the composition of asset 
	portfolios across a range of investment possibilities (stocks, bonds, commodities, etc.). 
	
	Next, using these estimated factors, we fit a sparse $\VAR(2)$ model to the augmented $(\widehat{F}^\top_t,X^\top_t)^\top$ system. The estimated transition matrices are depicted in Figures~\ref{fig:pre-crisis} to~\ref{fig:post-crisis} as networks\footnote{In all three figures, the left panel corresponds to $\widehat{A}^{(1)}$ and the right panel corresponds to $\widehat{A}^{(2)}$. Node sizes are proportional to node weighted degrees. Positive edges are in red and negative edges are in blue. Edges with higher saturation have larger magnitudes.}. It is apparent that the factors play an important role, both as emitters and receivers. The effects from the first lag are generally stronger to that from the second one. In particular, focusing on the first lag, the dominant nodes in the system have shifted over time from (OIL, SOYBEANS, ZINC) pre crisis to (SUGAR, WHEAT, COPPER) during the crisis, then to (OIL, SOYBEANS, RICE) post crisis. Based on node weighted degree, the role of OIL is dominant in both pre- and post-crisis periods, but is much weaker during the crisis. 
	
	\begin{figure}[H]
		\centering
		\caption{Estimated transition matrices for Pre-crisis period.}\label{fig:pre-crisis}\vspace*{-2mm}
		\begin{boxedminipage}{0.47\textwidth}
			\includegraphics[scale=0.4]{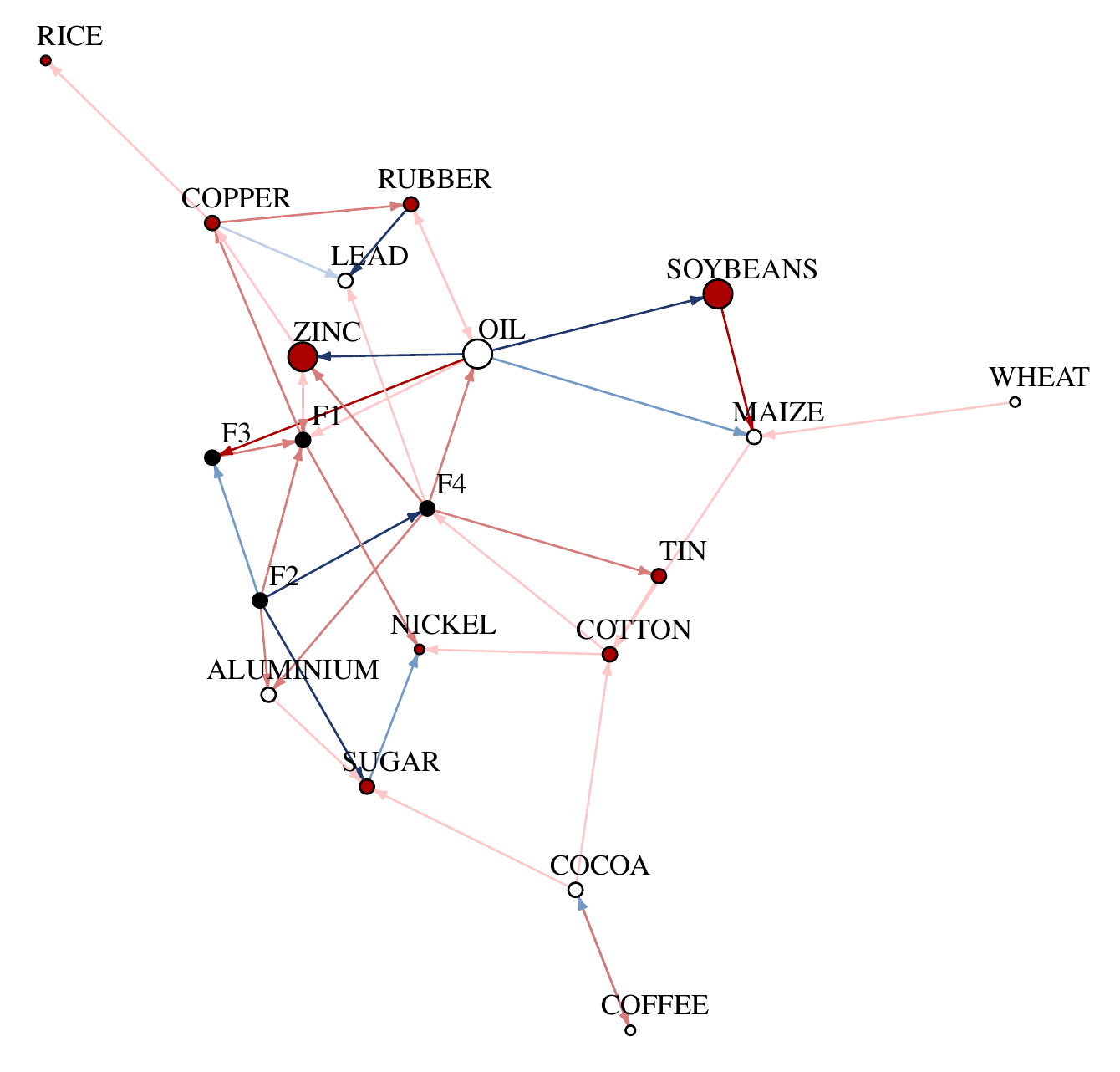}
		\end{boxedminipage}\quad 
		\begin{boxedminipage}{0.47\textwidth}
			\includegraphics[scale=0.398]{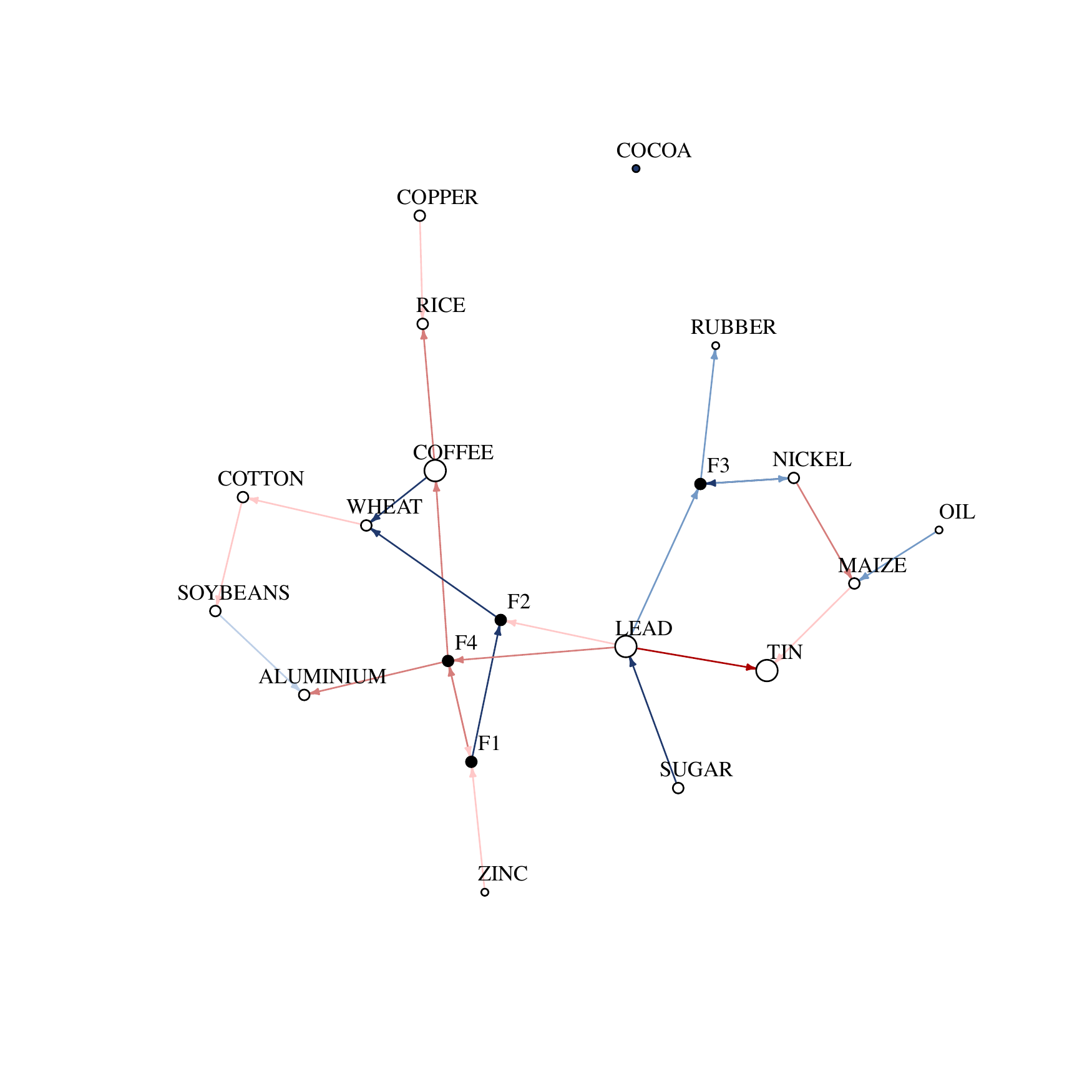}
		\end{boxedminipage}
	\end{figure}
	\begin{figure}[h]
		\centering
		\caption{Estimated transition matrices for the Crisis period. }\label{fig:crisis}\vspace*{-2mm}
		\begin{boxedminipage}{0.47\textwidth}
			\includegraphics[scale=0.4]{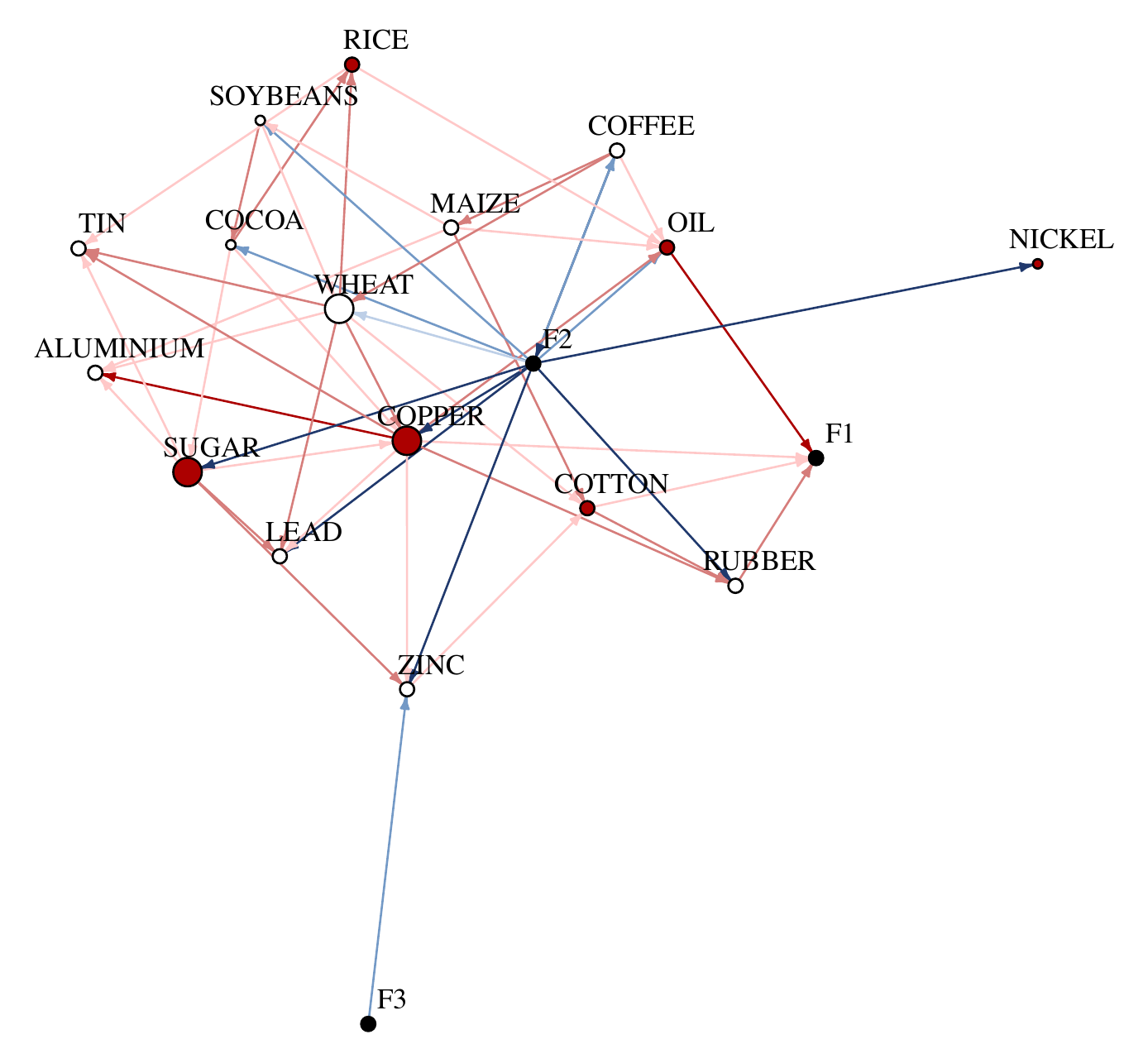}
		\end{boxedminipage}\quad 
		\begin{boxedminipage}{0.47\textwidth}
			\includegraphics[scale=0.385]{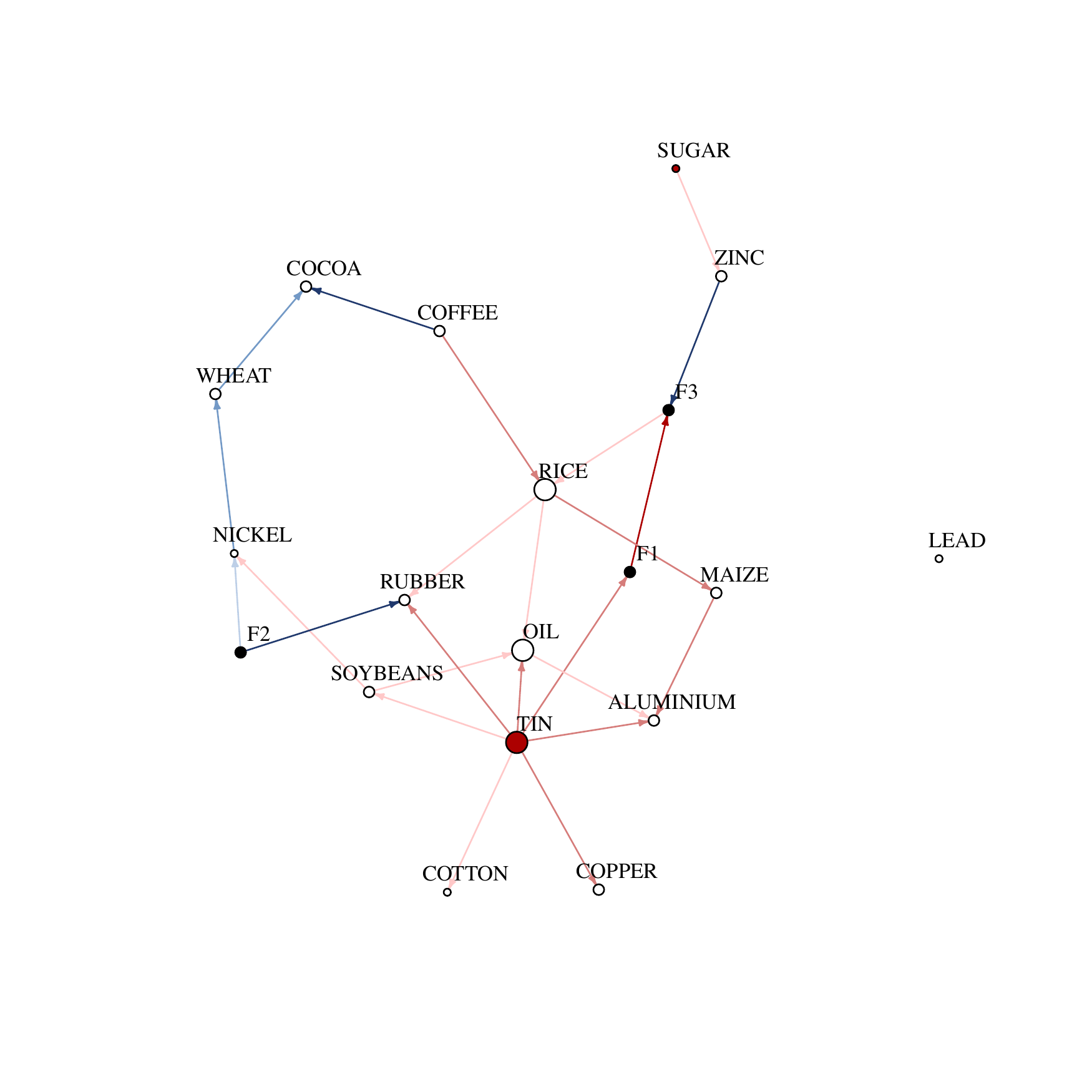}
		\end{boxedminipage}\vspace*{2mm}
		
		\caption{Estimated transition matrices for Post-crisis period}\label{fig:post-crisis}\vspace*{-2mm}
		\begin{boxedminipage}{0.47\textwidth}
			\includegraphics[scale=0.4]{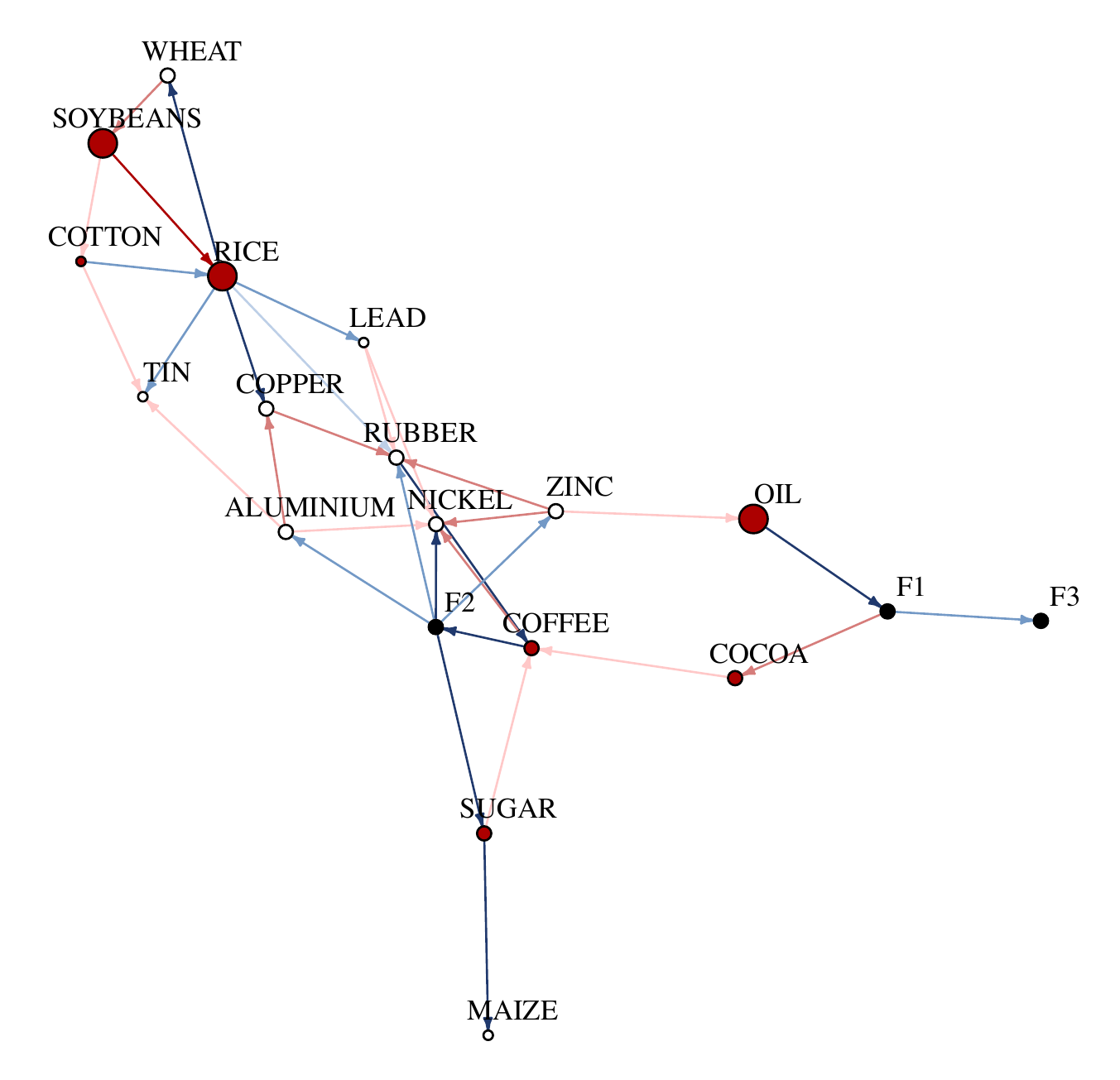}
		\end{boxedminipage}\quad 
		\begin{boxedminipage}{0.47\textwidth}
			\includegraphics[scale=0.4]{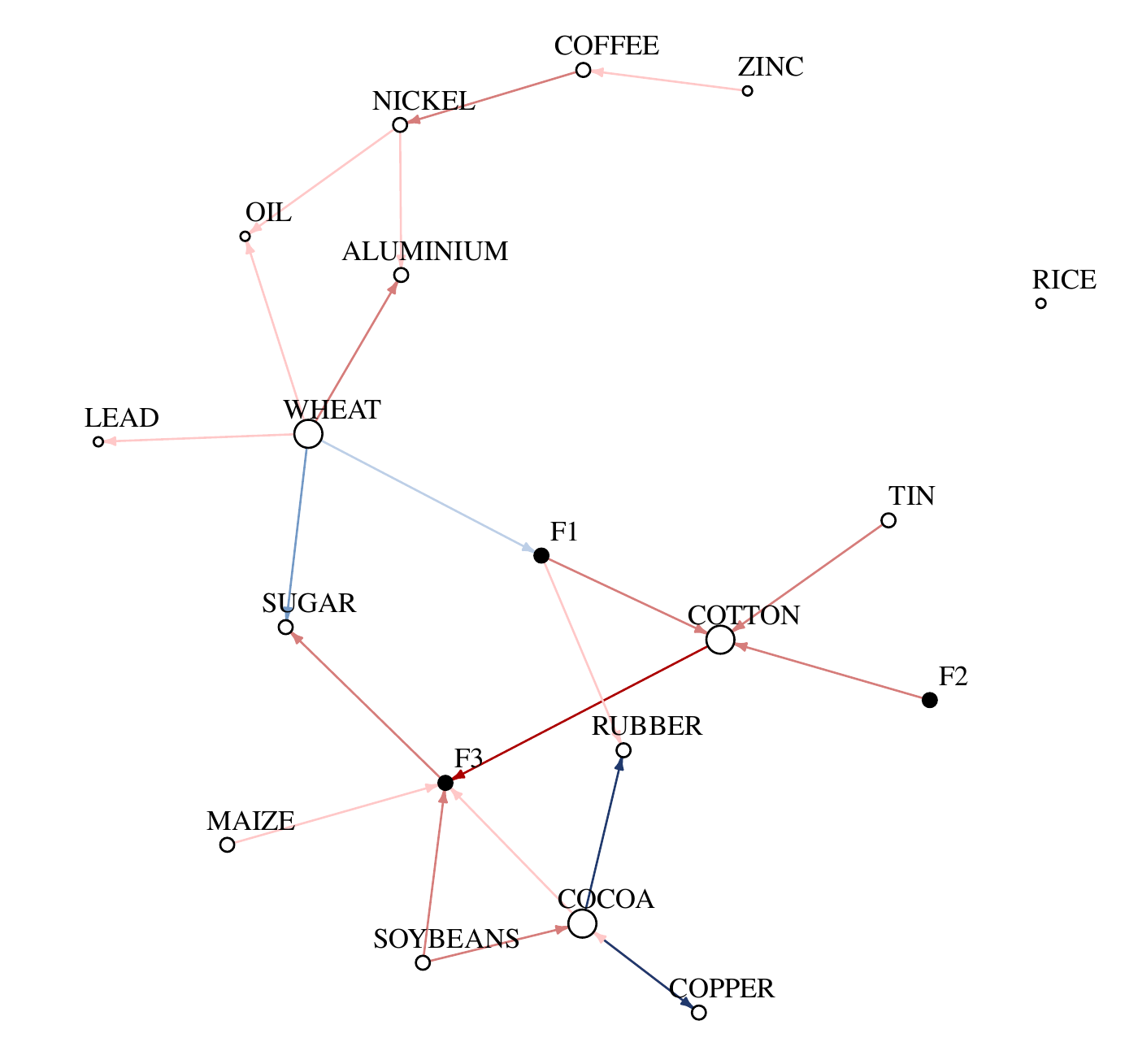}
		\end{boxedminipage} 
	\end{figure}
	
	Another key feature of the interlinkage networks is their increased connectivity during the crisis period, vis-a-vis the pre- and post-crisis periods. The same empirical finding has been noted for stock returns \citep[see][and references therein]{lin2017regularized}. Before the global financial crisis of 2008, commodity prices were fast rising primarily due to
	increased demand from China. Specifically, as Chinese industrial production quadrupled between 2001 and 2011, its consumption
	of industrial metals (Copper, Zinc, Aluminum, Lead) increased by 330\%, while its oil consumption by 98\%. This strong demand shock led to a sharp rise in these commodity prices, particularly accentuated beginning in 2006 (the onset of the crisis
	period considered in our analysis), briefly disrupted with a quick plunge of commodity prices in 2008 and their subsequent recovery in the ensuing period until late 2010, when demand from China subsided, which coupled with weak demand from
	the EU, Japan and the US in the aftermath of the crisis created an oversupply that put downward pressure on prices.
	These events induce strong inter-temporal and cross-temporal correlations amongst commodity prices, and hence are reflected in
	their estimated interlinkage network.

	\section{Discussion}\label{sec:discussion}
	This paper considered the estimation of FAVAR model under the high-dimensional scaling. It introduced
	an identifiability constraint (IR+Compactness) that is suitable for high-dimensional settings, and when such a constraint is incorporated in the optimization problem based upon the calibration equation, the global optimizer corresponds to model parameter estimates with bounded statistical errors. This development also allows for accurate estimation of the transition matrices of the VAR system, despite the plug-in factor block contains error due to the fact that it is an estimated quantity. Extensive numerical work illustrates the overall good performance of the proposed empirical implementation procedure, but also illustrates that the imposed constraint is not particularly stringent, especially in settings where the coefficient matrix $\Gamma$ of the observed predictor variables in the calibration equation exhibits sufficient level of sparsity.
	
	The key advantage of the FAVAR model is that it can leverage information from a large number of variables, while modeling the cross-temporal dependencies of a smaller number of them that are of primary interest to the analyst. 
	
	Recall that the nature of the FAVAR model results in estimating the transition matrix of a VAR system with one block of the observations (factors) being an estimated quantity, rather than conducting the estimation based on observed samples. Similar in flavor problems have been examined in the high-dimensional iid setting \citep[e.g.][]{loh2012high}, as well as low dimensional time series settings; for example, \citet{chanda1996asymptotic} examine parameter estimation of a
	univariate autoregressive process with error-in-variables and in more recent work \citet{komunjer2014measurement} investigate parameter identification of VAR-X and dynamic panel VAR models subject to measurement errors.

	\vskip 0.2in
	\bibliographystyle{chicago}
	\bibliography{ref_bib}
	
	\clearpage
	\appendix
	\numberwithin{lemma}{section}
	
	\section{Proofs for Theorems and Propositions}\label{sec:proof-thm-prop}
	
	This section is divided into two parts. In the first part, we provide proofs for the proposition and theorem related to Stage I estimates, i.e., $\widehat{\Theta}$ and $\widehat{\Gamma}$. In the second part, we give proofs for the statements related to Stage II estimates, namely $\widehat{A}$, with an emphasis on how to obtain the final high probability error bound through properly conditioning on related events. 
	
	\paragraph{Part 1.} Proofs for the $\widehat{\Theta}$ and $\widehat{\Gamma}$ estimates. 
	\begin{proof}[Proof of Proposition~\ref{prop:Theta-Gamma-fixed}]
		Using the optimality of $(\widehat{\Gamma},\widehat{\Theta})$ and the feasibility of $(\Gamma^\star,\Theta^\star)$, the following {\em basic inequality} holds:
		\begin{align}\label{eqn:basic1}
		\frac{1}{2n} \smalliii{\mathbf{X}\Delta_{\Gamma}^\top + \Delta_{\Theta}}_{\F}^2 \leq \frac{1}{n}\Big(\llangle \Delta_\Gamma^\top,\mathbf{X}^\top \mathbf{E}\rrangle + \llangle \Delta_{\Theta},\mathbf{E}\rrangle  \Big) + \lambda_\Gamma\Big(\smallii{\Gamma^\star}_1 - \smallii{\widehat{\Gamma}}_1 \Big),
		\end{align}
		which after rearranging terms gives
		\begin{align}\label{eqn:basic2}
		\begin{split}
		\frac{1}{2n}\smalliii{\mathbf{X}\Delta_{\Gamma}^\top}_{\F}^2 + &\frac{1}{2}\smalliii{\Delta_\Theta/\sqrt{n}}_{\F}^2  \leq  \frac{1}{n}\llangle \mathbf{X}\Delta_{\Gamma}^\top,\widehat{\Theta}-\Theta^\star \rrangle \\
		&+ \frac{1}{n} \Big( \llangle \Delta_{\Gamma}^\top,\mathbf{X}^\top\mathbf{E}\rrangle + \llangle \Delta_{\Theta},\mathbf{E} \rrangle\Big) + \lambda_\Gamma\Big(\smallii{\Gamma^\star}_1 - \smallii{\widehat{\Gamma}}_1 \Big).
		\end{split}
		\end{align}
		The remainder of the proof proceeds in three steps: in Step (i), we obtain a lower bound for the left-hand-side (LHS) leveraging the RSC condition; in Step (ii), an upper bound for the right-hand-side (RHS) based on the designated choice of $\lambda_\Gamma$ is derived; in Step (iii), the two sides are aligned to yield the desired error bound after rearranging terms.
		
		To complete the proof, we first define a few quantities that are associated with the support set of $\Gamma$ and its complement:
		\begin{equation*}
		\begin{split}
		\mathbb{S} &:= \big\{\Delta \in \R^{q\times p_2}|\Delta_{ij}=0 \text{ for }(i,j)\notin S_{\Gamma^\star} \big\}, \\
		\mathbb{S}^c &:= \big\{\Delta \in \R^{q\times p_2}|\Delta_{ij}=0 \text{ for }(i,j)\in S_{\Gamma^\star} \big\},
		\end{split}
		\end{equation*}
		where $S_{\Gamma^\star}$ is the support of $\Gamma^\star$. Further, define $\Delta_{\mathbb{S}}$ and $\Delta_{\mathbb{S}^c}$  as
		\begin{equation*}
		\Delta_{\mathbb{S},ij} = 1\{(i,j)\in S_{\Gamma^\star}\}\Delta_{ij}, \qquad \Delta_{\mathbb{S}^c,ij} = 1\{(i,j)\in S^c_{\Gamma^\star}\}\Delta_{ij}, 
		\end{equation*}
		and note that they satisfy
		\begin{equation*}
		\Delta = \Delta_{\mathbb{S}} + \Delta_{\mathbb{S}^c}, \qquad \|\Delta\|_1 = \|\Delta_{\mathbb{S}}\|_1 + \|\Delta_{\mathbb{S}^c}\|_1,
		\end{equation*}
		and
		\begin{equation}\label{eqn:BSnorms}
		\|\Delta_{\mathbb{S}}\|_1 \leq \sqrt{s} \smalliii{\Delta_{\mathbb{S}}}_{\F} \leq \sqrt{s_{\Gamma^\star}} \smalliii{\Delta}_{\F}.
		\end{equation}
		
		\medskip\noindent Step (i). Since $\mathbf{X}$ satisfies the RSC condition, the first term on the LHS of~\eqref{eqn:basic2} is lower bounded by 
		\begin{equation}\label{eqn:lowerRSC}
		\frac{\alpha_{\text{RSC}}^{\mathbf{X}}}{2}\smalliii{\Delta_\Gamma}_{\F}^2 - \tau_{\mathbf{X}} \smallii{\Delta_\Gamma}_1^2. 
		\end{equation}
		To get a lower bound for~\eqref{eqn:lowerRSC}, consider an upper bound for $\smallii{\Delta_\Gamma}_1$ with the aid of~\eqref{eqn:basic1}. Specifically, for the first two terms in the RHS of~\eqref{eqn:basic1}, by H\"{o}lder's inequality, the following inequalities hold for the inner products:
		\begin{equation}\label{eqn:holder}
		\begin{split}
		\llangle \Delta_\Gamma^\top,\mathbf{X}^\top\mathbf{E} \rrangle & \leq \|\Delta_\Gamma\|_1\|\mathbf{X}^\top \mathbf{E}\|_\infty\\
		\llangle \Delta_\Theta,\mathbf{E} \rrangle & \leq \smalliii{\Delta_\Theta}_*\smalliii{\mathbf{E}}_{op} = n\smalliii{\Delta_{\Theta}}_*\Lambda^{1/2}_{\max}(S_{\mathbf{E}});
		\end{split}
		\end{equation}
		for the last term, since
		\begin{equation*}
		\begin{split}
		\|\widehat{\Gamma}\|_1 = \| \Gamma^\star_\mathbb{S} + \Gamma^\star_{\mathbb{S}^c} + \Delta_{\Gamma|\mathbb{S}} + \Delta_{\Gamma|\mathbb{S}^c} \|_1 &= \| \Gamma^\star_\mathbb{S} +  \Delta_{\Gamma|\mathbb{S}}\|_1 + \|\Delta_{\mathbb{S}|\mathbb{S}^c} \|_1 \\
		&\geq \|\Gamma^\star_\mathbb{S}\|_1 - \|\Delta_{\Gamma|\mathbb{S}}\|_1 + \|\Delta_{\Gamma|\mathbb{S}^c} \|_1,
		\end{split}
		\end{equation*}
		the following inequality holds:
		\begin{equation}\label{eqn:B1normdiff}
		\|\Gamma^\star\|_1 -\|\widehat{\Gamma}\|_1 \leq \|\Delta_{\Gamma|\mathbb{S}}\|_1 - \|\Delta_{\Gamma|\mathbb{S}^c} \|_1.
		\end{equation} 
		Using the non-negativity of the RHS in~\eqref{eqn:basic1}, by choosing 
		\begin{equation}\label{eqn:lambda1}
		\lambda_\Gamma\geq \max\Big\{ 2\|\mathbf{X}^\top \mathbf{E}/n\|_\infty,~\Lambda^{1/2}_{\max}(S_{\mathbf{E}})\Big\},
		\end{equation}
		the following inequality holds:
		\begin{equation*}
		\begin{split}
		0\leq \frac{\lambda_\Gamma}{2}\|\Delta_{\Gamma}\|_1 + &\lambda_\Gamma\smalliii{\Delta_{\Theta}/\sqrt{n}}_* + \lambda_\Gamma(\|\Delta_{\Gamma|\mathbb{S}}\|_1 - \|\Delta_{\Gamma|\mathbb{S}^c} \|_1) \\
		&= \frac{3\lambda_\Gamma}{2}\|\Delta_{\Gamma|\mathbb{S}}\|_1 - \frac{\lambda_\Gamma}{2}\|\Delta_{\Gamma|\mathbb{S}^c}\|_1 + \lambda_\Gamma\smalliii{\Delta_{\Theta}/\sqrt{n}}_*.
		\end{split}
		\end{equation*}
		Since $\Delta_{\Theta}=\widehat{\Theta}-\Theta^\star$ has rank at most $p_1+r$, $\smalliii{\Delta_{\Theta}/\sqrt{n}}_* \leq \sqrt{p_1+r}\smalliii{\Delta_{\Theta}/\sqrt{n}}_{\F}$. It follows that 
		\begin{equation*}
		\begin{split}
		\frac{\lambda_\Gamma}{2}\|\Delta_{\Gamma|\mathbb{S}^c}\|_1  &\leq \lambda_\Gamma\sqrt{p_1+r}\smalliii{\Delta_{\Theta}/\sqrt{n}}_{\F} + \frac{3\lambda_\Gamma}{2}\|\Delta_{\Gamma|\mathbb{S}}\|_1,\\
		\frac{\lambda_\Gamma}{2}\|\Delta_{\Gamma|\mathbb{S}}\|_1 +\frac{\lambda_\Gamma}{2}\|\Delta_{\Gamma|\mathbb{S}^c}\|_1  &\leq \lambda_\Gamma\sqrt{p_1+r}\smalliii{\Delta_\Theta/\sqrt{n}}_{\F} + \frac{3\lambda_\Gamma}{2}\|\Delta_{\Gamma|\mathbb{S}}\|_1 + \frac{\lambda_\Gamma}{2}\|\Delta_{\Gamma|\mathbb{S}}\|_1, \\
		\|\Delta_{\Gamma}\|_1 & \leq \sqrt{4(p_1+r)}\smalliii{\Delta_\Theta/\sqrt{n}}_{\F} + 4\|\Delta_{\Gamma|\mathbb{S}}\|_1 \\
		& \leq \sqrt{4(p_1+r)} \smalliii{\Delta_\Theta/\sqrt{n}}_{\F} + 4\sqrt{s_{\Gamma^\star}}\smalliii{\Delta_{\Gamma}}_{\F},
		\end{split}
		\end{equation*}
		where the second line is obtained by adding $\frac{\lambda_\Gamma}{2}\|\Delta_{\Gamma|\mathbb{S}}\|_1$ on both sides, and the last inequality uses~\eqref{eqn:BSnorms}. Further, by the Cauchy-Schwartz inequality, we have
		\begin{equation*}
		\smallii{\Delta_{\Gamma}}_1\leq \sqrt{(\sqrt{4(p_1+r)})^2 + (4\sqrt{s})^2} \sqrt{\smalliii{\Delta_{\Gamma}}_{\F}^2 + \smalliii{\Delta_\Theta/\sqrt{n}}_{\F}^2},
		\end{equation*}
		that is,
		\begin{equation}\label{eqn:DeltaB1bound}
		\|\Delta_{\Gamma}\|_1^2 \leq 4\big(p_1+r + 4s\big) \Big[ \smalliii{\Delta_{\Gamma}}_{\F}^2 + \smalliii{\Delta_\Theta/\sqrt{n}}_{\F}^2\Big]. 
		\end{equation}
		Combine~\eqref{eqn:lowerRSC} and~\eqref{eqn:DeltaB1bound}, a lower bound for the LHS of~\eqref{eqn:basic2} is given by 
		\begin{equation}\label{eqn:LHSlow2-nc}
		\Big(\frac{\alpha_{\text{RSC}}^{\mathbf{X}}}{2} - 4\tau_{\mathbf{X}}\big(p_1+r + 4s \big)\Big)\smalliii{\Delta_{\Gamma}}_{\F}^2 + \Big(\frac{1}{2}- 4\tau_{\mathbf{X}}\big(p_1+r + 4s \big) \Big) \smalliii{\Delta_\Theta/\sqrt{n}}_{\F}^2.
		\end{equation}
		
		\medskip\noindent Step (ii). For the first term in the RHS of~\eqref{eqn:basic2}, using the duality of $\ell_1$-$\ell_\infty$ dual norm pair, the following inequality holds:
		\begin{equation}\label{eqn:firstterm}
		\begin{split}
		\frac{1}{n}|\llangle \mathbf{X}\Delta_{\Gamma}^\top,\widehat{\Theta}-\Theta^\star\rrangle|  &\leq \frac{1}{n}|\llangle \Delta_{\Gamma}^\top,\mathbf{X}^\top\widehat{\Theta}\rrangle| + \frac{1}{n}|\llangle \Delta_{\Gamma}^\top,\mathbf{X}^\top\Theta^\star\rrangle| \\
		&\leq  \|\Delta_{\Gamma}\|_1 \|\mathbf{X}^\top\widehat{\Theta}/n\|_\infty +  \|\Delta_{\Gamma}\|_1  \|\mathbf{X}^\top\Theta^\star/n\|_\infty\\
		& \leq \|\Delta_{\Gamma}\|_1 \cdot\smallii{\mathbf{X}/n}_1\cdot \|\widehat{\Theta}\|_\infty + \|\Delta_{\Gamma}\|_1 \cdot\smallii{\mathbf{X}/n}_1\cdot \|\Theta^\star\|_\infty.
		\end{split}
		\end{equation}
		Using the fact that both $\Theta^\star$ and $\widehat{\Theta}$ are feasible and satisfy the box constraint $\|\Theta\|_\infty \leq \frac{\phi}{\kappa(\mathcal{R}^*)\smalliii{\mathbf{X}/\sqrt{n}}_{\op}}$, it follows that 
		\begin{equation*}
		\smalliii{\mathbf{X}/n}_1\|\widehat{\Theta}\|_\infty \leq \frac{\phi}{\kappa(\mathcal{R}^*)} \quad \text{and} \quad \smalliii{\mathbf{X}/n}_1\|\Theta^\star\|_\infty \leq \frac{\phi}{\kappa(\mathcal{R}^*)},
		\end{equation*}
		Consequently, \eqref{eqn:firstterm} is upper bounded by $\tfrac{2\phi}{\kappa(\mathcal{R}^*)}\cdot \|\Delta_{\Gamma}\|_1$. By additionally requiring $\lambda_\Gamma$ to satisfy 
		\begin{equation*}
		\lambda_\Gamma\geq 4\phi/\kappa(\mathcal{R}^*),
		\end{equation*}
		and combining~\eqref{eqn:holder}, \eqref{eqn:B1normdiff} and~\eqref{eqn:lambda1}, the following upper bound holds for the RHS of~\eqref{eqn:basic2}:
		\begin{align}
		\tfrac{\lambda_\Gamma}{2}\smallii{\Delta_{\Gamma}}_{1} + \tfrac{\lambda_\Gamma}{2}\|\Delta_{\Gamma}\|_1 &+ \lambda_\Gamma\sqrt{p_1+r}\smalliii{\Delta_\Theta/\sqrt{n}}_{\F} + \lambda_\Gamma(\|\Delta_{\Gamma|\mathbb{S}}\|_1 - \|\Delta_{\Gamma|\mathbb{S}^c} \|_1) \notag \\
		&\leq \lambda_\Gamma\Big( 2\sqrt{s_{\Gamma^\star}}\smalliii{\Delta_{\Gamma}}_{\F} + \sqrt{p_1+r}\smalliii{\Delta_\Theta/\sqrt{n}}_{\F} \Big) \notag \\ 
		&\leq \lambda_\Gamma\sqrt{4s_{\Gamma^\star} + p_1+r} \sqrt{\smalliii{\Delta_{\Gamma}}_{\F}^2 +\smalliii{\Delta_\Theta/\sqrt{n}}_{\F}^2 }. \label{eqn:RHSup-nc}
		\end{align}
		
		\medskip\noindent Step (iii). Combine~\eqref{eqn:LHSlow2-nc} and~\eqref{eqn:RHSup-nc}, by rearranging terms and requiring 
		$\tau_{\mathbf{X}}$ to satisfy $\tau_{\mathbf{X}}( p_1+r + 4s_{\Gamma^\star}) < \min\{\alpha_{\text{RSC}}^{\mathbf{X}},1\}/16$, the following inequality holds:
		\begin{equation*}
		\begin{split}
		\frac{\min\{\alpha_{\text{RSC}}^{\mathbf{X}},1\}}{4}&\Big(\smalliii{\Delta_{\Gamma}}_{\F}^2  + \smalliii{\Delta_{\Theta}/\sqrt{n}}_{\F}^2 \Big)  \\
		&\leq  \lambda_\Gamma\sqrt{4s_{\Gamma^\star} + p_1+r} \sqrt{\smalliii{\Delta_{\Gamma}}_{\F}^2 +\smalliii{\Delta_\Theta/\sqrt{n}}_{\F}^2 },
		\end{split}
		\end{equation*}
		which gives
		\begin{equation*}
		\smalliii{\Delta_{\Gamma}}_{\F}^2 + \smalliii{\Delta_{\Theta}/\sqrt{n}}_{\F}^2 \leq \frac{16\lambda_\Gamma^2\Big( p_1+r + 4s_{\Gamma^\star} \Big) }{\min\{\alpha_{\text{RSC}}^{\mathbf{X}},1\}^2}.
		\end{equation*}
	\end{proof}
	
	\begin{proof}[Proof sketch for Theorem~\ref{thm:info}]
		First we note that the requirement on the tuning parameter $\lambda_\Gamma$ determines the leading term in the ultimate high probability error bound. By Lemma~\ref{lemma:Emax}, to have adequate concentration for the leading eigenvalue $\Lambda_{\max}(S_{\mathbf{E}})$ of the sample covariance matrices, the requirement imposed on the sample size makes $\sqrt{\log(p_2q)/n}$ a lower order term relative to $\Lambda^{1/2}_{\max}(\Sigma_e)$, with the latter being an $\mathcal{O}(1)$ term. Consequently, the choice of the tuning parameter effectively becomes 
		\begin{equation*}
		\lambda_\Gamma\asymp \mathcal{O}(1),
		\end{equation*}
		The conclusion readily follows as a result of Proposition~\ref{prop:Theta-Gamma-fixed}. 
	\end{proof}
	
	\paragraph{Part 2.} This part contains the proofs for the results related to $\widehat{A}$.
	
	\begin{proof}[Proof sketch for Proposition~\ref{prop:A-fixed}] The result follows along the lines of \citet[Proposition 4.1]{basu2015estimation}. In particular, in \citet{basu2015estimation}, the authors consider estimation of $A$ based on the directly observed samples of the $X_t$ process, with the restricted eigenvalue (RE) condition imposed on the corresponding Hessian matrix and the tuning parameter selected in accordance to the deviation bound defined in Definition~\ref{defn:db}.  
		
		On the other hand, in the current setting, estimation of the transition matrix is based on quantities that are surrogates for the true sample quantities. Consequently, as long as the required conditions are imposed on their counterparts associated with these surrogate quantities, the conclusion directly follows.  
		
		Finally, we would like to remark that the RSC condition used is in essence identical to the RE condition required in \citet{basu2015estimation} in the setting under consideration. 
	\end{proof}
	
	\begin{proof}[Proof of Theorem~\ref{thm:A}]
		First, we note that under (IR), by Theorem~\ref{thm:info}, there exists some constant $K_1$ that is independent of $n,p_1,p_2$ and $q$ such that the following event holds with probability at least $\mathsf{P}_1:=1-c_1\exp(-c_2\log(p_2q))$:
		\begin{equation*}
		\mathcal{E}_1 := \Big\{\smalliii{\Delta_{\mathbf{F}}/\sqrt{n}}_{\F} \leq K_1 \Big\}.
		\end{equation*}
		Conditional on $\mathcal{E}_1$, by Proposition~\ref{prop:A-fixed}, Lemmas~\ref{lemma:RSCZ} and~\ref{lemma:deviation-A}, with high probability, the following event holds:
		\begin{equation*}
		\mathcal{E}_2 := \Big\{ \smalliii{\Delta_A}_{\F} \leq \varphi(n,p_1,p_2,K_1) \Big\}, 
		\end{equation*}
		for some function $\varphi(\cdot)$ that not only depends on sample size and dimensions, but also on $K_1$, provided that the ``conditional" RSC condition is satisfied. What are left to be examined are: (i) what does $\mathcal{E}_1$ imply in terms of the RSC condition being satisfied {\em unconditionally}; and (ii) what does $\mathcal{E}_1$ imply in terms of the bound in $\mathcal{E}_2$, 
		
		Towards this end, for (i), we note that since
		\begin{equation*}
		\Lambda^{1/2}_{\max}(S_{\Delta_{\mathbf{F}_{n-1}}}) = \smalliii{\Delta_{\mathbf{F}}/\sqrt{n}}_{op} \leq \smalliii{\Delta_{\mathbf{F}}/\sqrt{n}}_{\F} \leq K_1, 
		\end{equation*}
		then as long as $C_Z$ in condition C3 satisfies $C_Z\geq c_0K_1$ with the specified $c_0\geq 6\sqrt{165\pi}$, with probability at least $P_1P_{2,\RSC}$ where we define $\mathsf{P}_{2,\RSC}:=1-c'_1\exp(-c'_2n)$, by Lemma~\ref{lemma:RSCZ} the required RSC condition is guaranteed to be satisfied with a positive curvature. For (ii), with the aid of Lemma~\ref{lemma:deviation-A}, with probability at least $\mathsf{P}_1\mathsf{P}_{2,{\text{DB}}}$ where we define $\mathsf{P}_{2,\text{DB}}:=1-c_1'\exp(-c_2'\log(p_1+p_2))$, the following bound holds for the deviation bound $C(n,p_1,p_2)$ {\em unconditionally}:\footnote{Note that it can be shown that $\smalliii{\varepsilon_n}_{\F}^2 = O(\smalliii{\Delta_{\mathbf{F}}}_{\F}^2)$} 
		\begin{equation*}
		\begin{split}
		C(&n,p_1,p_2) \leq ~C_1\big(\mathcal{M}(f_Z)+\frac{\Sigma_w}{2\pi}+\mathcal{M}(f_{Z,W})\big)\sqrt{\frac{\log(p_1+p_2)}{n}} \\ &+C_2\mathcal{M}^{1/2}(f_Z)\sqrt{\frac{\log(p_1+p_2)+\log p_1}{n}} 
		~+ C_3\Lambda^{1/2}_{\max}(\Sigma_w)\sqrt{\frac{\log(p_1+p_2)}{n}} + C_4,
		\end{split}
		\end{equation*}
		where the constants $\{C_i\}$ have already absorbed the upper error bound $K_1$ of the Stage I estimates, compared with the original expression in Proposition~\ref{prop:A-fixed}. With the required sample size, the constant becomes the leading term, so that there exists some constant $K_2$ such that {\em unconditionally}:
		\begin{equation*}
		C(n,p_1,p_2)\leq K_2 \asymp \mathcal{O}(1).
		\end{equation*}
		Combine (i) and (ii), and with probability at least $\min\{\mathsf{P}_1\mathsf{P}_{2,\RSC}, \mathsf{P}_1\mathsf{P}_{2,{\text{DB}}}\}$, the bound in Theorem~\ref{thm:A} holds. 
	\end{proof}
	
	
	\section{Proof for Lemmas}\label{appendix:lemma}
	In this section, we provide proofs for the lemmas in Section~\ref{sec:theory-random}. 
	\begin{proof}[Proof of Lemma~\ref{lemma:err-Lambda}]
		Note that
		\begin{equation*}
		\begin{split}
		\widehat{\Theta} = \Theta^\star + \Delta_\Theta& = (\mathbf{F}+\Delta_{\F})(\Lambda^\star + \Delta_\Lambda)^\top \\
		\Delta_{\Theta} &= \Delta_{\F} (\Lambda^\star)^\top + \widehat{\mathbf{F}}\Delta_\Lambda^\top.
		\end{split}
		\end{equation*}
		Multiply the left inverse of $\widehat{\mathbf{F}}$ which gives 
		\begin{equation*}
		\Delta_\Lambda^\top  =  \big(\widehat{\mathbf{F}}^\top\widehat{\mathbf{F}}\big)^{-1}\widehat{\mathbf{F}}^\top  \Delta_{\Theta} + \big(\widehat{\mathbf{F}}^\top\widehat{\mathbf{F}}\big)^{-1}\widehat{\mathbf{F}}^\top  \Delta_{\mathbf{\mathbf{F}}} (\Lambda^\star)^\top.
		\end{equation*}
		Since for some generic matrix $M$, we have $\smalliii{M^{-1}}_{\F}\geq (\smalliii{M}_{\F})^{-1}$, an application of the triangle inequality gives
		\begin{equation*}
		\begin{split}
		\smalliii{\Delta_\Lambda}_{\F} & \leq \frac{\smalliii{\widehat{\mathbf{F}}}_{\F}} { \smalliii{\widehat{\mathbf{F}}^\top\widehat{\mathbf{F}} }_{\F}}\Big( \smalliii{\Delta_\Theta }_{\F} + \smalliii{\Delta_\mathbf{F} (\Lambda^\star)^\top }_{\F}\Big) \\
		& =   \frac{\smalliii{\widehat{\mathbf{F}}/\sqrt{n}}_{\F}} { \smalliii{\frac{1}{n}\widehat{\mathbf{F}}^\top\widehat{\mathbf{F}} }_{\F}} (\frac{1}{\sqrt{n}})\Big( \smalliii{\Delta_\Theta}_{\F} + \smalliii{\Delta_\mathbf{F} (\Lambda^\star)^\top }_{\F}\Big)\\
		& \leq \sqrt{p_1}\Lambda^{-1/2}_{\max}(S_{\widehat{\mathbf{F}}}) \smalliii{\Delta_\Theta/\sqrt{n} }_{\F} \Big( 1+ \smalliii{\Lambda^\star }_{\F}\Big),
		\end{split}
		\end{equation*}
		where $S_{\widehat{\mathbf{F}}}:=\frac{1}{n}\widehat{\mathbf{F}}^\top\widehat{\mathbf{F}}$, and the numerator and the denominator of $\frac{\smalliii{\widehat{\mathbf{F}}}_{\F}} { \smalliii{\widehat{\mathbf{F}}^\top\widehat{\mathbf{F}} }_{\F}}$ are respectively by 
		\begin{equation*}
		\smalliii{\widehat{\mathbf{F}}}_{\F}\leq \sqrt{p_1}\smalliii{\widehat{\mathbf{F}}}_{\op}, \qquad \smalliii{\widehat{\mathbf{F}}^\top\widehat{\mathbf{F}}}_{\F}\geq \smalliii{\widehat{\mathbf{F}}^\top\widehat{\mathbf{F}}}_{\op}.
		\end{equation*}
		Further, note that $\smalliii{\widehat{\mathbf{F}}/\sqrt{n}}^2_{\op} = \Lambda_{\max}(S_{\widehat{\mathbf{F}}}) = \smalliii{S_{\widehat{\mathbf{F}}}}_{\op}$. What remains is to obtain a lower bound for
		\begin{equation*}
		\Lambda^{1/2}_{\max}\big(S_{\widehat{\mathbf{F}}}\big) = \smalliii{ (\mathbf{F}+\Delta_{\mathbf{F}})/\sqrt{n} }_{\op}.
		\end{equation*}
		One such bound is given by 
		\begin{equation*}
		\begin{split}
		\smalliii{ (\mathbf{F}+\Delta_{\mathbf{F}})/\sqrt{n} }_{\op} \geq \smalliii{\mathbf{F}/\sqrt{n}}_{\op} - \smalliii{\Delta_{\mathbf{F}}/\sqrt{n}}_{\op} &\geq \smalliii{\mathbf{F}/\sqrt{n}}_{\op} - \smalliii{\Delta_{\mathbf{F}}/\sqrt{n}}_{\F}\\
		&\geq \Lambda^{1/2}_{\max}(S_{\mathbf{F}})- \smalliii{\Delta_{\Theta}/\sqrt{n}}_{\F}, 
		\end{split}
		\end{equation*}
		which leads to the following bound for $\smalliii{\Delta_\Lambda}_{\F}$, provided that the RHS is positive:
		\begin{equation*}
		\frac{\smalliii{\Delta_\Lambda}_{\F}}{\smalliii{\Lambda^\star}_{\F}}\leq \sqrt{p_1}\frac{\smalliii{\Delta_{\Theta}/\sqrt{n}}_{\F}}{\Lambda^{1/2}_{\max}(S_{\mathbf{F}})- \smalliii{\Delta_{\Theta}/\sqrt{n}}_{\F}}\Big(1 + 1/\smalliii{\Lambda^\star}_{\F} \Big).
		\end{equation*}
	\end{proof}
	
	\begin{proof}[Proof of Lemma~\ref{lemma:RSCX}]
		First, suppose we have
		\begin{equation}\label{eqn:RSC-stronger}
		\frac{1}{2}v' S_{\mathbf{X}} v = \frac{1}{2}v'\Big( \frac{\mathbf{X}'\mathbf{X}}{n} \Big)v\geq \frac{\alpha_{\text{RSC}}}{2} \|v\|_2^2 - \tau_n\|v\|_1^2,\quad \forall~v\in\R^p;
		\end{equation}
		then, for all $\Delta\in\R^{p\times p}$, and letting $\Delta_j$ denote its $j$th column, the RSC condition automatically holds since
		\begin{equation*}
		\begin{split}
		\frac{1}{2n}\smalliii{\mathbf{X}\Delta}_{\F}^2 = \frac{1}{2}\sum_{j=1}^q \Delta_j'\big( \tfrac{\mathbf{X}'\mathbf{X}}{n} \big) \Delta_j &\geq \frac{\alpha_{\text{RSC}}}{2} \sum_{j=1}^q \|\Delta_j\|_2^2 - \tau_n \sum_{j=1}^q \|\Delta_j\|_1^2 \\
		& \geq \frac{\alpha_{\text{RSC}}}{2}\smalliii{\Delta}_{\F}^2 - \tau_n \|\Delta\|_1^2.
		\end{split}
		\end{equation*}
		Therefore, it suffices to verify that~\eqref{eqn:RSC-stronger} holds. In \citet[Proposition 4.2]{basu2015estimation}, the authors prove a similar result under the assumption that $X_t$ is a $\text{VAR}(d)$ process. Here, we adopt the same proof strategy and state the result for a {\em more general process} $X_t$. 
		
		Specifically, by \citet[Proposition 2.4(a)]{basu2015estimation}, $\forall v\in\R^p,\|v\|\leq 1$ and $\eta >0$,
		\begin{equation*}
		\PP \Big[ \big| v'\big(S_{\mathbf{X}}-\Sigma_X(h)\big) v \big| >  2\pi\mathcal{M}(g_X)\eta \Big]\leq 2\eta \exp \Big(-cn\min\{\eta^2,\eta\} \Big). 
		\end{equation*}
		Applying the discretization in \citet[Lemma F.2]{basu2015estimation} and taking the union bound, define $\mathbb{K}(2s):=\{v\in\R^p, \|v\|\leq 1, \|v\|_0\leq 2k\}$, and the following inequality holds:
		\begin{equation*}
		\begin{split}
		\PP \Big[ \sup\limits_{v\in\mathbb{K}(2k)}\big| &v'\big(S_{\mathbf{X}}-\Sigma_X(h)\big) v \big| >  2\pi\mathcal{M}(g_X)\eta \Big] \\
		& \leq 2\exp\Big( -cn\min\{\eta,\eta^2\} + 2k \min\{\log p, \log (21 ep/2k) \} \Big). 
		\end{split}
		\end{equation*}
		With the specified $\gamma=54\mathcal{M}(g_X)/\mathfrak{m}(g_X)$, set $\eta=\gamma^{-1}$, then apply results from \citet[Lemma 12]{loh2012high} with $\Gamma=S_{\mathbf{X}}-\Sigma_X(0)$ and $\delta=\pi\mathfrak{m}(g_X)/27$, so that the following holds
		\begin{equation*}
		\frac{1}{2}v' S_{\mathbf{X}} v \geq \frac{\alpha_{\text{RSC}}}{2}\|v\|^2 - \frac{\alpha_{\text{RSC}}}{2k}\|v\|_1^2,
		\end{equation*}
		with probability at least $1-2\exp\big(-cn\min\{\gamma^{-2},1\} + 2k\log p \big)$ and note that $\min\{\gamma^{-2},1\}=\gamma^{-2}$ since $\gamma>1$. Finally, let $k= \min\{cn\gamma^{-2}/(c'\log p),1\}$ for some $c'>2$, and conclude that with probability at least $1-c_1\exp(-c_2n)$, the inequality in~\eqref{eqn:RSC-stronger} holds with 
		\begin{equation*}
		\alpha_{\text{RSC}} = \pi\mathfrak{m}(g_X), \qquad \tau_n = \alpha_{\text{RSC}}\gamma^2\frac{\log p}{2n},
		\end{equation*}
		and so does the RSC condition. 
	\end{proof}
	\begin{proof}[Proof of Lemma~\ref{lemma:boundXE}]
		We note that
		\begin{equation*}
		\frac{1}{n}\vertii{\mathbf{X}^\top \mathbf{E}}_\infty = \max\limits_{1\leq i,j\leq p} \big|e_i^\top \big(\mathbf{X}^\top \mathbf{E}/n\big)e_j\big|,
		\end{equation*}
		where $e_i$ is the $p$-dimensional standard basis with its $i$-th entry being 1. Applying \citet[Proposition 2.4(b)]{basu2015estimation}, for an arbitrary pair of $(i,j)$, the following inequality holds:
		\begin{equation*}
		\PP\Big[ \big|e_i^\top \big(\mathbf{X}^\top \mathbf{E}/n\big)e_j \big| > 2\pi \big(\mathcal{M}(g_X) + \frac{\Lambda_{\max}(\Sigma_e)}{2\pi} \big)\eta\Big] \leq 6\exp\Big( -cn \min\{\eta^2,\eta\}\Big),
		\end{equation*}	
		and note that $e_t$ is a pure noise term that is assumed to be independent of $X_t$; hence, there is no cross-dependence term to consider. Take the union bound over all $1\leq i\leq p_2,1\leq j\leq q$, and the following bound holds:
		\begin{equation*}
		\begin{split}
		\PP \Big[ \max\limits_{1\leq i\leq p_2,1\leq j\leq q} \big|e_i^\top \big(\mathbf{X}^\top \mathbf{E}/n\big)e_j\big|&  > 2\pi \big(\mathcal{M}(g_X) + \frac{\Lambda_{\max}(\Sigma_e)}{2\pi}\big)\eta\Big] \\
		& \leq 6\exp \Big( -cn \min\{\eta^2,\eta\} + \log (p_2q)\Big).
		\end{split}{}
		\end{equation*}
		Set $\eta=c'\sqrt{\log p/n}$ for $c'>(1/c)$ and with the choice of $n\gtrsim \log (p_2q)$, $\min\{\eta^2,\eta\} = \eta^2$, then with probability at least $1-c_1\exp(-c_2\log p_2q)$, there exists some $c_0$ such that the following bound holds:
		\begin{equation*}
		\frac{1}{n}\vertii{\mathbf{X}^\top \mathbf{E}}_\infty\leq c_0 \big(2\pi\mathcal{M}(g_X) + \Lambda_{\max}(\Sigma_e)\big)\sqrt{\frac{\log (p_2q)}{n}}.
		\end{equation*}
	\end{proof}
	
	\begin{proof}[Proof of Lemma~\ref{lemma:Emax}]
		For $\mathbf{E}$ whose rows are iid realizations of a sub-Gaussian random vector $e_t$, by \citet[Lemma 9]{wainwright2009sharp}, the following bound holds:
		\begin{equation*}
		\mathbb{P}\Big[ \vertiii{S_{\mathbf{E}} - \Sigma_e }_{op}\geq \Lambda_{\max}(\Sigma_e) \delta(n,q,\eta)\Big]\leq 2\exp(-n\eta^2/2), 
		\end{equation*}
		where $\delta(n,q,\eta):=2\big(\sqrt{\frac{q}{n}}+\eta\big)+\big(\sqrt{\frac{q}{n}}+\eta\big)^2$. In particular, by triangle inequality, with probability at least $1-2\exp(-n\eta^2/2)$, 
		\begin{equation*}
		\smalliii{S_{\mathbf{E}}}_{op} \leq \smalliii{\Sigma_\epsilon}_{op} + \smalliii{S_{\mathbf{E}}-\Sigma_\epsilon}_{op} \leq \Lambda_{\max}(\Sigma_\epsilon) + \Lambda_{\max}(\Sigma_\epsilon) \delta(n,q,t).
		\end{equation*}
		So for $n\gtrsim q$, by setting $\eta=1$, which yields $\delta(n,q,\eta)\leq 8$ so that with probability at least $1-2\exp(-n/2)$, the following bound holds:
		\begin{equation*}
		\Lambda_{\max}(S_{\mathbf{E}}) \leq 9 \Lambda_{\max}(\Sigma_\epsilon).
		\end{equation*}
	\end{proof}

	\begin{proof}[Proof of Lemma~\ref{lemma:Smax}]
		To prove this lemma, we use a similar strategy as in the proof of \citet[Lemma 3]{negahban2011estimation} while taking into consideration the temporal dependence present in the rows of $\mathbf{X}$. In the remainder of the proof,
		we use $p$ (instead of $p_2$) to denote generically the dimension of the process. 
		
		Let $S^p=\{u\in\mathbb{R}^p|\|u\|=1\}$ denote the $p$-dimensional unit sphere. Then, $\Lambda_{\max}(S_\mathbf{X})$ is the operator norm of $S_{\mathbf{X}}$, which has the following variational representation form:
		\begin{equation*}
		\Lambda_{\max}(S_\mathbf{X}) = \frac{1}{n}\vertiii{\mathbf{X}'\mathbf{X}}_{\text{op}} = \frac{1}{n} \sup\limits_{u\in S^{p}}u'\mathbf{X}'\mathbf{X}u.
		\end{equation*} 
		For a positive scalar $s$, define
		\begin{equation*}
		\Psi(s) := \sup\limits_{u\in sS^{p_1}}\langle \mathbf{X} u,  \mathbf{X}u\rangle;
		\end{equation*}
		the goal is to establish an upper bound for $\Psi(1)/n$. Let $\mathcal{A}=\{u^1,\cdots,u^A\}$ denote the $1/4$ covering of $S^{p}$. \citet{negahban2011estimation} established that
		\begin{equation*}
		\Psi(1)\leq 4\max\limits_{u^a\in\mathcal{A}}\langle \mathbf{X} u^a,  \mathbf{X}u^a\rangle;
		\end{equation*} 
		further, according to \citet{anderson2011statistical}, there exists a $1/4$ covering of $S^{p}$ with at most $|\mathcal{A}|\leq 8^{p}$ elements. Consequently, 
		\begin{equation*}
		\mathbb{P}\Big[ \big|\frac{1}{n}\Psi(1)\big|\geq 4\delta  \Big] \leq 8^{p} \max\limits_{u^a} \mathbb{P}\Big[\frac{|(u^a)'\mathbf{X}\mathbf{X}(u^a)|}{n} \geq \delta\Big].
		\end{equation*}
		What remains to be bounded is $\frac{1}{n}u'\mathbf{X}'\mathbf{X}u$, for an arbitrary $u\in S^{p}$. By \citet[Proposition 2.4(b)]{basu2015estimation}, we have
		\begin{equation*}
		\mathbb{P}\Big[ \Big|u'\Big(\big(\frac{\mathbf{X}'\mathbf{X}}{n}\big) - \Sigma_X(0)\Big)u\Big| > 2\pi\mathcal{M}(f_X)\eta \Big] \leq 2\exp\left(-cn\min\{\eta,\eta^2\}\right),
		\end{equation*}
		and thus 
		\begin{equation*}
		\mathbb{P}\Big[ u'\big(\frac{\mathbf{X}'\mathbf{X}}{n}\big)u > 2\pi\mathcal{M}(f_X)\eta + \smalliii{\Sigma_X(0)}_{\op}\Big] \leq 2\exp\left(-cn\min\{\eta,\eta^2\}\right).
		\end{equation*}
		Therefore, it follows that
		\begin{equation*}
		\begin{split}
		\mathbb{P}\Big[ \big|\frac{1}{n}\Psi(1)\big|\geq 8\pi\mathcal{M}(f_X)\eta + 4\vertiii{\Sigma_{X}(0)}_{\text{op}}\Big] \leq 2\exp\big(p\log 8 - cn\min\{\eta,\eta^2\}\big).
		\end{split}
		\end{equation*}
		With the specified choice of sample size $n$, the probability vanishes by choosing $\eta = c'_0$ for constant $c'_0$ sufficiently large. Finally, by Proposition 2.3 in \citet{basu2015estimation}, $\vertiii{\Sigma_{X}(0)}_{\text{op}}\leq 2\pi\mathcal{M}(f_X)$, and thus the conclusion in Lemma~\ref{lemma:Smax} holds.
	\end{proof}
	
	\begin{proof}[Proof of Lemma~\ref{lemma:RSCZ}]
		It suffices to show that the following inequality holds with high probability for some curvature $\alpha^{\widehat{\mathbf{Z}}}_{\RSC}>0$ and tolerance $\tau_{\mathbf{Z}}$, where we define $\widehat{S}_{\mathbf{Z}}:=\tfrac{1}{n}\widehat{\mathbf{Z}}_{n-1}^\top\widehat{\mathbf{Z}}_{n-1}$:
		\begin{equation*}
		\frac{1}{2}\theta^\top\widehat{S}_{\mathbf{Z}}\theta \geq \frac{\alpha^{\widehat{\mathbf{Z}}}_{\RSC}}{2}\|\theta\|^2 - \tau_{\mathbf{Z}}\|\theta\|_1^2, \qquad \forall~\theta\in\R^{p}.
		\end{equation*}
		Define $S_{\mathbf{Z}} :=\tfrac{1}{n}\mathbf{Z}_{n-1}^\top\mathbf{Z}_{n-1}$, then $\widehat{S}_{\mathbf{Z}}$ can be written as
		\begin{equation}\label{eqn:0}
		\widehat{S}_{\mathbf{Z}} = S_{\mathbf{Z}}  + \Big(\frac{1}{n} \mathbf{Z}^\top_{n-1}\Delta_{\mathbf{Z}_{n-1}} + \frac{1}{n}\Delta_{\mathbf{Z}_{n-1}}^\top\mathbf{Z}_{n-1} \Big)+ \Big(\frac{1}{n}\Delta_{\mathbf{Z}_{n-1}}^\top\Delta_{\mathbf{Z}_{n-1}}\Big),
		\end{equation}
		First, notice that the last term satisfies the following natural lower bound {\em deterministically}, since  $\Delta_{\mathbf{F}}$ is assumed non-random and $\Delta_{\mathbf{Z}}=[\Delta_{\mathbf{F}},O]$:
		\begin{equation*}
		\theta^\top \Big(\frac{1}{n}\Delta_{\mathbf{Z}_{n-1}}^\top\Delta_{\mathbf{Z}_{n-1}}\Big) \theta \geq 0\qquad \forall~\theta\in\R^p,
		\end{equation*} 
		which however, does not contribute to the ``positive" part of curvature. For the first two terms, we adopt the following strategy, using Lemma 12 in \citet{loh2012high} as an intermediate step. Specifically, \citet[Lemma 12]{loh2012high} proves that for any fixed generic matrix $\Gamma\in\R^{p\times p}$ that satisfies $|\theta^\top\Gamma\theta|\leq \delta$ for any $\theta\in\mathbb{K}(2s)$\footnote{$\mathbb{K}(2s):=\{ \theta: \|\theta\|_0=2s \}$ is the set of $2s$-sparse  vectors.}, the following bound holds
		\begin{equation}\label{eqn:Loh}
		|\theta^\top\Gamma\theta|\leq 27\delta \big( \|\theta\|_2^2 + \frac{1}{s}\|\theta\|_1^2\big), ~~~~~\forall~\theta\in\R^p.
		\end{equation}
		Then, based on~\eqref{eqn:Loh}, consider $\Gamma=\widehat{\Gamma}-\Sigma$ then rearrange terms, so that $\theta^\top \widehat{\Gamma}\theta \geq \theta^\top \Sigma \theta - \frac{27\delta}{2}\big(\|\theta\|_2^2 + \tfrac{1}{2}\|\theta\|_1^2 \big)$. The RE condition follows by setting $\delta$ to be some quantity related to $\Lambda_{\min}(\Sigma)$.

		In light of this, for the first two terms in~\eqref{eqn:0}, let
		\begin{equation*}
		\Psi:=S_{\mathbf{Z}} + \Big(\frac{1}{n} \mathbf{Z}^\top_{n-1}\Delta_{\mathbf{Z}_{n-1}} + \frac{1}{n}\Delta_{\mathbf{Z}_{n-1}}^\top\mathbf{Z}_{n-1} \Big),
		\end{equation*}
		denote their sum, in order to obtain an upper bound for $\big|\theta^\top\big(\Psi-\Sigma_Z(0)\big)\theta\big|$, so that Lemma 12 in \citet{loh2012high} can be applied. To this end, since 
		\begin{equation*}
		\Big| \theta^\top\big[ \Psi-\Sigma_Z(0)\big] \theta \Big| \leq \Big| \theta'(S_{\mathbf{Z}}-\Sigma_Z(0))\theta \Big| + \Big| \theta'\Big(\frac{1}{n} \mathbf{Z}'_{n-1}\Delta_{\mathbf{Z}_{n-1}} + \frac{1}{n}\Delta_{\mathbf{Z}_{n-1}}'\mathbf{Z}_{n-1} \Big)\theta \Big|,
		\end{equation*}
		we consider getting upper bounds for each of the two terms:
		\begin{equation*}
		\text{(i)}~~\Big| \theta'(S_{\mathbf{Z}}-\Sigma_Z(0))\theta \Big|~,\qquad \text{(ii)}~~\Big| \theta'\Big(\frac{1}{n} \mathbf{Z}'_{n-1}\Delta_{\mathbf{Z}_{n-1}} + \frac{1}{n}\Delta_{\mathbf{Z}_{n-1}}'\mathbf{Z}_{n-1} \Big)\theta \Big|.
		\end{equation*}	
		For (i), we follow the derivation in \citet[Proposition 2.4(a)]{basu2015estimation}, that is, for all $\|\theta\|\leq 1$, 
		\begin{equation*}
		\mathbb{P}\Big[ \Big| \theta'\big( S_{\mathbf{Z}}-\Sigma_Z(0)\big)\theta\Big| > 2\pi \mathcal{M}(f_Z)\eta\Big]\leq 2\exp\big[-cn\min\{\eta^2,\eta\}\big],
		\end{equation*}
		and further with probability at least $$1-2\exp\big(-cn \min\{\eta^2,\eta\} + 2s\min\{\log p,\log(21ep/2s)\} \big),$$ the following bound holds:
		\begin{equation}\label{eqn:2s-sparse}
		\sup\limits_{\theta\in\mathbb{K}(2s)} \Big| \theta^\top\big( S_{\mathbf{Z}}-\Sigma_Z(0)\big)\theta \Big| < 2\pi\mathcal{M}(f_Z)\eta.
		\end{equation}
		For (ii), the two terms are identical, with either one given by
		\begin{equation*}
		\frac{1}{n}(\mathbf{Z}_{n-1}\theta)^\top(\Delta_{\mathbf{Z}_{n-1}}\theta).
		\end{equation*}
		To obtain its upper bound, consider the following inequality, based on which we bound the two terms in the product separately:
		\begin{equation}\label{eqn:two-term}
		\sup\limits_{\theta\in\mathbb{K}(2s)}\Big|\frac{1}{n}\langle \mathbf{Z}_{n-1}\theta, \Delta_{\mathbf{Z}_{n-1}}\theta\rangle\Big|\leq \Big(\sup\limits_{\theta\in\mathbb{K}(2s)}\smallii{\frac{\mathbf{Z}_{n-1}\theta}{\sqrt{n}}}\Big) \Big(\sup\limits_{\|\theta\|\leq 1}\smallii{\frac{\Delta_{\mathbf{Z}_{n-1}}\theta}{\sqrt{n}}}\Big).
		\end{equation}
		For the first term in~\eqref{eqn:two-term}, since rows of $\mathbf{Z}_{n-1}$ are time series realizations from~\eqref{md:FAVAR-VAR1}, then if we let $\xi:=\mathbf{Z}_{n-1}\theta$, $\xi\sim\mathcal{N}(0_{n\times 1},Q_{n\times n})$ is Gaussian with $Q_{st} = \theta'\Sigma_Z(t-s)\theta$. To get its upper bound, we bound its square, and use again~\eqref{eqn:2s-sparse}, that is, 
		{\small\begin{equation*}
			\sup\limits_{\theta\in\mathbb{K}(2s)} \Big| \theta^\top\Big(\frac{1}{n}\mathbf{Z}_{n-1}^\top\mathbf{Z}_{n-1} \Big)\theta \Big| \leq  \sup\limits_{\theta\in\mathbb{K}(2s)}\theta'\Sigma_Z(0)\theta + 2\pi \mathcal{M}(f_Z)\leq 2\pi \mathcal{M}(f_Z) + 2\pi \mathcal{M}(f_Z)\eta.
			\end{equation*}}%
		For the second term $\|\Delta_{\mathbf{Z}_{n-1}}\theta/\sqrt{n}\|$, this is non-random, and for all $\|\theta\|\leq 1$, $\|\Delta_{\mathbf{Z}_{n-1}}\theta/\sqrt{n}\|\leq \Lambda^{1/2}_{\max}\big( S_{\Delta_{\mathbf{Z}_{n-1}}}\big) = \Lambda^{1/2}_{\max}\big( S_{\Delta_{\mathbf{F}_{n-1}}}\big)$. Therefore, the following bound holds for~\eqref{eqn:two-term}:
		{\small\begin{equation}\label{eqn:bound2nd}
			\sup\limits_{\theta\in\mathbb{K}(2s)}\Big|\frac{1}{n}\langle \mathbf{Z}_{n-1}\theta, \Delta_{\mathbf{Z}_{n-1}}\theta\rangle\Big|\leq \Lambda^{1/2}_{\max}\big(S_{\Delta_{\mathbf{F}_{n-1}}}\big)\sqrt{2\pi\mathcal{M}(f_Z) + 2\pi\mathcal{M}(f_Z)\eta}.
			\end{equation}}%
		Combine~\eqref{eqn:2s-sparse} and~\eqref{eqn:bound2nd} that are respectively the bounds for (i) and (ii), and the following bound holds with probability at least $1-2\exp\big(-cn \min\{\eta^2,\eta\} + 2s\min\{\log p,\log(21ep/2s)\} \big)$:
		{\small\begin{equation}\label{eqn:bound12}
			\sup\limits_{\theta\in\mathbb{K}(2s)}\Big| \theta^\top \Big(\Psi - \Sigma_Z(0)\Big) \theta \Big| \leq 2\pi\mathcal{M}(f_Z)\eta + 2\Lambda^{1/2}_{\max}\big( S_{\Delta_{\mathbf{F}_{n-1}}} \big)\sqrt{2\pi\mathcal{M}(f_Z) + 2\pi\mathcal{M}(f_Z)\eta}.
			\end{equation}}%
		Now applying \citet[Lemma 12]{loh2012high} to $\Gamma=\Psi - \Sigma_Z(0)$, and $\delta$ being the RHS of~\eqref{eqn:bound12}, then the following bound holds:
		\begin{equation*}
		\theta^\top\widehat{S}_{\mathbf{Z}}\theta \geq 2\pi\mathfrak{m}(f_Z) \|\theta\|_2^2 - 27\delta(\|\theta\|_2^2 + \frac{1}{s}\|\theta\|_1^2) = \big(2\pi\mathfrak{m}(f_Z) -27\delta\big)\|\theta\|^2 - \frac{27\delta}{s}\|\theta\|_1^2.
		\end{equation*}
		By setting $\eta=\omega^{-1}:=\frac{\mathfrak{m}(f_Z)}{54\mathcal{M}(f_Z)}$, 
		\begin{equation*}
		\begin{split}
		\delta = \frac{\pi}{27}\mathfrak{m}(f_Z) + 2 \Lambda^{1/2}_{\max}\big( S_{\Delta_{\mathbf{F}_{n-1}}} \big)&\sqrt{2\pi \mathcal{M}(f_Z) + \pi\mathfrak{m}(f_Z)/27} \\
		&\leq \frac{\pi}{27}\mathfrak{m}(f_Z) + 2 \Lambda^{1/2}_{\max}\big( S_{\Delta_{\mathbf{F}_{n-1}}} \big)\sqrt{\frac{55\pi}{27}\mathcal{M}(f_Z)}.
		\end{split}
		\end{equation*}
		Since we have required that $\mathfrak{m}(f_Z)/\mathcal{M}^{1/2}(f_Z) > c_0\cdot\Lambda^{1/2}_{\max}(S_{\Delta_{\mathbf{F}_{n-1}}})$ with $c_0\geq 6\sqrt{165\pi}$, $2\pi\mathfrak{m}(f_Z) - 27\delta > 0$. Therefore, the RSC condition is satisfied with curvature
		\begin{equation*}
		\alpha^{\widehat{\mathbf{Z}}}_{\RSC} = 2\pi\mathfrak{m}(f_Z) - 27\delta = \pi\mathfrak{m}(f_Z) - 54\Lambda^{1/2}_{\max}\big(S_{\Delta_{\mathbf{F}_{n-1}}}\big) \sqrt{2\pi\mathcal{M}(f_Z) + \pi\mathfrak{m}(f_Z)/27} > 0, 
		\end{equation*}
		and tolerance $27\delta/(2s)$, with probability at least $1-2\exp\Big(-cn\omega^{-2}+ 2s\log p\Big)$. Finally, set $s=\lceil {cn\omega^{-1}/4\log p}\rceil$, we get the desired conclusion. 
	\end{proof}
	
	\begin{proof}[Proof of Lemma~\ref{lemma:deviation-A}]
		First, we note that the quantity of interest can be upper bounded by the following four terms:
		{\small \begin{align}
			\frac{1}{n}\|\widehat{\mathbf{Z}}^\top_{n-1} \Big( \widehat{\mathbf{Z}}_n -&  \widehat{\mathbf{Z}}_{n-1}(A^\star)^\top\Big)\|_\infty \notag\\
			& =  \frac{1}{n}\smallii{ \Big(\mathbf{Z}_{n-1} + \Delta_{\mathbf{Z}_{n-1}}\Big)^\top \Big( \mathbf{W} + \Delta_{\mathbf{Z}_{n}} - \Delta_{\mathbf{Z}_{n-1}}(A^\star)^\top  \Big) }_\infty & \notag\\
			& \leq \smallii{ \frac{1}{n}\mathbf{Z}_{n-1}^\top \mathbf{W}}_\infty + \smallii{ \frac{1}{n}\Delta_{\mathbf{Z}_{n-1}}^\top \mathbf{W}}_\infty + \smallii{ \frac{1}{n}\mathbf{Z}^\top_{n-1}\Big( \Delta_{\mathbf{Z}_{n}} - \Delta_{\mathbf{Z}_{n-1}}(A^\star)^\top  \Big) }_\infty  \notag \\
			& \quad + \smallii{ \frac{1}{n}\Delta_{\mathbf{Z}_{n-1}}^\top\Big( \Delta_{\mathbf{Z}_{n}} - \Delta_{\mathbf{Z}_{n-1}}(A^\star)^\top  \Big) }_\infty  \notag\\
			& := T_1 + T_2 + T_3 + T_4. \label{eqn:deviation}
			\end{align}}%
		We provide bounds on each term in~\eqref{eqn:deviation} sequentially. $T_1$ is the standard Deviation Bound, which according to previous derivations (e.g., \citet{basu2015estimation} for the expression specifically derived for $\VAR(1)$) satisfies 
		\begin{equation*}
		\frac{1}{n}\smallii{\mathbf{Z}_{n-1}^\top\mathbf{W}}_\infty \leq c_0\big[\mathcal{M}(f_Z) + \mathcal{M}(f_W) + \mathcal{M}(f_{Z,W^+})\big]\sqrt{\frac{\log (p_1+p_2)}{n}}
		\end{equation*}
		with probability at least $1-c_1\exp(-c_2\log (p_1+p_2))$ for some $\{c_i\}$. For $T_2$, since rows of $\mathbf{W}$ are iid realizations from $\mathcal{N}(0,\Sigma_w)$, then for $\Delta_{\mathbf{Z}_{n-1}}^\top \mathbf{W} \in \R^{(p_1+p_2)\times(p_1+p_2)}$ which has at most $p_1\times (p_1+p_2)$ nonzero entries, each entry $(i,j)$ given by
		\begin{equation*}
		\kappa_{ij}:=\big(\frac{1}{n}\Delta_{\mathbf{Z}_{n-1}}^\top \mathbf{W}\big)_{ij} = \frac{1}{n}\Delta_{\mathbf{Z}_{n-1},\cdot i}^\top \mathbf{W}_{\cdot j}
		\end{equation*}
		is Gaussian, and the following tail bound holds:  
		\begin{equation*}
		\begin{split}
		\mathbb{P}\big[ |\kappa_{ij}| \geq t  \big]& \leq e\cdot \exp\Big( -\frac{cnt^2}{\Lambda_{\max}(\Sigma_w)\max\limits_{i\in\{1,\dots,p_1+p_2\}}\|\Delta_{\mathbf{Z}_{\cdot i}}/\sqrt{n}\|^2_2 }\Big) \\
		& =  e\cdot \exp\Big( -\frac{cnt^2}{\Lambda_{\max}(\Sigma_w)\max\limits_{i\in\{1,\dots,p_1\}}\|\Delta_{\mathbf{F}_{\cdot i}}/\sqrt{n}\|^2_2 }\Big).
		\end{split}
		\end{equation*}
		Taking the union bound over all $p_1\times (p_1+p_2)$ nonzero entries, the following bound holds:
		\begin{equation*}
		\mathbb{P}\Big[\frac{1}{n}\smallii{\Delta_{\mathbf{Z}_{n-1}}^\top \mathbf{W}}_\infty \geq t\Big] \leq \exp\Big(- \frac{cnt^2}{\Lambda_{\max}(\Sigma_w)\max\limits_{i\in\{1,\dots,p_1\}}\|\Delta_{\mathbf{F}_{\cdot i}}/\sqrt{n}\|^2_2 } + \log\big(ep_1(p_1+p_2)\big)\Big).
		\end{equation*}
		Choose $t = c_0\big(\Lambda^{1/2}_{\max}(\Sigma_w)\max\limits_{i=1,\dots,p_1}\|\Delta_{\mathbf{F}_{\cdot i}}/\sqrt{n}\|\big) \sqrt{\frac{\log (p_1(p_1+p_2))}{n}}$, the following bound holds with probability at least $1-\exp\Big(-c_1\log\big(p_1(p_1+p_2)\big)\Big)$:
		\begin{equation*}
		\frac{1}{n}\smallii{\Delta_{\mathbf{Z}_{n-1}}^\top \mathbf{W}}_\infty \leq c_0\Lambda^{1/2}_{\max}(\Sigma_w)\max\limits_{i=1,\dots,p_1}\|\Delta_{\mathbf{F}_{\cdot i}}/\sqrt{n}\|\sqrt{\frac{\log p_1 + \log(p_1+p_2)}{n}}.
		\end{equation*}
		For $T_3$, let $\varepsilon_n := \Delta_{\mathbf{Z}_{n}} - \Delta_{\mathbf{Z}_{n-1}}(A^\star)^\top = [\Delta_{\mathbf{F}_n}-\Delta_{\mathbf{F}_{n-1}}(A_{11}^\star)^\top, -\Delta_{\mathbf{F}_{n-1}}(A_{21}^\star)^\top ]$, then each entry of $\tfrac{1}{n}\mathbf{Z}_{n-1}^\top \varepsilon_n$ is given~by
		\begin{equation*}
		\Big(\tfrac{1}{n}\mathbf{Z}_{n-1}^\top \varepsilon_n\Big)_{ij} = \frac{1}{n}\mathbf{Z}_{n-1,\cdot i}^{\top} \varepsilon_{n,\cdot j},
		\end{equation*}
		and it has $(p_1+p_2)\times (p_1+p_2)$ entries. Next, note that column $i$ of $\mathbf{Z}_{n-1}\in\R^n$ can be viewed as a mean-zero Gaussian random vector with covariance matrix $Q^i$ where $(Q^i)_{st} = [\Sigma_Z(t-s)]_{ii}$ satisfying $\Lambda_{\max}(Q^i)\leq \Lambda_{\max}(\Sigma_Z(0))\leq 2\pi\mathcal{M}(f_Z)$, so for any $(i,j)$, $\big(\tfrac{1}{n}\mathbf{Z}_{n-1}^\top \varepsilon_n\big)_{ij}$ satisfies  
		\begin{equation*}
		\mathbb{P}\Big[\big|\big(\tfrac{1}{n}\mathbf{Z}_{n-1}^\top \varepsilon_n\big)_{ij}\big|>t\Big] \leq \exp\Big( 1- \frac{cnt^2}{\Lambda_{\max}(\Sigma_Z(0))\max\limits_{j\in\{1,\dots,p_1\}}\|\varepsilon_{n,\cdot j}/\sqrt{n}\|^2}\Big).
		\end{equation*}
		Again by taking the union bound over all $(p_1+p_2)^2$ entries, and let $$t=c_0\big(2\pi\mathcal{M}(f_Z)\big)^{1/2}\max\limits_{j\in\{1,\dots,p_1\}}\|\varepsilon_{n,\cdot j}/\sqrt{n}\|\sqrt{\frac{\log p_1 + \log (p_1+p_2)}{n}},$$ the following bound holds w.p. at least $1-\exp(-c_1\log(p_1+p_2))$:
		\begin{equation*}
		\begin{split}
		\frac{1}{n}\|\mathbf{Z}_{n-1}^\top &\Big( \Delta_{\mathbf{Z}_{n}} - \Delta_{\mathbf{Z}_{n-1}}(A^\star)^\top\Big) \|_\infty \\
		&\leq c_0\big(2\pi\mathcal{M}(f_Z)\big)^{1/2}\max\limits_{j\in\{1,\dots,(p_1+p_2)\}}\|\varepsilon_{n,\cdot j}/\sqrt{n}\|\sqrt{\frac{\log (p_1+p_2)}{n}}.
		\end{split}
		\end{equation*}
		For $T_4$, it is deterministic, and satisfies
		\begin{equation*}
		\begin{split}
		\frac{1}{n}\smallii{ \Delta_{\mathbf{Z}_{n-1}}^\top\Big( \Delta_{\mathbf{Z}_{n}} - \Delta_{\mathbf{Z}_{n-1}}(A^\star)^\top  \Big) }_\infty &\leq 
		\|\frac{1}{n}\Delta_{\mathbf{Z}_{n-1}}^\top \Delta_{\mathbf{Z}_{n}}\|_\infty + \|\frac{1}{n}\Delta_{\mathbf{Z}_{n-1}}^\top \Delta_{\mathbf{Z}_{n-1}}(A^\star)^\top \|_\infty \\
		& = \|\frac{1}{n}\Delta_{\mathbf{F}_{n-1}}^\top \Delta_{\mathbf{F}_{n}}\|_\infty + \|\frac{1}{n}\Delta_{\mathbf{F}_{n-1}}^\top \Delta_{\mathbf{F}_{n-1}}(A_{11}^\star)^\top \|_\infty
		\end{split}
		\end{equation*}
		Combine all terms, and there exist some constant $C_1,C_2,C_3$ and $c_1,c_2$ such that with probability at least $1-c_1\exp\big(-c_2\log (p_1+p_2)\big)$, the bound in~\eqref{eqn:deviationCnp} holds.
	\end{proof}
	
	\section{Generalization of the Main Results to Sub-exponential Tailed Error Processes: a Sketch}\label{appendix:subexponential}
	
	In this section, we provide the counterpart of Theorem~\ref{thm:info} for the case where the underlying processes are linear with generalized sub-exponential tails. Specifically, the stable joint VAR process $Z_t=(F_t',X_t)'$ has the following moving average representation with absolutely summable coefficients $B_\ell$'s (c.f. \citet{rosenblatt2012stationary}):
	\begin{equation*}
	Z_t = \sum_{\ell=0}^\infty B_{\ell}w_{t-\ell}. 
	\end{equation*}
	In the case where the process is Gaussian, the $w_{t}$'s correspond to Gaussian white noise processes. Throughout this section, we relax the Gaussian assumption and assume $w_t$ is a white noise process whose coordinates have the following $\alpha$-sub-exponential tail decay, that is, there exist two constants $a,b$ such that the following holds:
	\begin{equation}\label{eqn:subexp}
	\mathbb{P}\Big( |w_{tj}| \geq \xi \Big) \leq a\exp(-b\xi^\alpha),~~\forall~\xi>0.
	\end{equation}
	Specifically, the case of sub-Gaussian tails corresponds to $\alpha=2$, whereas for $\alpha\in(0,1]$ it leads to distributions with heavier tails, such as the sub-exponential distribution ($\alpha=1$) or the Weibull distribution; see also \citet{erdHos2012bulk,gotze2019concentration}. As a consequence, $X_t$ and $F_t$ deviate from being Gaussian due to the recursive data generating mechanism. Additionally, we assume the noise term of the calibration equation $e_t$ comes from the same $\alpha$-sub-exponential family.
	
	\begin{proposition}[High probability error bounds for $\widehat{\Theta}$ and $\widehat{\Gamma}$]\label{thm:info-sube} 
		Suppose we are given some randomly observed snapshots $\{x_1,\dots,x_n\}$ and $\{y_1,\dots,y_n\}$ obtained from the stable processes $X_t$ and $Y_t$, whose dynamics are described in~\eqref{md:FAVAR-VAR1} and~\eqref{md:FAVAR-info}. 
		Assume that the same conditions as in Theorem~\ref{thm:info} hold. Then, there exist universal positive constants $\{C_i\}$ and $\{c_i\}$ such that by solving~\eqref{opt:solveTheta} with regularization parameter 
		\begin{equation}
		\begin{split}
		\lambda_\Gamma = \max\Big\{ C_1(2\pi\mathcal{M}(f_X)+\Lambda_{\max}(\Sigma_e))\dfrac{(\log p_2 + \log q)^{1/\alpha}}{\sqrt{n}},
		~C_2\phi/\sqrt{nq},~C_3\Lambda^{1/2}_{\max}(\Sigma_e)  \Big\},
		\end{split}
		\end{equation}
		the solution $(\widehat{\Theta},\widehat{\Gamma})$ has the following bound with probability at least $1-c_1\exp\{-c_2\big(\log(p_2q)\big)^{2/\alpha}\}$:
		\begin{equation}\label{bound:sube}
		\smalliii{\Delta_{\Theta}/\sqrt{n}}_{\F}^2 +\smalliii{\Delta_{\Gamma}}_{\F}^2 \lesssim \frac{\lambda^2_\Gamma}{\mathfrak{m}(f_X)} \psi(s_{\Gamma^\star},p_1,r),
		\end{equation}
		for a sufficiently large sample size and some function $\psi(\cdot)$ that depends linearly on $s_{\Gamma^\star}, p_1$ and $r$.
	\end{proposition}
	Note that the bounds for each individual probabilistic event (e.g., RSC condition, deviation bound) differ from those in the Gaussian case, although their expressions in~\eqref{bound:sube} do not exhibit marked differences compared to the Gaussian case; specifically, the bound for $\smalliii{\Delta_{\Theta}/\sqrt{n}}_{\F}^2 +\smalliii{\Delta_{\Gamma}}_{\F}^2$ is governed by the more stringent sample size requirement amongst its building components (i.e., concentration in the operator norm) and the slowest term in terms of probability decay.
	
	\bigskip
	In the rest of this section, we sketch the statements and proofs for key lemmas that underlie the high probability statements, assuming $\alpha$-sub-exponential tail decay where $\alpha\in(0,1]\cup\{2\}$. In particular, one can verify that the rates obtained below would coincide with the Gaussian case, if $\alpha=2$. Similar arguments can be applied to the Stage II estimate to arrive at the counterpart of Theorem~\ref{thm:A}, which are omitted.
	
	\medskip
	Lemmas~\ref{lemma:auxaux} generalizes Hanson-Wright type concentration inequality to samples of ${X_t}$.
	\begin{lemma}\label{lemma:auxaux} Consider some generic $p$-dimensional linear process given in the form of $X_t := \sum_{\ell=0}^\infty \Phi_{\ell}u_{t-\ell}$, where $u_t$ is i.i.d coming from the $\alpha$-sub-exponential family defined in~\eqref{eqn:subexp}. Denote its realization by $\mathbf{X}\in\R^{n\times p}$ with $n$ consecutive observations stacked in its rows. Then for a deterministic $np\times np$ matrix $A$, there exists some constant $C$ such that the following bound holds: 
		\begin{equation}\label{eqn:boundTalpha}
		\mathbb{P}\Big( \Big|  \text{vec}(\mathbf{X}^\top)^\top A\,\text{vec}(\mathbf{X}^\top) - \mathbb{E}\big[  \text{vec}(\mathbf{X}^\top)^\top A\, \text{vec}(\mathbf{X}^\top)\big] \Big| > 2\pi\eta \mathcal{M}(f_X)\Big) \leq \mathcal{T}(\eta,\alpha,A),
		\end{equation}
		where
		\begin{equation}\label{eqn:Talpha}
		\mathcal{T}(\eta,\alpha,A) := 2\exp\Big[-C\min\Big\{ \frac{\eta^2}{\text{rk}(A)\smalliii{A}^2_{\op}},\Big(\frac{\eta}{\smalliii{A}_{\op}}\Big)^\frac{\alpha}{2} \Big\}  \Big].
		\end{equation} 
	\end{lemma}
	\begin{proof}
		Let $\text{vec}(\mathbf{X}^\top) \stackrel{d}{=} \Omega^{1/2} Z$ where $\Omega$ is the covariance matrix of the $np$-dimensional random vector $\text{vec}(\mathbf{X}^\top)$ and $Z$ satisfies $\mathbb{E}Z=0,\mathbb{E}(ZZ^\top)=\mathrm{I}_{np}$. Applying \citet[Proposition 1.1]{gotze2019concentration} gives
		\begin{equation}\label{bound00}
		\begin{split}
		\mathbb{P}\Big( \Big|  \text{vec}(\mathbf{X}^\top)^\top A\,\text{vec}(\mathbf{X}^\top) &- \mathbb{E}\big[  \text{vec}(\mathbf{X}^\top)^\top A\, \text{vec}(\mathbf{X}^\top)\big] \Big| > 2\pi\eta \mathcal{M}(f_X)\Big)  \\
		&= \mathbb{P}\Big( \Big| Z^\top \Omega^{1/2} A\Omega^{1/2}Z - \mathbb{E}\big[ Z^{\top}\Omega^{1/2}A\Omega^{1/2}Z\big] \Big| > 2\pi\eta\mathcal{M}(f_X)\Big)\\
		&\leq 2\exp\Big\{-c_0 \cdot \nu\Big(\Omega^{1/2}A\Omega^{1/2},\alpha,2\pi\eta\mathcal{M}(f_X)\Big)\Big\} 
		\end{split}
		\end{equation}
		where 
		\begin{equation}\label{eqn:nu}
		\nu(A,\alpha,t):= \min\Big\{ \frac{t^2}{M^4\smalliii{A}^2_{\F}}, \Big( \frac{t}{M^2\smalliii{A}_{\op}}\Big)^{\alpha/2}\Big\};
		\end{equation}
		both $c_0$ and $M$ are constants that depend on $a,b$. Next, we consider the bounds for various norms of $\Omega^{1/2}A\Omega^{1/2}$:
		\begin{itemize}
			\setlength{\itemindent}{-1.5em}
			\small
			\item[--] $\smalliii{\Omega^{1/2}A\Omega^{1/2}}_{\op} \leq \smalliii{\Omega}_{\op}\smalliii{A}_{\op}\leq 2\pi\mathcal{M}(f_X)\smalliii{A}_{\op}$ where the last inequality follows from \citet[Proposition 2.3]{basu2015estimation} which applies to general linear processes;
			\item[--] $\smalliii{\Omega^{1/2}A\Omega^{1/2}}_{\F} \leq \sqrt{\text{rk}(\Omega^{1/2}A\Omega^{1/2})}\smalliii{\Omega^{1/2}A\Omega^{1/2}}_{\op} \leq 2\pi\sqrt{\text{rk}(A)}\smalliii{A}_{\op}\mathcal{M}(f_X)$;
		\end{itemize}
		Therefore, the last expression in~\eqref{bound00} can be upper bounded by~\eqref{eqn:Talpha} and the claim in~\eqref{eqn:boundTalpha} follows. 
	\end{proof}

	Lemma~\ref{lemma:importantaux} is a generalization of Proposition 2.4 in \citet{basu2015estimation} to the case where the underlying processes come from the $\alpha$-sub-exponential family.
	\begin{lemma}\label{lemma:importantaux}
		Consider some generic linear processes given in the form of $X_t := \sum_{\ell=0}^\infty \Phi_{\ell}u_{t-\ell}$, where $u_t$ comes from the $\alpha$-sub-exponential family. Let $\Sigma_X(0):=\text{Cov}(X_t,X_t)$. Denote its realization by $\mathbf{X}\in\R^{n\times p}$ and sample covariance by $S:=\frac{1}{n}\mathbf{X}^\top\mathbf{X}$,respectively.
		\begin{enumerate}[(i)]
			\item For unit vectors $v_1$ and $v_2$ satisfying $\|v_1\|\leq 1,\|v_2\|\leq 1$, the following bound holds:
			\begin{equation*}
			\mathbb{P}\Big( |v_1'(S-\Sigma_X(0))v_1| > 2\pi \eta \mathcal{M}(f_X)  \Big)  \leq \mathcal{T}'\big(\eta, \alpha, n\big),
			\end{equation*}
			and
			\begin{equation*}
			\mathbb{P}\Big( |v_1'(S-\Sigma_X(0))v_2| > 6\pi \eta \mathcal{M}(f_X)  \Big)  \leq 2\mathcal{T}'\big(\eta, \alpha, n\big).
			\end{equation*}
			\item Consider the linear process $Z_t:=\sum_{\ell=0}^\infty \Psi_{\ell}w_{t-\ell}\in\R^{q}$ with $w_t$ coming from the same family of distributions as $u_t$ and satisfies $\text{Cov}(X_t,Z_t)=0$; $\mathbf{Z}$ is similarly defined. Then, the following bound holds:
			\begin{equation*}
			\mathbb{P}\Big( |v_1'(\mathbf{X}^\top\mathbf{Z})v_2 | > 2\pi \eta ( \mathcal{M}(f_Z) +  \mathcal{M}(f_Z) + \mathcal{M}(f_{X,Z}) )  \Big)  \leq 3\,\mathcal{T}'\big(\eta, \alpha, n\big),
			\end{equation*}
			where $\mathcal{M}(f_{X,Z})$ is identically defined to the quantity in Section~\ref{sec:theory}.
		\end{enumerate}
		$\mathcal{T}'$ has the following functional form: 
		$$\mathcal{T}'(\eta,\alpha,n) = c_1\exp\big[ -c_2 \min\{ n\eta^2,(n\eta)^{\alpha/2} \}\big], ~~ \text{for some constants $c_1,c_2$}.$$
	\end{lemma}
	\begin{proof}[Proof of Lemma~\ref{lemma:importantaux}]
		First we note that with $A=\mathrm{I}_n$ and the definition of $\mathcal{T}(\eta,\alpha,n)$, the following holds for some constant $C>0$:
		\begin{equation*}
		\mathcal{T}(n\eta,\alpha,A) = 2\exp\Big[ -C \min\{ n\eta^2,(n\eta)^{\alpha/2} \} \Big].
		\end{equation*}
		Let $y_t := v_1^\top X_t$ and $\mathbf{Y} = \mathbf{X}v_1 \in\mathbb{R}^n$ be $n$ consecutive observations of the scalar process $\{y_t\}$, then
		\begin{equation*}
		v_1'Sv_1 \stackrel{d}{=} \frac{1}{n}\mathbf{Y}^\top \mathbf{Y} \qquad \text{and}\qquad  v_1'\Sigma_X(0)v_1 = \mathbb{E}\big[\mathbf{Y}^\top \mathbf{Y}/n\big].
		\end{equation*}
		Apply Lemma~\ref{lemma:auxaux} to process $\{Y_t\}$ with $A=\I_{n}$ (since moment properties are preserved under linear transformations), to obtain
		\begin{equation*}
		\mathbb{P}\Big( |v_1'(S-\Sigma_X(0))v_1| > 2\pi \eta \mathcal{M}(f_Y)  \Big)  = \mathbb{P}\Big( \big| \mathbf{Y}^\top \mathbf{Y} - \mathbb{E}\mathbf{Y}^\top \mathbf{Y} \big| > 2\pi (n\eta) \mathcal{M}(f_Y) \Big) \leq \mathcal{T}'\big(\eta,\alpha,n\big).
		\end{equation*}
		Further, by Lemma C.6 in \citet{sun2018large}, it follows that $\mathcal{M}(f_Y) \leq \|v_1\|^2\mathcal{M}(f_X) = \mathcal{M}(f_X)$; hence, the following bound holds:
		\begin{equation*}
		\mathbb{P}\Big( |v_1'(S-\Sigma_X(0))v_1| > 2\pi \eta \mathcal{M}(f_X)  \Big)  \leq \mathcal{T}'\big(\eta,\alpha,n\big).
		\end{equation*}
		This proves the first part in (i). The rest of the proof follows along similar lines to the derivation of Proposition 2.4 in \citet{basu2015estimation}, and we give an outline without getting into too many details. For $|v_1'(S-\Sigma_X(0))v_2|$, one considers the decomposition
		\begin{equation*}
		2|v_1'(S-\Sigma_X(0))v_2| \leq  |v_1'(S-\Sigma_X(0))v_1| + |v_2'(S-\Sigma_X(0))v_2| + |(v_1+v_2)'(S-\Sigma_X(0))(v_1+v_2)|
		\end{equation*}
		with $\|v_1+v_2\|\leq 2$. Repeating the steps above to each of the three terms yields the desired result. 
		
		\medskip
		For $|v_1'(\mathbf{X}^\top\mathbf{Z})v_2|$, let $\widetilde{y}_t=v_2^\top Z_t$; then 
		$v_1'(\mathbf{X}^\top\mathbf{Z})v_2  = \frac{1}{n}\sum_{t=1}^n y_t\widetilde{y}_t$ and it satisfies the following decomposition
		\begin{equation*}
		\begin{split}
		\frac{2}{n}\sum_{t=1}^n y_t\widetilde{y}_t & = \Big[ \frac{1}{n}\sum_{t=1}^n (y_t + \widetilde{y}_t)^2 - \text{Var}(y_t + \widetilde{y}_t)\Big] - \Big[ \frac{1}{n}\sum_{t=1}^n y^2_t - \text{Var}(y_t )\Big] - \Big[ \frac{1}{n}\sum_{t=1}^n \widetilde{y}_t^2 - \text{Var}(\widetilde{y}_t)\Big] \\
		& =: \big[ \mathbf{G}^\top\mathbf{G} -  \mathbb{E} \mathbf{G}^\top\mathbf{G} \big]  - \big[ \mathbf{Y}^\top\mathbf{Y} -  \mathbb{E} \mathbf{Y}^\top\mathbf{Y} \big]  - \big[ \widetilde{\mathbf{Y}}^\top\widetilde{\mathbf{Y}} -  \mathbb{E} \widetilde{\mathbf{Y}}^\top\widetilde{\mathbf{Y}} \big], 
		\end{split}
		\end{equation*}
		where $\{g_t := y_t + \widetilde{y}_t\}$ is the summation process; $\mathbf{G}$ and $\widetilde{\mathbf{Y}}$ are analogously defined to $\mathbf{Y}$. Repeating the above steps to each term. Note that
		\begin{equation*}
		\mathcal{M}(f_g) \leq \mathcal{M}(f_Z) + \mathcal{M}(f_X) + \mathcal{M}(f_{X,Z}),
		\end{equation*}
		and this completes the proof. 
	\end{proof}
	
	The following lemma considers the deviation bound. Of note, to ensure the deviation bound vanishes, the sample size requirement would be $n\gtrsim (\log p + \log q)^{\frac{2}{\alpha}}$.
	\begin{lemma}[high probability deviation bound]\label{lemma:deviation-sube} There exist positive constants $C$ and $c_i>0$ such that the following deviation bound holds
		\begin{equation*}
		\|\mathbf{X}^\top\mathbf{E}/n\|_\infty \leq C\cdot(\log p+\log q)^{\frac{1}{\alpha}}/\sqrt{n}
		\end{equation*}
		with probability at least 
		\begin{equation*}
		1 - c_1\exp\big\{ - c_2\big(\log (pq)\big)^{2/\alpha} \big\},
		\end{equation*}
		for any random realizations $\mathbf{X}\in\R^{n\times p}$ and $\mathbf{E}\in\R^{n\times q}$, drawn from the linear processes $\{X_t\in\mathbb{R}^p\}$ and $\{\varepsilon_t\in\mathbb{R}^q\}$ that are constructed as linear filters of the white noise processes coming from some $\alpha$-sub-exponential family.
	\end{lemma}
	\begin{proof}[Proof of Lemma~\ref{lemma:deviation-sube}]
		Apply Lemma~\ref{lemma:importantaux}, so that for any standard basis vector $e_k$ and $e_j$, the following holds:
		\begin{equation*}
		\mathbb{P}\Big( |e_k'(\mathbf{X}^\top\mathbf{E})e_j | > 2\pi \eta ( \mathcal{M}(f_X) +  \mathcal{M}(f_\varepsilon) + \mathcal{M}(f_{X,\varepsilon}) )  \Big)  \leq 3\,\mathcal{T}'\big(\eta,\alpha, n\big).
		\end{equation*}
		Taking the union bound across all $pq$ elements, with probability at least $1-3(pq)\mathcal{T}'\big(\eta, \alpha,n\big) = 1 - 3c_1\exp\{-c_2\min\{ n\eta^2,(n\eta)^{\alpha/2}\} +\log(pq)\}$, the following bound holds:
		\begin{equation*}
		\|\mathbf{X}^\top\mathbf{E}/n\|_\infty \leq 2\pi ( \mathcal{M}(f_X) +  \mathcal{M}(f_\varepsilon) + \mathcal{M}(f_{X,\varepsilon}) )\cdot \eta. 
		\end{equation*}
		Set $\eta:=c_0(\log p+\log q)^{\frac{1}{\alpha}}/\sqrt{n}$, the desired result holds for some sufficiently large $c_0$ provided that $n^{\alpha/4}\gtrsim \log(pq)^{(2/\alpha-1/2)}$ (which ensures that $\min\{n\eta^2,(n\eta)^{\alpha/2}\}$). Specifically, in the context of this problem, the most stringent sample size requirement is dictated by the concentration for the operator norm (see Lemma~\ref{lemma:Emax-sube}), and therefore this sample size requirement is automatically fulfilled. 
	\end{proof}
	
	The following lemma verifies the RSC condition. 
	
	\begin{lemma}[Verification of RSC]\label{lemma:RSC-sube} Consider a snapshot of random realizations $\mathbf{X}\in\mathbb{R}^{n\times p}$ drawn from the linear process $X_t := \sum_{\ell=0}^\infty \Phi_{\ell}u_{t-\ell}$ with $u_t$ coming from the $\alpha$-sub-exponential family. Then RSC holds for $\mathbf{X}$ with parameter $\alpha_{\text{RSC}}=\pi\mathfrak{m}(f_X)$ and tolerance $\tau:= c_0\alpha_{\text{RSC}}\log p/(n^{\alpha/2})$, with probability at least $1-c_1\exp\{-c_2n^{\alpha/2}\}$.
	\end{lemma}
	\begin{proof}[Proof of Lemma~\ref{lemma:RSC-sube}]
		Let $S = \frac{1}{n}\mathbf{X}^\top\mathbf{X}$. First, suppose we have
		\begin{equation}\label{eqn:RSC-stronger1}
		\frac{1}{2}v' S v = \frac{1}{2}v'\Big( \frac{\mathbf{X}'\mathbf{X}}{n} \Big)v\geq \frac{\alpha_{\text{RSC}}}{2} \|v\|_2^2 - \tau\|v\|_1^2,\quad \forall~v\in\R^{p};
		\end{equation}
		then, for all $\Delta\in\R^{p_Z\times p_Z}$, and letting $\Delta_j$ denote its $j$th column, the RSC condition automatically holds since
		\begin{equation*}
		\frac{1}{2T}\smalliii{\mathbf{X}\Delta}_{\F}^2 = \frac{1}{2}\sum_{j=1}^{p} \Delta_j'\big( \tfrac{\mathbf{X}'\mathbf{X}}{n} \big) \Delta_j \geq \frac{\alpha_{\text{RSC}}}{2} \sum_{j=1}^{p} \|\Delta_j\|_2^2 - \tau \sum_{j=1}^{p_Z} \|\Delta_j\|_1^2 \geq \frac{\alpha_{\text{RSC}}}{2}\smalliii{\Delta}_{\F}^2 - \tau \|\Delta\|_1^2.
		\end{equation*}
		Therefore, it suffices to verify that~\eqref{eqn:RSC-stronger1} holds. By Lemma~\ref{lemma:importantaux}, $\forall v\in\R^{p},\|v\|\leq 1$ and $\eta >0$,
		\begin{equation*}
		\mathbb{P} \Big[ \big| v'\big(S-\Sigma_X(0)\big) v \big| >  2\pi\mathcal{M}(f_X)\eta \Big]\leq 2\mathcal{T}'(\eta,\alpha,n).
		\end{equation*}
		Applying the discretization argument in \citet[Lemma F.2 \& Lemma F.3]{basu2015estimation}, define $\mathbb{K}(2s):=\{v\in\R^p, \|v\|\leq 1, \|v\|_0\leq 2s\}$, and taking the union bound in this $2s$-sparse cone gives the following inequality:
		\begin{equation}
		\begin{split}
		\mathbb{P} \Big[ \sup\limits_{v\in\mathbb{K}(2s)}\big| v'\big(S-\Sigma_X(0)\big) v \big| >  &2\pi\mathcal{M}(f_X)\eta \Big] \leq 2\cdot \min\{p^s, (21e\cdot p/s)^s\}\cdot \mathcal{T}'(\eta,\alpha,n)  \\
		&= 2c_1\exp\Big[ - c_2\min\{ n\eta^2,(n\eta)^{\alpha/2} \}  + s\min\{ \log p, \log(21ep/s) \} \Big].
		\end{split}
		\end{equation}
		Let $\eta=\mathfrak{m}(f_X)/[54\mathcal{M}(f_X)]$, then apply results from \citet[Lemma 12]{loh2012high} with $\Gamma=S-\Sigma_X(0)$ and $\delta=\pi\mathfrak{m}(f_X)/27$, so that the following holds
		\begin{equation*}
		\frac{1}{2}v' S v \geq \frac{\alpha_{\text{RSC}}}{2}\|v\|^2 - \frac{\alpha_{\text{RSC}}}{2s}\|v\|_1^2, \qquad \text{where}~~\alpha_{\text{RSC}} = \pi\mathfrak{m}(f_X),
		\end{equation*}
		with probability at least $1-2\min\{p^s, (21e\cdot p/s)^s\} \mathcal{T}'(\eta,\alpha,n)$. By letting $s:=c'_0n^{\alpha/2}/\log p$ for some small constant $c_0$, then $\tau$ can be expressed as $\tau = c_0\alpha_{\text{RSC}}\log p/(n^{\alpha/2})$ and the bound holds with probability at least $1-c_1\exp\{-c_2n^{\alpha/2}\}$.
	\end{proof}
	
	\begin{lemma}[High probability bound for $\Lambda_{\max}(S_{\mathbf{E}})$]\label{lemma:Emax-sube} Consider $\mathbf{E}\in\R^{n\times q}$ whose rows are independent realizations drawn from some mean-zero $\alpha$-sub-exponential distribution with covariance $\Sigma_e$. Then, the following holds for some constants $c_i>0$ provided that the sample size satisfies $n^{\alpha/2}\gtrsim q$:
		\begin{equation*}
		\Lambda_{\max}(S_{\mathbf{E}}) \leq c_0\Lambda_{\max}(\Sigma_e),
		\end{equation*}
		with probability at least $1-c_1\exp(-c_2n^{\alpha/2})$.
	\end{lemma}
	\begin{proof}
		The main arguments of the proof follow closely along the lines of those in the proof of Lemma~\ref{lemma:Smax}, while ignoring the temporal dependence. Specifically, using similar covering arguments, with the tail decay as in Lemma~\ref{lemma:importantaux}, there exists some constant $c_i>0$ such that
		\begin{equation*}
		\mathbb{P}\Big[ \Lambda_{\max}(S_{\mathbf{E}}) \geq c_0\eta\smalliii{\Sigma_e}_{\op} + \smalliii{\Sigma_e}_{\op} \Big] \leq c_1 \exp\{  - c_2\min\{ n\eta^2,(n\eta)^{\alpha/2}\}  + q\log 8 \}.
		\end{equation*}
		By choosing $\eta$ to be a sufficiently large constant, with $n^{\alpha/2}\gtrsim q$, the statement in the lemma holds. 
		
		\begin{remark} To ensure concentration of the operator norm, with the specified choice of $\eta$, the sample size requirement in~\eqref{lemma:Emax-sube} is more stringent than that of the Gaussian case. In particular, for the case of sub-exponential tails with $\alpha=1$, this would imply a sample size requirement $\sqrt{n}\gtrsim q$. If however, the elements of the random noise vector $e_t$'s are bounded, that is, $\|e_t\|_2\leq \sqrt{C}$ almost surely for some $C>0$, one can directly apply the matrix Bernstein inequality to obtain the following bound \citep[Corollary 6.20]{wainwright2019high}: 
			\begin{equation*}
			\mathbb{P}\Big[ \Big|\Lambda_{\max}(S_{\mathbf{E}}) - \smalliii{\Sigma_e}_{\op} \Big| \geq \eta \Big] \leq 2q\,\exp\Big\{  \frac{-n\eta^2}{2C(\smalliii{\Sigma_e}_{\op}+\eta)} \Big\}.
			\end{equation*}
			Depending on how $C$ grows with $q$, the sample size requirement could potentially be more relaxed to attain concentration.  
		\end{remark}
		
	\end{proof}
	
	\section{Additional Numerical Studies}\label{appendix:fail}
	
	In this section, we investigate selected scenarios where the relaxed implementation on estimating the calibration equation may fail to produce good estimates, due to the absence of the compactness constraint. For illustration purposes, it suffices to consider the setting where $X_t$ and $F_t$ jointly follow a multivariate Gaussian distribution and are independent
	and identically distributed across samples. Throughout, we set $n=200, p_1=5, p_2=50, q=100$, and $\Big(\begin{smallmatrix}
	X_t \\ F_t \end{smallmatrix}\Big)\sim\mathcal{N}(0,\Sigma)$ with  $\Sigma_{ij}=0.25~(i\neq j)$ and $\Sigma_{ii}=1$. The noise level is fixed at $\sigma_e=1$.
	
	First, we note that based on the performance evaluation shown in Section~\ref{sec:simulation}, the estimates demonstrate good performance even without the compactness constraint. The simulation settings are characterized by adequate sparsity
	in $\Gamma$, which in turn limits the size of the equivalence class $\mathcal{C}(Q_2)$ as mentioned in Section~\ref{sec:id}. Therefore, we focus on the following two issues:  (i) whether sparsity encourages additional ``approximate identification"; and (ii) whether a good initializer helps constrain estimates from subsequent iterations to a ball around the true value.
	
	We start by considering a non-sparse $\Gamma$. Specifically, for both $\Lambda$ and $\Gamma$, their entries are generated from $\mathsf{Unif}\{(-1.5,-1.2)\cup(1.2,1.5)\}$. Additionally, we specify one alternative model in $\mathcal{C}(Q_2)$ by setting $Q_2=\mathbf{5}_{p_1\times p_2}$, which will generate the corresponding $\check{\mathbf{F}}$, $\check{\Theta}$ and $\check{\Gamma}$. Table~\ref{table:nonsparse} depicts the performance of the estimated $\Theta$ based on different initializers:
	\begin{table}[!h]
		\captionsetup{font=small}
		\centering\scriptsize
		\caption{Performance evaluation of $\widehat{\Theta}$ obtained from different initializers under a non-sparse setting.}\label{table:nonsparse}
		\begin{tabular}{l|c|c|c|c}
			\hline
			initializer $\widehat{\Theta}^{(0)}$ & $\Theta^\star$ & $\mathbf{0}_{n\times q}$ & $\Theta^\star+ 0.1*\mathbf{Z}_{n\times q}$ & $\check{\Theta}$ \\ \hline
			Rel.Err & 0.09 & 0.63 & \text{fail to converge within 5000 iterations} & 1.82 (0.02, relative to $\check{\Theta}$)  \\
			\hline
		\end{tabular}	
	\end{table}
	
	\normalsize
	The results in Table~\ref{table:nonsparse} show that the algorithm converges (if at all) to different local optima whose
	values may deviate markedly for the true ones. Specifically, initializer $\Theta^\star + 0.1*\mathbf{Z}_{n\times q}$, where
	each entry $\Theta^\star$ is perturbed by an iid standard Gaussian random variable scaled by 0.1, fails to converge. Note
	that the perturbation is small, but the operator norm of the initializer far exceeds $\phi_0$.
	Initializer $\check{\Theta}$ yields an estimate that is far from the true data-generating factor hyperplane, yet close to its observationally equivalent one. This suggests that in non-sparse settings, without imposing the compactness constraint on the equivalence class, a good initializer is required for the actual relaxed implementation to produce a fairly good estimate of the true data generating parameters.
	
	However, this is not the case if there is sufficient sparsity in $\Gamma$. Specifically, using the same generating mechanism for $\Lambda$ and $\Gamma$ as in Section~\ref{sec:simulation}, we found that even with different initializers, the algorithm always produces estimates that are close to each other and also exhibit good performance. This finding strongly suggests
	that sparsity in $\Gamma$ effectively shrinks the size of the equivalence class and the algorithm after a few iterations
	produces updates that are close to each other, irrespective of the initializer employed. Hence, the effective equivalence class is constrained to the one whose elements are encoded by $\check{\Gamma}$ that have similar characteristics in terms of the location of the non-zero parameters to $\Gamma$. 
	
	Finally, we consider a case that lies between the above two settings, that is, there is a structured sparsity pattern 
	in $\Gamma$. Specifically, we set the last $5$ columns of $\Gamma$ to be dense while the remaining ones are sparse. The overall density level of $\Gamma$ is fixed at 10\%. Note that in this case, the size of the corresponding equivalence class is much larger to the one corresponding to a $\Gamma$ with 10\% uniformly distributed non-zeros entries, due to the presence
	of the five dense columns.
	\begin{table}[!h]
		\captionsetup{font=small}
		\centering\scriptsize
		\caption{Performance evaluation for $\widehat{\Theta}$ with different initializers under structured sparsity.}\label{table:structure-sparse}
		\begin{tabular}{l|c|c|c|c}
			\hline
			initializer $\widehat{\Theta}^{(0)}$ & $\Theta^\star$ & $\mathbf{0}_{n\times q}$ & $\Theta^\star+ 0.1*\mathbf{Z}_{n\times q}$ & $\mathbf{20}_{n\times q}$ \\ \hline
			Rel.Err & 0.65 & 0.65 & 0.65 & 0.68  \\
			\hline
		\end{tabular}
	\end{table}
	
	\normalsize
	As the results in Table~\ref{table:structure-sparse} indicate, when the initializer starts to deviate from the true value,  there exist initializers that would yield inferior estimates. 
	
	In summary, in a non-sparse setting without compactification of the equivalence class, different initializers yield drastically different estimates that are not close enough to the true data-generating model, as expected by the approximate
	(IR+) condition employed. The problem is largely mitigated for sufficiently sparse $\Gamma$, which leads to shrinking the equivalence class. However, an exact characterization of the equivalence class is hard to obtain in practice, since the
	location of the non-zero entries in $\Gamma$ is unknown. 
	
	\section{List of Commodities and Macroeconomic Variables}\label{appendix:list}
	\begin{table}[H]
		\captionsetup{font=small}
		\scriptsize
		\centering
		\caption{List of commodities considered in this study. Data source: International Monetary Fund.}\label{table:commodity} 
		\begin{tabular}{l|l|l}
			\hline
			Commodity & Key & Description \\ \hline
			ALUMINUM & PALUM & Aluminum, 99.5\% minimum purity, LME spot price\\
			COCOA & PCOCO & Cocoa beans, International Cocoa Organization cash price \\
			COFFEE & PCOFFOTM & Coffee, Other Mild Arabicas, International Coffee Organization New York cash price\\
			COPPER & PCOPP & Copper, grade A cathode, LME spot price\\
			COTTON & PCOTTIND & Cotton, Cotton Outlook 'A Index', Middling 1-3/32 inch staple \\
			LEAD & PLEAD & Lead, 99.97\% pure, LME spot price\\
			MAIZE & PMAIZMT & Maize (corn), U.S. No.2 Yellow, FOB Gulf of Mexico, U.S. price\\
			NICKEL & PNICK & Nickel, melting grade, LME spot price\\
			OIL & POILAPSP & Crude Oil (petroleum), simple average of three spot prices\\
			RICE & PRICENPQ & Rice, 5 percent broken milled white rice, Thailand nominal price quote\\
			RUBBER & PRUBB & Rubber, Singapore Commodity Exchange, No. 3 Rubber Smoked Sheets, 1st contract\\
			SOYBEANS & PSOYB & Soybeans, U.S. soybeans, Chicago Soybean futures contract (first contract forward)\\
			SUGAR & PSUGAUSA & Sugar, U.S. import price, contract no.14 nearest futures position\\
			TIN & PTIN & Tin, standard grade, LME spot price \\
			WHEAT & PWHEAMT & Wheat, No.1 Hard Red Winter, ordinary protein\\
			ZINC & PZINC & Zinc, high grade 98\% pure\\
			\hline
		\end{tabular}
	\end{table} 
	
	{\scriptsize
		\setlength\LTleft{-0.6in}
		\setlength\LTright{-1in}
		\begin{longtable}{l|llll}
			\hline
			Name	&	Description	&	tCode	&	Category	&	Region	\\ \hline
			IPI\_US	&	IP Index: total	&	5	&	Output \& Income	&	US	\\
			CUM\_US	&	Capacity Utilization: manufacturing	&	2	&	Output \& Income	&	US	\\
			UNEMP\_US	&	Civilian unemployment rate: all	&	2	&	Labor Market	&	US	\\
			HOUST\_US	&	Housing Starts: ttl new privately owned	&	4	&	Housing	&	US	\\
			ISR\_US	&	Total Business: inventories to sales ratio	&	2	&	Consumption	&	US	\\
			M2\_US	&	M2 Money Stock	&	6	&	Money \& Credit	&	US	\\
			BUSLN\_US	&	Commericial and industrial loans	&	6	&	Money \& Credit	&	US	\\
			REALN\_US	&	Real estate loans at all commercial banks	&	6	&	Money \& Credit	&	US	\\
			FFR\_US	&	Effective federal funds rate	&	2	&	Interest \& Exchange Rates	&	US	\\
			TB10Y\_US	&	10-year treasury rate	&	2	&	Interest \& Exchange Rates	&	US	\\
			BAA\_US	&	Moody's Baa corporate bond yield	&	2	&	Interest \& Exchange Rates	&	US	\\
			USDI\_US	&	Trade weighted U.S.dollar index	&	5	&	Interest \& Exchange Rates	&	US	\\
			CPI\_US	&	CPI: all iterms	&	5	&	Prices	&	US	\\
			PCEPI\_US	&	Personal Consumption Expenditure: chain index	&	5	&	Prices	&	US	\\
			SP500\_US	&	S\&P's Common Stock Price Index: composite	&	5	&	Stock Market	&	US	\\
			CPI\_EU	&	Consumer Price Indices, percent change	&	2	&	Prices	&	EU	\\
			IPI\_EU	&	Industrial Production Index: total industry (excluding construction)	&	5	&	Output \& Income	&	EU	\\
			IPICP\_EU	&	Industrial Production Index: construction	&	5	&	Output \& Income	&	EU	\\
			M3\_EU	&	Monetary aggregate M3	&	6	&	Money \& Credit	&	EU	\\
			LOANRES\_EU	&	Credit to resident sectors, non-MFI excluding gov	&	6	&	Money \& Credit	&	EU	\\
			LOANGOV\_EU	&	Credit to general government sector	&	6	&	Money \& Credit	&	EU	\\
			PPI\_EU	&	Producer Price Index: total industry (excluding construction)	&	6	&	Prices	&	EU	\\
			UNEMP\_EU	&	Unemployment rate: total	&	2	&	Labor Market	&	EU	\\
			IMPORT\_EU	&	Total trade: import value	&	6	&	Trade	&	EU	\\
			EXPORT\_EU	&	Total trade: export value	&	6	&	Trade	&	EU	\\
			EB1Y\_EU	&	Euribor 1 year	&	2	&	Interest \& Exchange Rates	&	EU	\\
			TB10Y\_EU	&	10-year government benchmark bond yield	&	2	&	Interest \& Exchange Rates	&	EU	\\
			EFFEXR\_EU	&	ECB nominal effective exchange rate againt group of trading partners 	&	2	&	Interest \& Exchange Rates	&	EU	\\
			EUROSTOXX50\_EU	&	Euro STOXX composite index	&	5	&	Stock Market	&	EU	\\
			IOP\_UK	&	Index of Production	&	5	&	Output \& Income	&	UK	\\
			CPI\_UK	&	CPI Index	&	5	&	Prices	&	UK	\\
			PPI\_UK	&	Output of manufactured products	&	5	&	Prices	&	UK	\\
			UNEMP\_UK	&	Unemployment rate: aged 16 and over	&	2	&	Labor Market	&	UK	\\
			EFFEXR\_UK	&	Effective exchange rate index, Sterling	&	2	&	Interest \& Exchange Rates	&	UK	\\
			TB10Y\_UK	&	10-year British government stock, nominal par yield	&	2	&	Interest \& Exchange Rates	&	UK	\\
			LIBOR6M\_UK	&	6 month interbank lending rate, month end	&	2	&	Interest \& Exchange Rates	&	UK	\\
			M3\_UK	&	Monetary aggregate M3	&	6	&	Money \& Credit	&	UK	\\
			CPI\_CN	&	CPI: all iterms	&	5	&	Prices	&	CN	\\
			PPI\_CN	&	Producer price index for industrial products (same month last year = 100)	&	2	&	Prices	&	CN	\\
			M2\_CN	&	Monetary aggregate M2	&	6	&	Money \& Credit	&	CN	\\
			EFFEXR\_CN	&	Real broad effective exchange rate	&	2	&	Interest \& Exchange Rates	&	CN	\\
			EXPORT\_CN	&	Value goods	&	6	&	Trade	&	CN	\\
			IMPORT\_CN	&	Value goods	&	6	&	Trade	&	CN	\\
			INDGR\_CN	&	Growth rate of industrial value added (last year = 100)	&	2	&	Output \& Income	&	CN	\\
			SHANGHAI\_CN	&	Shanghai Composite Index	&	5	&	Stock Market	&	CN	\\
			TB10Y\_JP	&	10-year government benchmark bond yield	&	2	&	Interest \& Exchange Rates	&	JP	\\
			EFFEXR\_JP	&	Real broad effective exchange rate	&	2	&	Interest \& Exchange Rates	&	JP	\\
			CPI\_JP	&	CPI Index: all items	&	5	&	Prices	&	JP	\\
			M2\_JP	&	Monetary aggregate M2	&	6	&	Money \& Credit	&	JP	\\
			UNEMP\_JP	&	Unemployment rate: aged 15-64	&	2	&	Labor Market	&	JP	\\
			IPI\_JP	&	Production of Total Industry 	&	5	&	Output \& Income	&	JP	\\
			IMPORT\_JP	&	Import price index: all commodities	&	6	&	Trade	&	JP	\\
			EXPORT\_JP	&	Value goods	&	6	&	Trade	&	JP	\\
			NIKKEI225\_JP	&	NIKKEI 225 composite index	&	5	&	Stock Market	&	JP	\\
			\hline
			\caption{List of macroeconomic variables in this study.}\label{table:macro} 
		\end{longtable}
		\noindent Data source: Fred St.Louis, ECB Statistical Data Warehouse, UK Office for National Statistics, Bank of England, National Bureau of Statistics of China, YAHOO!. tCode: 1: none; 2: $\Delta X_t$; 3: $\Delta^2 X_t$; 4: $\log X_t$; 5: $\Delta \log X_t$; 6: $\Delta^2 \log X_t$; 7: $\Delta(X_t/X_{t-1}-1)$.
	}%

\end{document}